\documentclass{article}

\usepackage{standalone}
\usepackage{etoolbox}

\usepackage[T1]{fontenc}

\usepackage{graphicx}
\graphicspath{{../graphics/}}

\usepackage{amssymb}
\usepackage{mathtools}

\usepackage{setspace}

\usepackage{placeins}

\usepackage[table]{xcolor}

\definecolor{linkcol}{rgb}{0,0,0.38}
\definecolor{citecol}{rgb}{0.8,0,0}
\definecolor{urlcol}{rgb}{0.1,0.35,0}

\usepackage[pdfpagelabels,bookmarks=false,hyperfootnotes=false]{hyperref}
\hypersetup{colorlinks, linkcolor=linkcol, citecolor=citecol, urlcolor=urlcol}

\DeclareFontFamily{U}{BOONDOX-calo}{\skewchar\font=45 }
\DeclareFontShape{U}{BOONDOX-calo}{m}{n}{
  <-> s*[1.05] BOONDOX-r-calo}{}
\DeclareFontShape{U}{BOONDOX-calo}{b}{n}{
  <-> s*[1.05] BOONDOX-b-calo}{}
\DeclareMathAlphabet{\link}{U}{BOONDOX-calo}{m}{n}
\DeclareMathAlphabet{\blink}{U}{BOONDOX-calo}{b}{n}

\usepackage[utf8]{inputenc}

\usepackage[
backend=biber,
style=alphabetic,
citestyle=alphabetic,
maxalphanames=4,
maxcitenames=99,
mincitenames=98,
maxbibnames=99,
giveninits=true,
]{biblatex}

\usepackage{lmodern}

\usepackage{bbm}

\usepackage[letterpaper,margin=1in]{geometry}
\usepackage{amsthm}
\usepackage{thm-restate}
\usepackage{amsmath}
\usepackage[nameinlink,capitalize]{cleveref}

 \newtheoremstyle{light} %
     {\topsep}                    %
     {\topsep}                    %
     {\itshape}                   %
     {}                           %
     {\scshape}                   %
     {.}                          %
     {.5em}                       %
     {}  %

 \newtheorem{theorem}{Theorem}[section]
 \newtheorem{lemma}[theorem]{Lemma}
 
 \newtheorem{proposition}[theorem]{Proposition}
 \newtheorem{definition}[theorem]{Definition}
 
 \newtheorem{corollary}[theorem]{Corollary}

 \newtheorem{observation}[theorem]{Observation}

 \newtheorem{claim}[theorem]{Claim}

\if@cref@capitalise
\crefname{claiminproof}{Claim}{Claims}
\else
\crefname{claiminproof}{claim}{claims}
\fi

\if@cref@capitalise
\crefname{observation}{Observation}{Observations}
\else
\crefname{observation}{observation}{observations}
\fi

\if@cref@capitalise
\crefname{algocf}{Algorithm}{Algorithms}
\else
\crefname{algocf}{algorithm}{algorithms}
\fi

\if@cref@capitalise
\crefname{conjecture}{Conjecture}{Conjectures}
\else
\crefname{conjecture}{conjecture}{conjectures}
\fi

\if@cref@capitalise
\crefname{thm}{Theorem}{Theorems}
\else
\crefname{thm}{theorem}{theorems}
\fi

\if@cref@capitalise
\crefname{lem}{Lemma}{Lemmas}
\else
\crefname{lem}{lemma}{lemmas}
\fi

\makeatother

\usepackage{url}
\usepackage{array}

\usepackage{setspace}

\usepackage{wrapfig}

\makeatletter
\newcommand{\labeltarget}[1]{\Hy@raisedlink{\hypertarget{#1}{}}}
\makeatother

\usepackage{upgreek}

\usepackage{complexity}

\usepackage[shortlabels,inline]{enumitem}
\usepackage{moreenum}

\usepackage[super]{nth}

\usepackage{bm}

\usepackage{xspace}

\usepackage{ifthen}

\usepackage[immediate]{silence}
\usepackage{sectsty}
\DeactivateWarningFilters[tmp]
\allsectionsfont{\boldmath}

\setlist[enumerate]{nosep,topsep=0.1em}
\setlist[enumerate,1]{label=(\roman*), leftmargin=2.2em}
\setlist[itemize]{nosep,topsep=0.3em}

\usepackage[textsize=footnotesize, color=blue!30!white]{todonotes}
\DeclareMathOperator{\child}{child}

\usepackage{xfrac}

\usepackage{tikz}
\usetikzlibrary{calc}
\usetikzlibrary{math}
\usetikzlibrary{shapes.geometric}
\usetikzlibrary{arrows}
\usetikzlibrary{decorations.pathreplacing, decorations.pathmorphing}

\usepackage{tcolorbox}
\tcbuselibrary{skins}

\usepackage{float}
\makeatletter
\newcommand\appendtographicspath[1]{%
  \g@addto@macro\Ginput@path{#1}%
}
\makeatother

\usepackage[margin=10pt,font=small,labelfont=bf,skip=-5pt]{caption}
\usepackage{subcaption}

\usepackage[linesnumbered,vlined,ruled,algo2e]{algorithm2e}
\SetAlgoSkip{bigskip}
\SetAlgoInsideSkip{smallskip}

\usepackage{mdframed}

\let\truehypersetup\hypersetup
\renewcommand\hypersetup[1]{}
\usepackage{bigfoot}
\let\hypersetup\truehypersetup
\DeclareRobustCommand{\cvs}{\ensuremath{\:/_{\!k}\:}}

\DeclareRobustCommand{\OPT}{\ensuremath{\mathrm{OPT}}}

\renewcommand{\epsilon}{\varepsilon}

\renewcommand{\P}{\text{P}}
\renewcommand{\NP}{\text{NP}}
\renewcommand{\epsilon}{\varepsilon}

\def\cupp{\stackrel{.}{\cup}}

\makeatletter
\let\@@pmod\pmod
\DeclareRobustCommand{\pmod}{\@ifstar\@pmods\@@pmod}
\def\@pmods#1{\mkern8mu({\operator@font mod}\mkern 6mu#1)}
\makeatother

\makeatletter
\let\@@mod\mod
\DeclareRobustCommand{\mod}{\@ifstar\@mods\@@mod}
\def\@mods#1{\mkern8mu{\operator@font mod}\mkern 6mu#1}
\makeatother

\definecolor{green}{rgb}{0.4,0.85,0.6}

 \author{
 Meike Neuwohner\thanks{
CNRS \& DIENS, ENS Paris - PSL Research University, Paris, France. \href{mailto:meike.neuwohner@ens.fr}%
{meike.neuwohner@ens.fr}.\newline
Part of the research was done while Meike Neuwohner was affiliated with the London School of Economics and Political Science. During this time, Meike Neuwohner was supported by EPSRC grant EP/X030989/1.
 }
 \and
 Vera Traub\thanks{
 Department of Computer Science, ETH Zurich, Zurich, Switzerland.
 Email: \href{mailto:vtraub@ethz.ch}%
 {vtraub@ethz.ch}.
 }
 \and
 Rico Zenklusen\thanks{
 Department of Mathematics, ETH Zurich, Zurich, Switzerland.
 Email: \href{mailto:ricoz@ethz.ch}%
 {ricoz@ethz.ch}.}
 }

\title{Approximation Schemes for Planar Graph Connectivity Problems}

\date{}

\begin{document}

\maketitle

\begin{abstract}
Finding a smallest subgraph that is $k$-edge-connected, or augmenting a $k$-edge-connected graph with a smallest subset of given candidate edges to become $(k+1)$-edge-connected, are among the most fundamental Network Design problems.
They are both APX-hard in general graphs.
However, this hardness does not carry over to the planar setting, which is not well understood, except for very small values of $k$.
One main obstacle in using standard decomposition techniques for planar graphs, like Baker's technique and extensions thereof, is that connectivity requirements are global (rather than local) properties that are not captured by existing frameworks.

We present a novel, and arguably clean, decomposition technique for such classical connectivity problems on planar graphs. This technique immediately implies PTASs for the problems of finding a smallest $k$-edge-connected or $k$-vertex-connected spanning subgraph of a planar graph for arbitrary $k$.
By leveraging structural results for minimally $k$-edge-connected graphs, we further obtain a PTAS for planar $k$-connectivity augmentation for any constant $k$. 
We complement this with an \NP-hardness result, showing that our results are essentially optimal. 
\end{abstract}

\section{Introduction}

The field of Network Design in the context of approximation algorithms has been largely shaped by basic questions on how to create graphs that fulfill certain connectivity requirements in the most economical way.
Some of the most fundamental such questions are the $k$-Edge/Vertex-Connected Spanning Subgraph problems and augmentation problems.
They aim at creating a graph with a prescribed connectivity.
The $k$-Edge/Vertex-Connected Spanning Subgraph problems ask about creating such a graph from scratch, whereas augmentation problems start with an existing graph and enlarge it by adding edges.

More formally, let $k\in \mathbb{Z}_{\geq 1}$. In the $k$-Edge-Connected Spanning Subgraph problem ($k$-ECSS), we are given a graph $G=(V,E)$, and the goal is to find a spanning subgraph of $G$ that is $k$-edge-connected with as few edges as possible.
The $k$-Vertex-Connected Spanning Subgraph problem ($k$-VCSS) requires $k$-vertex-connectivity instead.
In the $k$-Connectivity Augmentation problem ($k$-CAP), we are given a $k$-edge-connected graph $G=(V,E)$ together with a set $L\subseteq \binom{V}{2}$ of candidate edges to add, also called \emph{links}, and the task is to find a subset $F\subseteq L$ of minimum size so that $(V,E\dot{\cup} F)$ is $(k+1)$-edge-connected.%
\footnote{Note that $k$-ECSS seeks a $k$-edge-connected graph whereas $k$-CAP aims for $(k+1)$-edge-connectivity.
 Despite this $k$ vs.~$k+1$ discrepancy, we stick to these definitions of $k$-ECSS and $k$-CAP as they are common in the literature.
}
In the weighted counterparts of these problems, $k$-WECSS, $k$-WVCSS, and $k$-WCAP, respectively, edges/links have non-negative costs and the goal is to minimize the total cost of the selected edges/links.

The study of these and related problems, on some of which we expand later, has a long history and led to a rich and steadily growing body of results and new techniques.
We refer the interested reader to \cite{goemansGeneralApproximationTechnique1995,jainFactor2Approximation2001,vaziraniApproximationAlgorithms2001,williamsonDesignApproximationAlgorithms2011,traubBetterThan2ApproximationsWeighted2025} and references therein for further information.

The problems $k$-WECSS, $k$-WVCSS, and $k$-WCAP are well-known to be hard to approximate, in particular $\APX$-hard, in essentially all natural variations (see \cite{fernandesBetterApproximationRatio1998,czumajApproximabilityMinimumcostKconnected1999,kortsarzHardnessApproximationVertexConnectivity2004,pritchardKEdgeConnectivityApproximationLP2011}).
Therefore, the gold standard for approximation algorithms in general graphs is achieving strong constant-factor guarantees.

Another approach to these basic problems is to understand whether better approximation guarantees can be obtained in certain relevant graph classes.
One particularly interesting, and arguably among the most studied, classes in this context, are planar graphs, which is also the focus of this paper.
Contrary to general graphs, $k$-ECSS and $k$-VCSS, as well as further variations, are only known to be \NP-hard (instead of \APX-hard) in planar graphs~\cite{gareyComputersIntractabilityGuide1990}, and no hardness result for $k$-CAP in planar graphs was known prior to this work.
An additional reason to hope for stronger results in planar graphs is that many graph optimization problems that are \NP-hard in general graphs, admit polynomial time approximation schemes (PTAS) in planar graphs.

Indeed, important progress has been made on PTASs for these problems in planar graphs.
Unfortunately, this progress has largely been limited to small connectivity requirements.
The most relevant prior results in this context for our work are the following.
\textcite{provanTwoConnectedAugmentationProblems1999} showed that planar $1$-WCAP can be solved exactly in polynomial time.
Moreover, \textcite{czumajApproximationSchemesMinimum2004} presented PTASs for planar $2$-ECSS and planar $2$-VCSS.\footnote{They also presented PTASs for the weighted versions when the ratio of the total edge cost to the optimum solution cost is bounded.} The first result was later strengthened by \textcite{bergerMinimumWeight2EdgeConnected2007}, who presented a linear-time PTAS for planar $2$-ECSS and a PTAS for planar $2$-WECSS.
Both \textcite{czumajApproximationSchemesMinimum2004} and \textcite{bergerMinimumWeight2EdgeConnected2007} proceed by repeatedly buying cheap cycles, which allows for splitting the connectivity requirement into a requirement inside the cycle and one outside the cycle.
This technique is specific to $2$-edge/vertex-connectivity, and it is unknown how to generalize it to higher connectivities.
More recently, \textcite{zhengLineartimeApproximationSchemes2017} showed that also both planar $3$-ECSS and planar $3$-VCSS admit a PTAS.
While some of the used techniques have parallels to ours, namely the creation of overlapping layers, the analysis remains tailored to connectivity $3$, and is based on properties that fail for higher connectivities.

One of the arguably most generic frameworks to obtain PTASs for planar graph problems, which also inspired our work, grew out of Baker's technique~\cite{bakerApproximationAlgorithmsNPcomplete1994} and extensions thereof, for example by \textcite{kleinLineartimeApproximationScheme2005}.
These approaches have been hugely successful, leading to PTASs for a large class of problems on planar graphs, including maximum independent set, minimum vertex cover, minimum dominating set, a linear time PTAS for the Traveling Salesman Problem (TSP), and many others.
For optimization problems where edges have to be picked, the approaches can be divided into four main steps (where not all steps are always needed):
\begin{enumerate}[label=(\roman*)]\label{enum:baker_steps}
 \item \emph{Sparsifier step:} Delete some edges of the input graph while approximately preserving the optimal value.
 \item \emph{Decomposition step:} Decompose the resulting sparsified graph into independent subproblems of bounded treewidth, by fixing/buying a cheap substructure of the solution.
 \item \emph{Dynamic programming step:} Solve each subproblem optimally using dynamic programming.
 \item \emph{Combining step:} Combine the solutions of the subproblems to obtain a solution for the whole graph.
\end{enumerate}

The decomposition step typically considers the planar dual to partition the graph into layers of constant width, defined by their distance from the outer face.
Such layers are known to have constant treewidth.
To achieve this decomposition, one buys a small constant fraction of all edges.
To make sure that this constant fraction is small in cost, the sparsifier step is often necessary, whose goal is to make sure that the cost of all edges is at most a constant factor higher than the cost of an optimal solution.

Unfortunately, this approach is known to require a certain \emph{locality} of the considered problem, which is not satisfied by many problems with more global requirements.
This need for locality has been repeatedly highlighted as a major obstacle of Baker's technique (see, for example,~\cite{demaineApproximationSchemesPlanar2008,kleinLineartimeApproximationScheme2005,zhengLineartimeApproximationSchemes2017,fox-epsteinEmbeddingPlanarGraphs2019}), in particular when more global connectivity requirements are needed, like in $k$-ECSS.
There are some notable exceptions, which, however, do not apply to our setting of higher edge-connectivities.
One is the framework of bidimensionality introduced by \textcite{demaineBidimensionalityNewConnections2005}, which leads to a PTAS for connected dominating set, through a technique that can be interpreted as a non-trivial extension of Baker's approach.
Another one is the ubiquity framework of \textcite{cohen-addadApproximatingConnectivityDomination2016} and an extension thereof by \textcite{fominApproximationSchemesWidth2019}, which, if certain conditions are met by the problem under consideration, leads to a PTAS for planar graphs.
As mentioned, a Baker-type approach has been used to obtain a linear time PTAS for planar TSP~\cite{kleinLineartimeApproximationScheme2005}, improving on prior PTASs~\cite{grigniApproximationSchemePlanar1995,aroraPolynomialtimeApproximationScheme1998} not relying on Baker's technique.
\citeauthor{kleinLineartimeApproximationScheme2005} avoids the use of higher connectivities by phrasing TSP as the problem of finding a (singly) connected graph with only even degrees.
In summary, even $2$-edge-connectivity falls outside the classical realm of Baker's framework (and is also neither captured by the bidimensionality \cite{demaineBidimensionalityNewConnections2005} nor the ubiquity framework \cite{cohen-addadApproximatingConnectivityDomination2016,fominApproximationSchemesWidth2019}). We elaborate on this in \cref{sec:bidimensionality_and_ubiquity}.

The goal of this paper is to provide a novel and arguably clean way to use a Baker-type technique that, for the first time, can deal with larger connectivity requirements, thus addressing the locality issue of Baker's approach for these settings.

\subsection{Our results}

We start with our main result for planar connectivity augmentation, which works also for weighted instances as long as the \emph{cost ratio}, the ratio between the maximum and minimum cost of any link, is bounded by a constant.
\begin{theorem}\label{theorem:PTAS_WCAP}
 For any fixed $k \in \mathbb{Z}_{\geq 1}$, there is a PTAS for planar $k$-WCAP for bounded cost ratio instances. 
\end{theorem}
Prior to our work, an exact algorithm was known for planar $1$-WCAP~\cite{provanTwoConnectedAugmentationProblems1999}.
For larger $k$, we are unaware of any planar $k$-WCAP algorithms that improve over the best-known algorithms for general graphs~\cite{%
traub15varepsilonApproximationAlgorithmWeighted2023,
traubBetterThan2ApproximationsWeighted2025,
},
which only give constant-factor guarantees.

We complement the above with the following hardness result for planar $k$-CAP, which we derive through a reduction from planar $3$-SAT.
\begin{theorem}\label{thm:hardness_k-CAP}
For any $k\geq 2$, planar $k$-CAP is \NP-hard.
\end{theorem}
Prior to our work, it was unknown whether planar $k$-CAP is \NP-hard for any $k\geq 2$.
The presented PTAS together with the above hardness result therefore essentially settles the approximability of planar $k$-CAP.

Interestingly, trying to apply a classical Baker-type approach to (even unweighted) planar $k$-CAP, as outlined on page \pageref{enum:baker_steps}, faces two important issues.
First, it is highly unclear how to sparsify a planar $k$-CAP instance.
This is a major obstacle, because an optimal $k$-CAP solution may cost a negligible amount compared to the cost of all links. (See discussion and example in \Cref{sec:augmentation_short}.)
Second, as previously mentioned, prior decomposition techniques do not apply to higher connectivities.
This is why also for the (unweighted) $k$-ECSS problem, where sparsification turns out to be easy, a PTAS was only known for $k\leq 3$.
We resolve these issues through a combination of structural insights and a new decomposition technique.

Our new decomposition also readily applies to planar $k$-ECSS and its vertex-connectivity counterpart, planar $k$-VCSS, leading to the following result.

\begin{theorem}\label{thm:PTASs_k-WECSS_k-WVCSS_constant_k}
For any fixed $k \in \mathbb{Z}_{\geq 1}$, planar $k$-WECSS and planar $k$-WVCSS, both with bounded cost ratios, admit a linear-time PTAS.
\end{theorem}

Note that a planar graph can be at most $5$-vertex-connected because it contains a vertex with at most $5$ neighbors.\footnote{And there are planar graphs that are actually $5$-vertex-connected, like the icosahedral graph.
This is the graph formed by the vertices and edges of an icosahedron, one of the $5$ platonic solids.}
Hence, a PTAS for planar $k$-VCSS is only relevant for $k\leq 5$.

Combining the above result with known $k$-ECSS techniques that are strong for large $k$, we can lift the restriction of $k$ being a fixed constant at the expense of losing linear-time efficiency.

\begin{corollary}\label{cor:PTASs_k-WECSS}
Planar $k$-WECSS admits a PTAS when cost ratios are bounded and $k$ is part of the input.
\end{corollary}

Prior to our work, PTASs for planar $k$-ECSS and planar $k$-VCSS were only known for $k\leq 3$ \cite{czumajApproximationSchemesMinimum2004,bergerMinimumWeight2EdgeConnected2007,zhengLineartimeApproximationSchemes2017}.
(Where \cite{bergerMinimumWeight2EdgeConnected2007} presented a PTAS for planar $2$-WECSS, without restrictions on the cost ratio.)

\subsection{Further related results}\label{sec:further-related-results}

Augmentation problems have been heavily studied in general graphs and many special graph classes (see \cite{%
fredericksonApproximationAlgorithmsSeveral1981,%
khullerApproximationAlgorithmsGraph1993,%
nagamochiApproximationFindingSmallest2003,%
even18ApproximationAlgorithm2009,%
cohen1Ln2Approximation2013,%
kortsarzSimplified15ApproximationAlgorithm2016,%
nutovTreeAugmentationProblem2017,%
cheriyanApproximatingUnweightedTree2018a,%
cheriyanApproximatingUnweightedTree2018,%
adjiashviliBeatingApproximationFactor2018,%
fioriniApproximatingWeightedTree2018,%
grandoniImprovedApproximationTree2018,%
kortsarzLPRelaxationsTreeAugmentation2018,%
byrkaBreaching2ApproximationBarrier2020,%
cecchettoBridgingGapTree2021,%
nutovApproximationAlgorithmsConnectivity2021,%
traub15varepsilonApproximationAlgorithmWeighted2023,%
cecchettoBetterthan43approximationsLeaftoleafTree2024,%
traubBetterThan2ApproximationsWeighted2025,%
} and references therein), including planar graphs~\cite{gutwengerSubgraphInducedPlanar2003,rutterAugmentingConnectivityPlanar2008}.
The (unweighted) augmentation variant for planar graphs studied in \cite{gutwengerSubgraphInducedPlanar2003} is a variation of $k$-CAP.
In particular, they fix a subset of the vertices $W \subseteq V$ and allow edges to be added between \emph{any} pair of vertices in $W$ as long as planarity of the whole graph is preserved.
Their goal is to make the subgraph induced by $W$ (singly) connected.
Similarly, \textcite{rutterAugmentingConnectivityPlanar2008} show \NP-hardness for making a planar graph $2$-edge-connected if one is allowed to add edges between \emph{any} pair of vertices, as long as the resulting graph remains planar.

In general (not necessarily planar) graphs, $k$-ECSS, for any $k\in \mathbb{Z}_{\geq 2}$, is known not to admit a PTAS unless $\P=\NP$ \cite{fernandesBetterApproximationRatio1998,czumajApproximabilityMinimumcostKconnected1999,pritchardKEdgeConnectivityApproximationLP2011}.
For the weighted version, $k$-WECSS, the best constant-factor approximation algorithm is an over four decades old $2$-approximation \cite{fredericksonApproximationAlgorithmsSeveral1981,fredericksonRelationshipBiconnectivityAugmentation1982}, which can also be matched with more recent techniques, like primal-dual \cite{goemansGeneralApproximationTechnique1995} or iterative rounding \cite{jainFactor2Approximation2001}.
Progress beyond the factor $2$ has been achieved for (the unweighted) $2$-ECSS, where successive improvements~\cite{%
khullerBiconnectivityApproximationsGraph1994,%
cheriyanImproving15ApproximationSmallest2001,%
hunkenschroder43ApproximationAlgorithm2019,%
seboShorterToursNicer2014,%
gargImprovedApproximationTwoEdgeConnectivity2023,%
kobayashiApproximationAlgorithmTwoEdgeConnected2023,%
bosch-calvo54ApproximationTwoEdge2025,%
} led to the currently best approximation factor of slightly below $\sfrac{5}{4}$~\cite{hommelsheimBetterThan$54$ApproximationTwoEdge2026}.
For arbitrary $k\geq 2$, it is known how to get $(1+\mathcal{O}(\sfrac{1}{k}))$-approximation algorithms~\cite{cheriyanApproximatingMinimumSizeKConnected2000,gabowApproximatingSmallestIt2009,gabowIteratedRoundingAlgorithms2012}.

We note that the difficulties of obtaining results for larger connectivities faced here are not shared by some other variants of the problem, for example when multiple copies of the same edge can be bought, which is known as the (weighted) $k$-edge-connected spanning multi-subgraph problem ($k$-(W)ECSM).
Indeed, one can for example extend the idea to contract cheap edge sets to simplify the instance.
More precisely, if one can identify an edge set $U$ of cost at most $\mathcal{O}(\frac{\varepsilon}{k} c(\OPT))$, whose contraction simplifies the instance, then one can buy $k$ many times each edge of $U$ and contract $U$.
Exploiting this, \textcite{demaineApproximationAlgorithmsContraction2010} presented a PTAS for planar $k$-WECSM, for any fixed $k$.
Together with known $k$-WECSM results for the non-planar case, like the $(1+\mathcal{O}(\sfrac{1}{k}))$-approximation algorithm of~\cite{hershkowitzGhostValueAugmentation2024}, this implies a PTAS for planar $k$-WECSM even when $k$ is part of the input.%
\footnote{
 Together with the above-mentioned $(1+\mathcal{O}(\sfrac{1}{k}))$-approximation for $k$-WECSM, this leads to a PTAS for planar $k$-WECSM (with $k$ part of the input) as follows.
 If $\varepsilon$ is larger than the $\mathcal{O}(\sfrac{1}{k})$ term, then one can simply use the $(1+\mathcal{O}(\sfrac{1}{k}))$-approximation algorithm for planar $k$-WECSM; otherwise, $k$ is a constant and one can use the PTAS for planar $k$-WECSM by \textcite{demaineApproximationAlgorithmsContraction2010}.
}

Moreover, we want to highlight that there has been a very fruitful line of work on further survivable network design problems and related problems on planar graphs (and graphs of bounded genus), including Steiner Tree and generalizations thereof (we refer the interested reader to
\cite{%
borradailePolynomialtimeApproximationScheme2007,%
borradaileLogApproximationScheme2009,%
bateniApproximationSchemesSteiner2011,%
borradailePolynomialTimeApproximationSchemes2014,%
borradaileTwoEdgeConnectivitySurvivableNetwork2016,%
borradailePTASThreeEdgeConnectedSurvivable2017,%
kleinCorrelationClusteringTwoEdgeConnected2023,%
}
and references therein).

Finally, we note that the generality of Baker's technique has led to an interesting attempt to provide a syntactic characterization of problems that admit a PTAS on planar graphs through components based on Baker's technique~\cite{khannaSyntacticCharacterizationPTAS1996}.

\subsection{Organization of the paper}

In \Cref{sec:overview}, we provide an overview of our techniques.
In particular, \Cref{sec:summary_k_safe_covers} provides an overview of our new decomposition technique, which we call \emph{$k$-edge-safe cover}, using planar $k$-ECSS as an example.
This allows us to use Baker-type techniques for higher connectivities.
Moreover, \Cref{sec:augmentation_short} discusses how further structural insights together with this technique can be used to obtain our augmentation result, \Cref{theorem:PTAS_WCAP}.

\Cref{sec:k_safe_covers} contains missing details of the $k$-safe-cover decomposition.
\Cref{sec:WECSS_and_WVCSS} provides formal proofs of \Cref{thm:PTASs_k-WECSS_k-WVCSS_constant_k,cor:PTASs_k-WECSS}.
Missing proofs for our augmentation result are contained in \Cref{sec:augmentation}.
\Cref{sec:hardness} proves our hardness result for planar $k$-CAP, \Cref{thm:hardness_k-CAP}.
Finally, \Cref{sec:bidimensionality_and_ubiquity} discusses why both the bidimensionality and ubiquity frameworks do not apply to planar problems requiring higher connectivities.

\section{Overview of techniques}\label{sec:overview}

In this section, we provide an overview of our techniques.
Most proofs are deferred to later sections for better readability.

\subsection{Instance decomposition via $k$-edge-safe covers}\label{sec:summary_k_safe_covers}

We now exemplify our decomposition technique for connectivity problems on planar graphs, which can be thought of as an extension of Baker's framework.
For the sake of simplicity of presentation, we focus on planar $k$-WECSS in this overview.
The same approach can be used for $k$-WVCSS (see \Cref{sec:WECSS_and_WVCSS}) and we will use it, together with additional techniques, to obtain our main result on planar connectivity augmentation problems.

Let $k\in \mathbb{Z}_{\geq 1}$, and let $G=(V,E)$ be a planar graph with positive edge costs $c\colon E \to \mathbb{R}_{> 0}$.
Moreover, we denote by $\Delta \coloneqq  \max_{e\in E}c(e)/\min_{e\in E}c(e)$ the max-to-min cost ratio of the instance.

As in the classical framework of Baker, we aim at decomposing the problem into independent instances of bounded treewidth, which are solvable in linear time by a result of \textcite{chalermsook_et_al:LIPIcs.APPROX-RANDOM.2018.8}.
\begin{theorem}[see~\cite{chalermsook_et_al:LIPIcs.APPROX-RANDOM.2018.8}, Theorem 3]\label{thm:bounded_treewidth_kWECSS}
For any fixed $k\in \mathbb{Z}_{\geq 1}$, planar $k$-WECSS and planar $k$-WVCSS can be solved in linear time in graphs of bounded treewidth.
\end{theorem}
Continuing with classical Baker, we will buy a small set of edges $E_{\mathcal{U}}$ (the choice of notation will become clear soon), so that the problem decomposes into independent instances of small treewidth. Interestingly, despite the fact that they can be solved independently of each other, they are actually not disjoint in terms of edges/vertices; this is in contrast to classical applications of Baker's framework.
Another key difference to prior approaches is that the edges we buy do not have the desired property themselves, i.e., the connected components of $E_{\mathcal{U}}$ are not necessarily $k$-edge-connected.
Instead, we take a different perspective.
Recall that a graph is $k$-edge-connected if every cut has at least $k$ crossing edges.
We want to buy an edge set $E_{\mathcal{U}}$ and decompose the graph into (overlapping) parts, so that every cut that is not fully contained in one of the parts has at least $k$ edges in $E_{\mathcal{U}}$ that cross it.
Hence, it remains to take care of cuts fully contained in one of the parts, which means that it suffices to $k$-edge-connect the parts, which decomposes the problem.
We now formalize this idea.

We start by observing that, to check $k$-edge-connectivity of a connected graph, it suffices to check that every \emph{connected cut} is crossed by at least $k$ edges.
Connected cuts, also sometimes called \emph{simple cuts}, are defined as follows.
\begin{definition}[connected cut]
Let $G=(V,E)$ be a graph.
A nonempty vertex set $S\subsetneq V$ is a connected cut (in $G$) if both induced subgraphs $G[S]$ and $G[V\setminus S]$ are connected graphs.
\end{definition}

A connected graph $G$ is $k$-edge-connected if and only if every connected cut has at least $k$ crossing edges.\footnote{
To see this, note that in a connected graph $G=(V,E)$, an arbitrary cut, which is a nonempty vertex set $S\subsetneq V$, is connected if and only if it is minimal in the following sense: there is no other cut $\overline{S}$ with $\delta(\overline{S}) \subsetneq \delta(S)$, where $\delta(W)$, for $W\subseteq V$, contains all edges with precisely one endpoint in $W$.
} Based on this observation, we introduce the notion of a \emph{$k$-edge-safe cover}, which allows us to decompose a $k$-ECSS instance.
\begin{definition}\label{def:k-edge-safe_cover}
Let $G=(V,E)$ be a graph and let $\mathcal{U}=(U_i)_{i\in I}$ be a collection of nonempty subsets of $V$ that cover $V$.
Let $V_{\mathcal{U}}$ be the set of vertices that appear in at least $2$ distinct $U_i$'s and let $E_\mathcal{U}$ be the set of edges incident to a vertex in $V_{\mathcal{U}}$. 

We call $\mathcal{U}$ a \emph{$k$-edge-safe cover of $G$} if for every connected cut $S$, we have
\begin{enumerate}[(i)]
    \item \label{def:k_edge_safe_1}
     $|\delta(S) \cap E_\mathcal{U}|\ge k$, or 
     \item \label{def:k_edge_safe_2}
     $\delta(S)\subseteq E[U_i]$ for some $i\in I$.   
\end{enumerate}
\end{definition}

The definition of a $k$-edge-safe cover aims at providing the conditions we need to decompose the problem after buying the edges $E_{\mathcal{U}}$.
In the lemma below, the expression $(V,F_i)/E[V\setminus U_i]$ denotes the graph obtained from $(V,F_i)$ by contracting all edges in $E[V\setminus U_i]$, i.e., we contract all edges with both endpoints in $V \setminus U_i$.\footnote{
An edge contraction merges the endpoints of the edge into a single vertex and deletes resulting loops.
}

\begin{restatable}{lemma}{GluingSolutions}\label{lemma:glue_solutions_together}
Let $G=(V,E)$ be a connected graph, and let $\mathcal{U}=(U_i)_{i\in I}$ be a $k$-edge-safe cover of~$G$.
Let $(F_i)_{i\in I}$ be edge sets such that $(V,F_i)/E[V\setminus U_i]$ is $k$-edge-connected. 
Let $F\coloneqq E_\mathcal{U}\cup \bigcup_{i\in I} (F_i\cap E[U_i])$. Then $(V,F)$ is $k$-edge-connected.
\end{restatable}

See \cref{sec:WECSS_and_WVCSS} for a proof.
\Cref{lemma:glue_solutions_together} shows that, after we buy the edges in $E_{\mathcal{U}}$, it indeed suffices to find $k$-edge-connected subgraphs independently in each graph $(V,F_i)/E[V\setminus U_i]$.
For such a solution to be competitive with respect to an optimum solution in the original graph $(V,E)$, we need to ensure that both 
\begin{enumerate*}
    \item $E_{\mathcal{U}}$ is cheap, and 
    \item after having bought $E_{\mathcal{U}}$, the cost of optimum $k$-ECSS solutions in $(V,E)/E[V\setminus U_i]$ compares to the cost of an optimum solution in $(V,E)$.
\end{enumerate*}
We start by showing that the second condition is satisfied by any $k$-edge-safe cover.
(The first condition will be fulfilled by choosing the $k$-edge-safe cover appropriately, which we discuss afterwards.)

\begin{restatable}{lemma}{CostCombinedSolution} \label{lemma:cost_combined_solution}
Let $G=(V,E)$ be a $k$-edge-connected graph with edge costs $c\colon E \to \mathbb{R}_{\geq 0}$, and let $\mathcal{U}=(U_i)_{i\in I}$ be a $k$-edge-safe cover of $G$.
For $i\in I$, let $G_i\coloneqq (V,E)/E[V\setminus U_i]$ and obtain $c_i$ from $c$ by setting the costs of all edges in $E_{\mathcal{U}}$ and in $E(G_i)\setminus E[U_i]$ to zero. Let $\OPT_i$ be an optimum $k$-WECSS solution for $(G_i,c_i)$.

Then, $F\coloneqq E_\mathcal{U}\cup \bigcup_{i\in I} (\OPT_i\cap E[U_i])$ is a $k$-WECSS solution in $G$ of cost 
\[
c(F) \leq c(\OPT) + c(E_{\mathcal{U}}),
\]
where $\OPT$ is an optimum $k$-WECSS solution in $(G,c)$.
\end{restatable}
\begin{proof}[Sketch of proof]
By \Cref{lemma:glue_solutions_together}, $F$ is a feasible solution to $k$-WECSS in $G$. 
Moreover, for each $i\in I$,
$T_i\ \coloneqq\ \OPT \setminus E[V\setminus U_i]$ is a feasible solution to $k$-WECSS in $(G_i,c_i)$, implying $c_i(\OPT_i)\le c_i(T_i)=c((\OPT\cap E[U_i])\setminus E_{\mathcal{U}})$.
Thus, using $E[U_i]\cap E[U_j]=E[U_i\cap U_j]\subseteq E_{\mathcal{U}}$ for every $i,j\in I$ with $i\neq j$, we get
\begin{align*}
    c(F)& \le c(E_{\mathcal{U}}) + \sum_{i\in I} c((\OPT_i\cap E[U_i])\setminus E_{\mathcal{U}}) \le c(E_\mathcal{U}) + c(\OPT).
\end{align*}

\end{proof}

Hence, decomposing the problem using $k$-edge-safe covers allows us to find a solution that is competitive with respect to an optimum solution in the original graph, as long as we ensure that $c(E_{\mathcal{U}})$ is small.
Moreover, to be able to solve the resulting instances efficiently, we want the graphs $G_i$ to have small treewidth.
The following statement shows that such $k$-edge-safe covers do exist, even without requiring $G$ to be $k$-edge-connected.\footnote{However, if the graph $G$ is not $k$-edge-connected, then we won't be able to find a feasible $k$-ECSS solution in each of the $G_i$.}
In terms of the cost of $E_{\mathcal{U}}$, the statement only shows that it can be made arbitrarily small compared to $c(E)$, which we compare afterwards against $c(\OPT)$.

\begin{theorem}\label{theorem:compute_edge_safe_cover_overview}
Let $\epsilon \in (0,1)$. Given a planar connected graph $G=(V,E)$ with edge costs $c\colon E \to \mathbb{R}_{\geq 0}$, we can, in time $\mathcal{O}(|V|+|E|)$, compute a $k$-edge-safe cover $\mathcal{U}=(U_i)_{i\in I}$ of $G$, together with the graphs $G_i\coloneqq (V,E)/E[V\setminus U_i]$ for $i\in I$, such that
\begin{enumerate}[(i)]
    \item \label{k-edge-safe_cover_costs} $c(E_{\mathcal{U}}) \le \epsilon \cdot c(E)$, and
    \item \label{k-edge-safe_cover_treewidth} $G_i$ has treewidth at most $\frac{26k}{\epsilon}$ for each $i\in I$.
\end{enumerate}
\end{theorem}

Our construction of a $k$-edge-safe cover to obtain \Cref{theorem:compute_edge_safe_cover_overview} can be interpreted as an extension of Baker's framework.
Loosely speaking, each $U_i$ can be thought of as a ring of vertices whose distances in the planar dual graph of $G$ to a fixed root is within some bounded range.
In order to be able to extend our results to the vertex-connectivity setting, we use the following slightly modified version of the dual graph: the \emph{vertex-face} graph $D_G$ for a fixed planar embedding of $G=(V,E)$ has vertex set $V\cup F$, where $F$ is the set of faces, and there is an edge between $v\in V$ and $f\in F$ if vertex $v$ is on face $f$.
Then we fix a root $r\in V$ and partition $V$ into rings $R_j$ (for $j=1,2,\dots$) of the form
$
R_j = \{ v\in V \colon \mathrm{dist}_{D_G}(r,v) \in [2\alpha_j,2\alpha_{j+1})\},
$
where the integers $\alpha_j$ are chosen such that the \emph{width} $\alpha_{j+1} -\alpha_j$ of each ring $R_j$ is $\mathcal{O}(k)$.
Each $U_i$ is then defined as the union of $\mathcal{O}(\frac{1}{\epsilon})$ many consecutive rings $R_j$, where $U_i$ and $U_{i+1}$ are overlapping in one ring $R_j$.
We provide details about this construction in \Cref{sec:k_safe_covers}.

To put the pieces together, it remains to show that $c(E) = \mathcal{O}(c(\OPT))$.
In general, for $k$-WECSS, even with bounded cost ratio, this is not true, as there could be many parallel edges that increase the cost of $E$ without contributing to $\OPT$.
However, we can easily assume without loss of generality that a $k$-WECSS instance has at most $k$ parallel edges between any pair of vertices.
Under this additional assumption, we have $c(E) \leq 6\cdot \Delta\cdot c(\OPT)$ because any simple planar graph on $n$ vertices has
fewer than $3n$ edges and $\OPT$ must have at least $\frac{n \cdot k}{2}$ edges.

Combining these results yields our PTAS for planar $k$-WECSS with fixed $k$ and bounded cost ratio.
For details, see \cref{sec:WECSS_and_WVCSS}.

\subsection{A PTAS for planar $k$-WCAP}\label{sec:augmentation_short}

As $k$-WCAP constitutes a special case of $(k+1)$-WECSS, it appears tempting to follow the same strategy as in the previous section: Given an instance $(G=(V,E),L,c)$ of planar $k$-WCAP, first apply \cref{theorem:compute_edge_safe_cover_overview} to obtain a $(k+1)$-edge-safe cover $\mathcal{U}=(U_i)_{i\in I}$ for $G+L\coloneqq(V,E\dot{\cup}L)$ such that the total cost of $L_{\mathcal{U}}$, the set of links incident to $V_{\mathcal{U}}$, is bounded by $\epsilon\cdot c(L)$. Next, for $i\in I$, solve the $k$-WCAP optimally in the subinstance $(G_i,L_i,c_i)$ of bounded treewidth that arises from $(G,L,c)$ by contracting $(E\dot{\cup} L)[V\setminus U_i]$. Finally, combine the obtained solutions with the set $L_{\mathcal{U}}$ to a solution to the whole instance.

 Unfortunately, this approach does not result in a PTAS, even when all links have unit costs. The reason for this is that $c(L)$, the total cost of the link set, can be much larger than the cost of an optimum solution, rendering it too expensive to purchase all of $L_{\mathcal{U}}$ for our solution.\footnote{Note that this does not contradict the results in the previous section: An instance of $k$-CAP corresponds to an instance of $(k+1)$-WECSS with $\{0,1\}$-costs; such an instance does not have a bounded cost ratio.} To see this, consider the instance of planar $3$-CAP shown in \cref{fig:snug_chain_1}. Adding a single link suffices to render the graph $4$-edge-connected; however, the total number of links may be linear in the number of vertices.
\begin{figure}[t]
	\begin{center}
		\begin{tikzpicture}[xscale=0.9, mynode/.style={circle, draw, fill, inner sep = 0pt, minimum size = 2mm}]
			\draw[thick, red, dashed] (0,0) to [out = -20, in = 180](6,-1);
			\draw[thick, red, dashed] (6,-1) to [out = 0, in = -160](12,0);
			\foreach \i in {0,...,5}
			{
				\node[mynode] (n\i) at (\i,0){};
			}
			\foreach \i in {0,...,4}
			{
				\draw[thick, black] (\i,0)--(\i+1,0);
			}
			\foreach \i in {0,2}
			{
				\draw[thick, black] (\i,0) to [out = 50, in = 180] (\i+1,0.5);
				\draw[thick, black] (\i+1,0.5) to [out = 0, in = 130] (\i+2,0);
				\draw[thick, black] (\i+1,0) to [out = -50, in = 180] (\i+2,-0.5);
				\draw[thick, black] (\i+2,-0.5) to [out = 0, in = -130] (\i+3,0);
			}
			\draw[thick, black] (0,0) to [bend left = 45] (1,0);
			\draw[thick, black] (4,0) to [out = 50, in = 180] (5,0.5);
			\draw[thick, black] (5,0.5) to (5.3,0.5);
			\draw[thick, black] (5,0) to (5.3,0);
			\draw[thick, black] (5,0) to ({5+0.3*cos(50)},{-0.3*sin(50)});
			\node at (6,0){\dots};
			\node at (6,-0.5){\dots};
			\node at (6,0.5){\dots};
			\begin{scope}[shift={(12,0)}, xscale = -1]
				\foreach \i in {0,...,5}
				{
					\node[mynode] (n\i) at (\i,0){};
				}
				\foreach \i in {0,...,4}
				{
					\draw[thick, black] (\i,0)--(\i+1,0);
				}
				\foreach \i in {0,2}
				{
					\draw[thick, black] (\i,0) to [out = 50, in = 180] (\i+1,0.5);
					\draw[thick, black] (\i+1,0.5) to [out = 0, in = 130] (\i+2,0);
					\draw[thick, black] (\i+1,0) to [out = -50, in = 180] (\i+2,-0.5);
					\draw[thick, black] (\i+2,-0.5) to [out = 0, in = -130] (\i+3,0);
				}
				\draw[thick, black] (0,0) to [bend left = 45] (1,0);
				\draw[thick, black] (4,0) to [out = 50, in = 180] (5,0.5);
				\draw[thick, black] (5,0.5) to (5.3,0.5);
				\draw[thick, black] (5,0) to (5.3,0);
				\draw[thick, black] (5,0) to ({5+0.3*cos(50)},{-0.3*sin(50)});	
			\end{scope}
			\foreach \i in {0,...,11}
			{
				\draw[thick, dashed] (\i+0.5,-1.2)--(\i+0.5,0.7);
			}
		\end{tikzpicture}
	\end{center}
	\caption{A $3$-edge-connected planar graph $G$, with $3$-cuts indicated by dashed vertical lines. One link (dashed, red) suffices to cover all $3$-cuts in $G$.}\label{fig:snug_chain_1}
\end{figure}
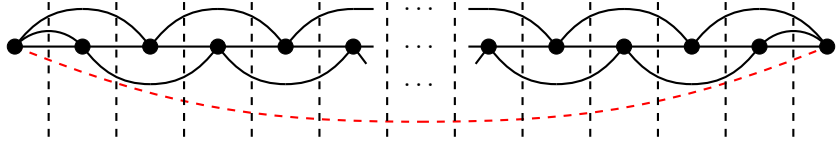
We show that structures similar to the one shown in \cref{fig:snug_chain_1}, which we call \emph{snug paths}, form the only obstacle in applying the above strategy to obtain a PTAS for planar $k$-WCAP with bounded cost ratio. More precisely, we prove that if an instance of planar $k$-WCAP with bounded cost ratio does not contain any snug paths, then the total cost of all links can be bounded in terms of the cost of an optimum solution. To obtain this result, we will assume that a given instance $(G,L,c)$ of planar $k$-WCAP is feasible (i.e., $G+L$ is $(k+1)$-edge-connected) and has the property that $G$ is \emph{minimally $k$-edge-connected}, meaning that every edge of $G$ appears in a cut of size $k$.
Both assumptions are without loss of generality: Feasibility can be checked in polynomial time and if $G$ is not minimally $k$-edge-connected, we can contract all edges that do not appear in any $k$-cut.
Note that contractions preserve planarity.

\subsubsection{Identifying the problem: snug vertices and snug paths}

For the remainder of this section, we fix a planar $k$-WCAP instance $(G=(V,E),L,c)$, where $G$ is minimally $k$-edge-connected, as well as an arbitrary \emph{root} $r\in V$. 
We call a vertex $v\in V\setminus \{r\}$ \emph{snug} if $\mathrm{deg}(v)\ge k+1$ (where $\mathrm{deg}(v)$ denotes the degree of $v$) and there exist two distinct $k$-cuts $S_1, S_2 \subseteq V\setminus\{r\}$ with $S_2 = S_1 \cup \{v\}$. 
We denote the set of snug vertices by $V_{\mathrm{snug}}$. One can show that for $v\in V_{\mathrm{snug}}$, there exists a \emph{unique} pair of distinct $k$-cuts $(S_1^v,S_2^v)$ with $S^v_1, S^v_2\subseteq V\setminus\{r\}$ and $S^v_2 = S^v_1 \cup \{v\}$ (see \cref{lemma:snug_shores_unique}), which we call the \emph{snug shores} for $v$. 
In \cref{fig:snug_chain_1}, every vertex except for the outer two vertices is snug, its snug shores are the $k$-cuts drawn to its left and right; the snug vertex ``snugly fits in-between''.

To capture the relations between snug vertices and their snug shores, we introduce the directed \emph{chain graph} $G_{\mathrm{chain}}=(V_{\mathrm{snug}},A)$. It contains an arc $(u,v)$ if and only if $u$ and $v$ are snug vertices such that $S_2^u=S_1^v$, i.e., the ``outer'' snug shore for $u$ coincides with the ``inner'' snug shore for $v$.
In \cref{fig:snug_chain_1}, the chain graph consists of a single directed path traversing all vertices except for the two outer ones from left to right. Note that $G_{\rm{chain}}$ can be computed in polynomial time, e.g., by enumerating all minimum cuts of $G$ (see \cite{doi:10.1137/S0895480194271323}). In \cref{sec:snug_vertices_and_chains}, we prove the following structure results.
\begin{restatable}{lemma}{lemmachaingraphpaths}\label{lemma:chain_graph_paths}
$G_{\mathrm{chain}}$ is a vertex-disjoint union of directed paths.
\end{restatable}
\begin{restatable}{lemma}{lemmachaingraphsubgraph}\label{lemma:chain_graph_subgraph}
	The underlying undirected graph of $G_{\mathrm{chain}}$ is a subgraph of $G$.
\end{restatable}
\begin{restatable}{lemma}{lemmaonlycutintersectingsnugedge}\label{lemma:only_cut_intersecting_snug_edge}
An arc $(u,v)\in A$ does not cross any $k$-cut other than $S_2^u=S_1^v$ (and $V\setminus S_2^u$).
\end{restatable} 
We call the maximal paths in  $G_{\mathrm{chain}}$ the \emph{snug paths} of $G$ and denote their collection by $\mathcal{P}^G_{\mathrm{chain}}$. \cref{lemma:bound_non_snug_vertices_and_snug_paths}, proven in \cref{sec:thinning}, formalizes our claim that snug paths constitute the only obstacle in bounding $c(L)$.
\begin{lemma} \label{lemma:bound_non_snug_vertices_and_snug_paths}
We have $|V\setminus V_{\mathrm{snug}}|< 4\cdot n_k(G)$ and $|\mathcal{P}^G_{\mathrm{chain}}|<2\cdot n_k(G)$, where $n_k(G)$ denotes the number of vertices of $G$ of degree $k$.
\end{lemma}
Observe that if $G_{\rm{chain}}$ does not contain any arcs, i.e., every snug path consists of a single snug vertex, then \cref{lemma:bound_non_snug_vertices_and_snug_paths} and planarity of $G+L$ allow us to linearly bound $|L|$ in terms of $n_k(G)$. On the other hand, in an optimum solution $\OPT$, every vertex of degree $k$ must be incident to at least one link in $\OPT$, implying $2\cdot |\OPT|\ge n_k(G)$. As the cost ratio of the instance is bounded, $c(L)$ can then be at most a constant factor larger than $c(\OPT)$.

\subsubsection{Solving the problem: contracting snug paths}

\cref{lemma:bound_non_snug_vertices_and_snug_paths} tells us that in the graph $G/\mathcal{P}^G_{\mathrm{chain}}$ that arises from $G$ by \emph{contracting every snug path into a single vertex}, the total number of vertices will be in $\mathcal{O}(|\OPT|)$. Note that by \cref{lemma:chain_graph_subgraph}, these contractions preserve planarity. By \cref{lemma:only_cut_intersecting_snug_edge}, contracting $\mathcal{P}^G_{\mathrm{chain}}$ preserves all $k$-cuts except for snug shores. 

This suggests the following modified strategy: Apply \cref{theorem:compute_edge_safe_cover_overview} to $\widetilde{G}_\mathcal{P}\coloneqq(V_{\mathcal{P}},E_{\mathcal{P}}\dot{\cup} L_{\mathcal{P}})\coloneqq (G+L)/\mathcal{P}^G_{\rm{chain}}$, the graph arising from $G+L$ by the contraction of snug paths. This yields a $(k+1)$-edge-safe cover $\mathcal{U}^\mathcal{P}=(U^\mathcal{P}_i)_{i\in I}$ of $\widetilde{G}_{\mathcal{P}}$ such that each subinstance $\widetilde{G}_{\mathcal{P}}/(E_{\mathcal{P}}\dot{\cup} L_{\mathcal{P}})[V_{\mathcal{P}}\setminus U^\mathcal{P}_i]$ has bounded treewidth and such that the set of links incident to $\bigcup_{i, j\in I: i\neq j} (U^\mathcal{P}_i\cap U^\mathcal{P}_j)$ is cheap. Lift $\mathcal{U}^\mathcal{P}$ to a ``cover'' $\mathcal{U}$ for $G$.

We remark that we cannot hope for $\mathcal{U}$ to be a $(k+1)$-edge-safe cover for $G+L$; this is because we cannot control the interaction between snug paths and connected cuts of size more than $k$. Instead, we introduce the weaker notion of a \emph{$k$-augmentation-safe cover}, which is tailored to the augmentation setting and only considers $k$-cuts of $G$ rather than all connected cuts (see \cref{def:augmentation_safe_cover}).

To obtain a good $k$-augmentation-safe cover by lifting a $(k+1)$-edge-safe cover in $\widetilde{G}_{\mathcal{P}}$, there are multiple challenges we have to address.

First of all, we have to discuss how to solve the $k$-WCAP in each of the subinstances we obtain in the uncontracted instance, which may no longer have bounded treewidth.
To this end, we introduce the concept of \emph{snug-treewidth}.\footnote{We remark that $\mathcal{P}^G_{\mathrm{chain}}$ depends on the (fixed) choice of the root that we implicitly make. While this dependence can be avoided, it simplifies the presentation.} 
\begin{definition}
	The \emph{snug-treewidth} $\mathrm{snugtw}(G,L)$ of an instance $(G,L,c)$ of $k$-WCAP is the treewidth of $(G+L)/\mathcal{P}^G_{\mathrm{chain}}$.
\end{definition}
 In \cref{sec:bounded_snug_treewidth} we prove that instances with constant snug-treewidth can be solved optimally via a dynamic programming approach.
\begin{restatable}{theorem}{wcapBoundedSnugTreewidth}\label{thm:WCAP_bounded_snug_treewidth}
	For any constant $\tau\in\mathbb{Z}_{\ge 0}$, $k$-WCAP can be solved in polynomial time on instances of snug-treewidth at most $\tau$.
\end{restatable}
While \cref{theorem:compute_edge_safe_cover_overview} allows us to compute a $(k+1)$-edge-safe cover $\mathcal{U}^\mathcal{P}$ in $\widetilde{G}_{\mathcal{P}}$ such that each piece $\widetilde{G}_{\mathcal{P}}/(E_{\mathcal{P}}\dot{\cup} L_{\mathcal{P}})[V_{\mathcal{P}}\setminus U^\mathcal{P}_i]$ has bounded treewidth, we point out that its lift $\mathcal{U}$ to $G$ will not necessarily have bounded snug-treewidth in each piece. The reason for this is that a snug path of $G$ contained in $G/(E\dot{\cup}L)[V\setminus U_i]$ does not necessarily constitute a snug path of this subinstance. The latter can only be guaranteed for snug paths that are ``far enough from the boundary of $U_i$'', ensuring that the corresponding snug shores also correspond to $k$-cuts in $G/(E\dot{\cup}L)[V\setminus U_i]$. 
Similarly, we cannot even guarantee the properties of a $k$-augmentation-safe cover for snug shores too close to the ``boundaries'' of the $U_i$.  
To deal with these issues, we have to ``buy'' a set of additional links covering the snug shores of the collection $\mathcal{Q}$ of snug paths situated close to the ``boundaries'' of the sets $(U_i)_{i\in I}$, making sure that we can safely keep these paths contracted.

Another challenge we need to address is how to bound the total cost of the links present in the contracted instance by $\mathcal{O}(c(\OPT))$ to make sure that the link set $L_{\mathcal{U}}$ is cheap. We point out that this is not an immediate consequence of \cref{lemma:bound_non_snug_vertices_and_snug_paths}.
The reason for this is that the contraction produces parallel links, and to ensure feasibility after un-contracting, it is not sufficient to keep only a cheapest copy.
Instead, we work with a carefully ``thinned'' link set $L^\star$, ensuring that $L^\star_{\mathcal{U}}$ is cheap, and combining it with solutions to the subinstances guarantees feasibility. 
Our main technical result is \cref{theorem:PTAS_key_theorem}, which we prove in \cref{sec:PTAS_WCAP}. 
The link set $L^{\star\star}$ in the theorem below contains both $L_{\mathcal{U}}^{\star}$ and the above-mentioned links we buy to cover the snug shores of the collection $\mathcal{Q}$ of snug paths situated close to the boundaries.
\begin{restatable}{theorem}{theoremPTASkeytheorem}\label{theorem:PTAS_key_theorem}
Let $\epsilon\in(0,1)$. Given an instance $(G,L,c)$ of planar $k$-WCAP with cost ratio $\Delta$, we can, in polynomial time, compute $\mathcal{Q}\subseteq \mathcal{P}^G_{\rm{chain}}$, link sets $L^\star,L^{\star\star}\subseteq L$ such that the graph $(G+L^\star)/\mathcal{Q}$ is $(k+1)$-edge-connected, and a $k$-augmentation-safe cover $\mathcal{U}=(U_i)_{i\in I}$ for the instance $(G_\mathcal{Q}=(V_\mathcal{Q},E_{\mathcal{Q}}),\allowbreak L^\star_\mathcal{Q},c_\mathcal{Q})\coloneqq (G,L^\star,c)/\mathcal{Q}$ such that:
\begin{enumerate}[(i)]
	\item \label{item:size_of_cover} $|I|\in \mathcal{O}(|V(G)|)$;
	\item  \label{item:bounded_snug_treewidth_pieces} For $i\in I$, the subinstance $(G^i_\mathcal{Q},L^i_\mathcal{Q},c^i_\mathcal{Q})\coloneqq (G_\mathcal{Q},L^\star_\mathcal{Q},c_\mathcal{Q})/(E_\mathcal{Q}\dot{\cup}L^\star_\mathcal{Q})[V_\mathcal{Q}\setminus U_i]$ has snug-treewidth $\mathcal{O}(\frac{k\Delta^2}{\epsilon^2})$;
	\item \label{item:combine_to_good_solution} For each  $i\in I$, let $\OPT_i\subseteq L^\star$ correspond to an optimum solution to the subinstance $(G^i_\mathcal{Q},L^i_\mathcal{Q},c^i_\mathcal{Q})$. Then $F\coloneqq(\bigcup_{i\in I} \OPT_i)\cup L^{\star\star}$ is a solution to the instance $(G,L,c)$ with $c(F)\le (1+\epsilon)\cdot c(\OPT)$.
\end{enumerate}
\end{restatable}

\section{Constructing $k$-safe covers}\label{sec:k_safe_covers}

In this section, we show how to construct good $k$-edge-safe covers, that is, we prove the following slight strengthening of \Cref{theorem:compute_edge_safe_cover_overview}.
The additional properties \Cref{item:inAtMostTwo,item:noEdgesAcross} will prove helpful in the running time analysis.

\begin{theorem}\label{theorem:compute_edge_safe_cover}
    Given a planar connected graph $G=(V,E)$ with edge costs $c\colon E \to \mathbb{R}_{\geq 0}$ and $\epsilon \in (0,1)$, we can compute, in time $\mathcal{O}(|V|+|E|)$, a $k$-edge-safe cover $\mathcal{U}=(U_i)_{i\in I}$ of $G$, together with the graphs $G_i\coloneqq (V,E)/E[V\setminus U_i]$ for $i\in I$, such that
\begin{enumerate}[(i)]
    \item \label{k-edge-safe_cover_costs_full} $c(E_{\mathcal{U}}) \le \epsilon \cdot c(E)$,
    \item \label{k-edge-safe_cover_treewidth_full} $G_i$ has treewidth at most $\frac{26k}{\epsilon}$ for each $i\in I$,
    \item\label{item:inAtMostTwo} every vertex is contained in at most $2$ sets $U_i$, and
    \item\label{item:noEdgesAcross} for distinct $i,j\in I$, no edge has one endpoint in $U_i \setminus U_j$ and one in $U_j \setminus U_i$.
\end{enumerate}
\end{theorem}

In order to also be able to handle $k$-vertex-connectivity, we also introduce the notion of a \emph{$k$-vertex-safe cover}, analogous to $k$-edge-safe covers.
For  a graph $G=(V,E)$, an edge set $F\subseteq E$ and a vertex set $S\subseteq V$, we denote by $\Gamma_F(S) \coloneqq \{ v \in V\setminus S : \{v,w\}\in F\text{ for some }w\in S\}$ the neighbors of $S$ with respect to edges in $F$.
For $F=E$, we write $\Gamma(S)\coloneqq \Gamma_E(S)$.
Moreover, if we want to be explicit about the underlying graph $G$, then we also write $\Gamma_G(S) \coloneqq \Gamma(S)$.
We use analogous notation for the set of edges crossing a cut $S$, i.e., for an edge set $F$, we denote by $\delta_F(S)$ the set of edges in $F$ with one endpoint in $S$ and one endpoint in $V\setminus S$. Moreover, we will write $\delta_G(S)$ to be explicit about the underlying graph $G$.

\begin{definition}[$k$-vertex-safe cover]\label{def:k_vertex_safe_cover}
Let $G=(V,E)$ be a graph and let $\mathcal{U}=(U_i)_{i\in I}$ be a collection of nonempty subsets of $V$ that cover $V$.
Let $V_{\mathcal{U}}$ be the set of vertices that appear in at least $2$ distinct $U_i$'s and let $E_\mathcal{U}$ be all edges with at least one endpoint in $V_{\mathcal{U}}$. 
We call $\mathcal{U}$ a \emph{$k$-vertex-safe cover of $G$} if for every connected cut $S\subseteq V$, we have
\begin{enumerate}[(i)]
    \item $|\Gamma_{E_\mathcal{U}}(S)|\ge k$, or   
    \item $\delta(S)\subseteq E[U_i]$ for some $i\in I$.
\end{enumerate}
\end{definition}

We remark that a $k$-vertex-safe cover is also a $k$-edge-safe cover, because $|\Gamma_{E_\mathcal{U}}(S)| \leq |\delta(S) \cap E_\mathcal{U}|$ for all $S\subseteq V$.
We will prove the existence of good $k$-vertex-safe covers, analogous to  \Cref{theorem:compute_edge_safe_cover} for $k$-edge-safe covers.

\begin{theorem}\label{theorem:compute_k_vertex_safe_cover}
Given a planar connected graph $G=(V,E)$ with edge costs $c: E \to \mathbb{R}_{\geq 0}$ and $\epsilon \in (0,1)$, we can compute, in time $\mathcal{O}(|V|+|E|)$, a $k$-vertex-safe cover $\mathcal{U}=(U_i)_{i\in I}$ of $G$, together with the graphs $G_i\coloneqq (V,E)/E[V\setminus U_i]$ for $i\in I$, such that
\begin{enumerate}[(i)]
    \item \label{k-vertex_safe_edge_costs} $c(E_{\mathcal{U}}) \le \epsilon \cdot c(E)$,
    \item \label{k-vertex_safe_tree_width}
        $G_i$ has treewidth at most $\frac{26k}{\epsilon}$ for each $i\in I$,
    \item\label{item:inAtMostTwo_vertex} every vertex is contained in at most $2$ sets $U_i$, and
    \item\label{item:noEdgesAcross_vertex} for distinct $i,j\in I$, no edge has one endpoint in $U_i \setminus U_j$ and one in $U_j \setminus U_i$.
\end{enumerate}
\end{theorem}

The notion of $k$-vertex-safe covers captures the properties we need for applications involving $k$-vertex-connectivity.
For later extensions to planar $k$-CAP (in \Cref{sec:augmentation}), we will prove the following generalization of \Cref{theorem:compute_edge_safe_cover,theorem:compute_k_vertex_safe_cover} that, in addition to edge costs, also allows for incorporating vertex weights. 
\begin{theorem}\label{theorem:compute_k_vertex_safe_cover_with_vertex_weights}
Given a planar connected graph $G=(V,E)$ with vertex weights $w\colon V\rightarrow\mathbb{R}_{\ge 0}$, edge costs $c: E \to \mathbb{R}_{\geq 0}$, and $\delta \in (0,1)$, we can compute, in time $\mathcal{O}(|V|+|E|)$, a $k$-vertex-safe cover $\mathcal{U}=(U_i)_{i\in I}$ of $G$, together with the graphs $G_i\coloneqq (V,E)/E[V\setminus U_i]$ for $i\in I$, such that
\begin{enumerate}[(i)]
	\item \label{k-vertex_safe_edge_and_vertex_costs} $w(V_{\mathcal{U}}\cup\Gamma(V_{\mathcal{U}}))+ c(E_{\mathcal{U}}) \le \delta \cdot (w(V)+c(E))$, 
	\item \label{k-vertex_safe_stronger_tree_width}
    $G_i$ has treewidth at most $\frac{26k}{\delta}$ for each $i\in I$,
	\item\label{item:inAtMostTwo_stronger} every vertex is contained in at most $2$ sets $U_i$, and
	\item\label{item:noEdgesAcross_stronger} for distinct $i,j\in I$, no edge has one endpoint in $U_i \setminus U_j$ and one in $U_j \setminus U_i$.
\end{enumerate}
\end{theorem}
Then \Cref{theorem:compute_edge_safe_cover,theorem:compute_k_vertex_safe_cover} follow by setting all vertex weights to $0$ and because every $k$-vertex-safe cover is a $k$-edge-safe cover.

Throughout our algorithm, we work with a fixed planar embedding of our graph $G$.
In our proof of \Cref{theorem:compute_k_vertex_safe_cover_with_vertex_weights}, we will exploit the following well-known correspondence between connected cuts and cycles in the planar dual.

\begin{observation}\label{obs:cuts_dual_cycle}
Let $G=(V,E)$ be a connected planar graph.
A cut $S \subsetneq V, S\neq \emptyset$ is connected if and only if $\delta(S)$ corresponds to a cycle in the planar dual of $G$.
\end{observation}

Motivated by this observation, we work with the planar dual of the graph $G$.
To handle vertex-connectivity, it will be useful to work with the following variant of the planar dual:

\begin{definition}[vertex-face graph]
Let $G=(V,E)$ be a planar graph with a fixed planar embedding.
We define a bipartite graph $D_G$ with
\begin{itemize}
    \item vertex set $V\cup F$, where $F$ is the set of faces of the fixed planar embedding of $G$, and
    \item there is an edge between $v\in V$ and $f\in F$ if $v$ is a vertex on face $f$.
\end{itemize}
\end{definition}

Observe that this graph is bipartite and two faces have distance $2$ in the graph $D_G$ if and only if the cycles bounding them share a vertex.
(For comparison, in the planar dual of $G$ two faces are adjacent if and only if the cycles bounding them share an edge.)

To construct a $k$-vertex-safe cover, we fix an arbitrary vertex $r\in V$, which we call \emph{root}.
We remark that in $D_G$, every vertex has an even distance from the root while every face has an odd distance from $r$.
The sets $U_i$ belonging to our $k$-vertex-safe cover $\mathcal{U}$ will correspond to rings around the root $r$.
More precisely, $U_i$ will be the set of all vertices whose distance  from $r$ in $D_G$ is at least $2\alpha_i$ and less than $2\beta_i$ for some suitable $\alpha_i,\beta_i \in \mathbb{Z}_{\geq 0}$.
  
\begin{definition}[$[\alpha, \beta)$-ring]
For $\alpha, \beta \in \mathbb{Z}_{\ge 0}$ with $\alpha \leq \beta$, the $[\alpha, \beta)$-ring (centered at $r$) is the vertex set
\[
V_{[\alpha,\beta)} \coloneqq \Big\{ v\in V \colon \mathrm{dist}_{D_G}(r,v) \in [2\alpha,2\beta) \Big\}.
\]
We also define $V_{< \alpha} \coloneqq V_{[0,\alpha)}$, $V_{\geq\beta} \coloneqq V_{[\beta, \infty)}$, and $V_\alpha \coloneqq V_{[\alpha, \alpha+1)}$.
\end{definition}

Then every set $U_i$ in our $k$-vertex-safe cover will be defined as $U_i \coloneqq V_{[\alpha_i,\beta_i)}$ for some $\alpha_i,\beta_i \in \mathbb{Z}_{\geq 0}$.
In \Cref{theorem:compute_k_vertex_safe_cover}, we want that the graph $G_i\coloneqq (V,E)/E[V\setminus U_i]$ has small treewidth.
In order to establish this property, we will show that the graph $G_i = G/E[V\setminus V_{[\alpha_i,\beta_i)}]$ has treewidth at most $3(\beta_i-\alpha_i) + 5$.
To prove this, we need the following lemma:

\begin{lemma}\label{lem:G_alpha_connected}
Let $\alpha\in \mathbb{Z}_{\geq 1}$. Then the graph $G[V_{<\alpha}]$ is connected.    
\end{lemma}
\begin{proof}
Let $v\in V_{<\alpha}$. We show the existence of an $r$-$v$ walk in $G[V_{<\alpha}]$ by induction on $\mathrm{dist}_{D_G}(r,v)$. If $\mathrm{dist}_{D_G}(r,v)=0$, then $r=v$ and there is nothing to show. Next, assume $\mathrm{dist}_{D_G}(r,v)>0$. Let $f\in F$ be the predecessor of $v$ on a shortest $r$-$v$ path in $D_G$ and let $w\in V$ be the predecessor of~$f$. Then $\mathrm{dist}_{D_G}(r,w)=\mathrm{dist}_{D_G}(r,v)-2$, $\mathrm{dist}_{D_G}(r,f)=\mathrm{dist}_{D_G}(r,v)-1$, and by walking along the edges of $G$ bounding the face~$f$, we obtain a $w$-$v$ walk $Q$ in $G$ with the property that every vertex on $Q$ is incident to $f$, and as such, has a distance of at most $\mathrm{dist}_{D_G}(r,f)+1=\mathrm{dist}_{D_G}(r,v)< \alpha$ to~$r$. Combining $Q$ with an $r$-$w$ walk in $G[V_{<\alpha}]$ concludes the induction step.

\end{proof}

The following result by \textcite{bodlaenderPartialKarboretumGraphs1998} shows that in order to prove that the graphs $G_i$ have small treewidth, it suffices to show that they are $k$-outerplanar for some small number $k$.
A graph $H$ is $k$-outerplanar if it can be embedded in the plane such that, after removing the vertices on the outer face, the remaining graph is $(k-1)$-outerplanar; moreover, $1$-outerplanar graphs are exactly the outerplanar graphs, where all vertices lie on the outer face.
Note that this is equivalent to saying that $H$ admits an embedding for which there exists a face $f_0$ such that every vertex has distance at most $2k-1$ from $f_0$ in $D_H$.

\begin{lemma}[\textcite{bodlaenderPartialKarboretumGraphs1998}]\label{lem:outerplanar_treewidth}
Any $k$-outerplanar graph has treewidth at most $3 k-1$.
\end{lemma}

The bound from \Cref{lem:outerplanar_treewidth} has later been shown to be tight (see~\cite{kammerLowerBoundTreewidth2009}).
Using \Cref{lem:G_alpha_connected,lem:outerplanar_treewidth}, we can now establish an upper bound on the treewidth of  $G_i = G/E[V\setminus V_{[\alpha_i,\beta_i)}]$.

\begin{lemma}\label{lemma:treewidth_G_i}
Let $\alpha, \beta \in \mathbb{Z}_{\geq 0}$ with $\alpha < \beta$.
Then the graph  $G/E[V\setminus V_{[\alpha,\beta)}]$ is $(\beta-\alpha+2)$-outerplanar. 
In particular, $\mathrm{treewidth}(G/E[V\setminus V_{[\alpha,\beta)}]) \le 3\cdot(\beta-\alpha)+5$.
\end{lemma}
\begin{proof}
The fixed embedding of $G$ gives rise to an embedding of $G'\coloneqq G/E[V\setminus V_{[\alpha,\beta)}]$. 
By \Cref{lem:G_alpha_connected}, the subgraphs $G[V_{< \alpha}]$ and $G[V_{< \alpha+1}]$ of $G$ are connected  and thus, every vertex in $V_{\alpha}$ is connected to the super-vertex $r'$ arising from the contraction of $V_{< \alpha}$ in $G'$. 
We will show that every vertex $v$ of $G'$ has distance at most $2(\beta-\alpha+1)$ to the vertex $r'$ in $D_{G'}$.
Then, choosing $f_0$ as a face of $G'$ incident to $r'$, every vertex of $G'$ has a distance of at most $2\cdot(\beta-\alpha+2)-1$ to $f_0$ in $D_{G'}$. 
Hence, $G'$ is $(\beta-\alpha+2)$-outerplanar. 
In particular, by \Cref{lem:outerplanar_treewidth}, the graph $G'$ has treewidth at most $3\cdot(\beta-\alpha)+5$, as claimed.

It remains to prove that every vertex $v$ of $G'$ has distance at most $2(\beta-\alpha+1)$ to the vertex $r'$ in $D_{G'}$.
Let $v\in V_{[\alpha+1,\beta)}$ and let $P$ be a shortest $r$-$v$ path in $D_G$ with vertices $v_0=r,f_1,v_2,f_3,\dots,v_d=v$ visited by $P$ in this order.
As $f_{2\alpha+1}$,$f_{2\alpha+3},\dots,f_{d-1}$ are only incident to vertices in $V_{[\alpha,\beta)}$ and none of these are contracted, the $v_{2\alpha}$-$v$ subpath of $P$ also exists in $D_{G'}$. 
Hence, $\mathrm{dist}_{D_{G'}}(r',v)\le 2\cdot(d-\alpha)+2\le 2(\beta-\alpha)$. 
Finally, each super-vertex $w$ corresponding to a connected component of $G[V\setminus V_{[\alpha,\beta)}]$ other than $V_{<\alpha}$ must be adjacent to some vertex in $v\in V_{[\alpha,\beta)}$, implying $\mathrm{dist}_{D_{G'}}(r',w)\le \mathrm{dist}_{D_{G'}}(r',v)+\mathrm{dist}_{D_{G'}}(v,w)\le 2(\beta-\alpha)+2=2(\beta-\alpha+1)$. 

\end{proof}

We conclude this section by showing that good $k$-vertex-safe covers exist and can be constructed quickly, as claimed in \cref{theorem:compute_k_vertex_safe_cover_with_vertex_weights}.
See \cref{fig:safe_cover} for an illustration of the construction.

\begin{figure}
\begin{center}
\begin{tikzpicture}[
xscale=0.55,
yscale=0.55,
vertex/.style={thick,draw=white,fill=black,circle,minimum size=4,inner sep=2pt},
edge/.style={line width=2pt},
defaultedge/.style={},
cut/.style={line width=1.5pt, opacity=0.5, green!70!black}
]

\foreach \copy/\yshift in {1/5.5,2/-5.5} {
\begin{scope}[shift={(0,\yshift)}]

\foreach \i/\color in {1/{yellow!90!black},2/{orange},3/{orange},4/{red!80!black},5/{red!80!black}} {
  \pgfmathsetmacro\numsegments{int(\i*7)}
  
  \foreach \j  [evaluate=\j as \block using int(ceil(\j/\i)), evaluate=\j as \offset using {int(mod(\j-1, \i))}] in {1,...,\numsegments}{
  \pgfmathsetmacro\r{2*\i + (\j)*(360/\numsegments)}
  \node[vertex,fill=\color] (v\copy-\i-\block-\offset) at (\r:\i) {};
  }
}

\foreach [evaluate=\block as \nextblock using {int(mod(\block, 7) + 1)}] \block in {1, ...,7}{
		\draw (v\copy-1-\block-0) -- (v\copy-1-\nextblock-0);
		\begin{scope}[edge]
		\draw (v\copy-2-\block-0) -- (v\copy-2-\block-1) --(v\copy-2-\nextblock-0);
		\draw (v\copy-3-\block-0) -- (v\copy-3-\block-1)  -- (v\copy-3-\block-2) -- (v\copy-3-\nextblock-0);
		\end{scope}
		\draw (v\copy-4-\block-0) --  (v\copy-4-\block-1)  -- (v\copy-4-\block-2) -- (v\copy-4-\block-3) -- (v\copy-4-\nextblock-0);
		\draw (v\copy-5-\block-0) --  (v\copy-5-\block-1)  -- (v\copy-5-\block-2) -- (v\copy-5-\block-3) -- (v\copy-5-\block-4)-- (v\copy-5-\nextblock-0);
		
		\begin{scope}[edge]
		\draw (v\copy-1-\block-0) --  (v\copy-2-\block-0);
		\draw (v\copy-1-\block-0) -- (v\copy-2-\block-1);
		
		\draw (v\copy-2-\block-0) --  node[vertex, pos=0.5, fill=orange]  {} (v\copy-3-\block-1);
                 \draw (v\copy-2-\block-1) -- (v\copy-3-\nextblock-0);
                 
                 \draw (v\copy-4-\block-3) -- (v\copy-3-\block-1);
                 \draw (v\copy-4-\block-3) -- (v\copy-3-\nextblock-0);
                 \draw (v\copy-3-\block-1) -- (v\copy-4-\block-1);
                 \end{scope}
                 
                 \draw(v\copy-4-\block-3) -- (v\copy-5-\block-3);
                 \draw(v\copy-4-\block-3) -- (v\copy-5-\block-4);
                 \draw(v\copy-4-\block-3) -- node[vertex, pos=0.7, fill=red!80!black]  {} (v\copy-5-\nextblock-1);
                 \draw (v\copy-4-\block-2) -- (v\copy-5-\block-1);
}
\end{scope}
}

\foreach \copy/\yshift/\angle in {1/5.5/-35,2/-5.5/145} {
\begin{scope}[shift={(0,\yshift)}]
\foreach \i/\color/\edgestyle [evaluate=\i as \lastoffset using int(\i-1)] in {
	6/{red!80!black}/defaultedge,
	7/{red!80!black}/defaultedge,
	8/{red!50!blue}/edge,
	9/{red!50!blue}/edge,
	10/{blue!90!black}/defaultedge%
} {
  \pgfmathsetmacro\numsegments{int(\i*7)}
  \pgfmathsetmacro\numnodes{int(\i*4)}
  
  \foreach \j  [
  	evaluate=\j as \block using int(ceil(\j/\i)),
	evaluate=\j as \offset using {int(mod(\j-1, \i))},
  	evaluate=\j as \lastblock using int(ceil((\j-1)/\i)),
	evaluate=\j as \lastoffset using {int(Mod(\j-2, \i))}
	] in {1,...,\numnodes}{
  \pgfmathsetmacro\r{2*\i + (\j)*(360/\numsegments)+\angle}
  \node[vertex, fill=\color] (v\copy-\i-\block-\offset) at (\r:\i) {};
  \ifthenelse{\j > 1}{
  	\draw[\edgestyle] (v\copy-\i-\lastblock-\lastoffset) -- (v\copy-\i-\block-\offset);
  }{} 
  }
}
\end{scope}

}

\foreach \i/\color/\edgestyle [evaluate=\i as \lastoffset using int(\i-1)] in {
	6/{red!80!black}/defaultedge,
	7/{red!80!black}/defaultedge,
	8/{red!50!blue}/edge,
	9/{red!50!blue}/edge,
	10/{blue!90!black}/defaultedge%
}{
	\foreach \fromcopy/\tocopy in {1/2, 2/1}{
	 	\path[bend left] (v\fromcopy-\i-4-\lastoffset) to
			node[vertex, pos=0.1, fill=\color] (v\fromcopy-\i-5-0) {}
			node[vertex, pos=0.2, fill=\color] (v\fromcopy-\i-5-1) {}
			node[vertex, pos=0.3, fill=\color] (v\fromcopy-\i-5-2) {}
			node[vertex, pos=0.4, fill=\color] (v\fromcopy-\i-5-3) {}
			node[vertex, pos=0.5, fill=\color] (v\fromcopy-\i-5-4) {}
			node[vertex, pos=0.6, fill=\color] (v\fromcopy-\i-5-5) {}
			node[vertex, pos=0.7, fill=\color] (v\fromcopy-\i-5-6) {}
			node[vertex, pos=0.8, fill=\color] (v\fromcopy-\i-5-7) {}
			node[vertex, pos=0.9, fill=\color] (v\fromcopy-\i-5-8) {}
		(v\tocopy-\i-1-0);
		\foreach \fromoffset [evaluate=\fromoffset as \tooffset using int(\fromoffset+1)] in {0,...,7}{
			\draw[\edgestyle] (v\fromcopy-\i-5-\fromoffset) -- (v\fromcopy-\i-5-\tooffset);
		}
		\draw[\edgestyle] (v\fromcopy-\i-4-\lastoffset) -- (v\fromcopy-\i-5-0);
		\draw[\edgestyle] (v\fromcopy-\i-5-8) -- (v\tocopy-\i-1-0);
	}
}

\foreach \copy in {1,2}{
\foreach [evaluate=\block as \nextblock using {int(mod(\block, 5) + 1)}] \block in {1, ...,4}{

		\draw (v\copy-6-\block-0) --node[vertex, pos=0.5, fill=red!80!black]  {}  (v\copy-7-\block-0);
		\draw (v\copy-6-\block-0) -- (v\copy-7-\block-1);
		\draw (v\copy-6-\block-2) -- (v\copy-7-\block-1);
		\draw (v\copy-6-\block-2) -- (v\copy-7-\block-3);
		\draw (v\copy-6-\block-2) -- (v\copy-7-\block-4);
		\draw (v\copy-6-\block-5) -- (v\copy-7-\block-5);
		\draw (v\copy-6-\block-5) -- (v\copy-7-\block-6);

	        \begin{scope}[edge]
		\draw (v\copy-7-\block-0) --   (v\copy-8-\block-1);
		\draw (v\copy-7-\block-2) --   (v\copy-8-\block-1);
		\draw (v\copy-7-\block-2) --   (v\copy-8-\block-3);
		\draw (v\copy-7-\block-5) --   (v\copy-8-\block-4);
		\draw (v\copy-7-\block-6) --   (v\copy-8-\block-6);
		\draw[draw=none] (v\copy-7-\block-6) --  node (x) [vertex, pos=0.5, fill=red!80!black]  {} (v\copy-8-\nextblock-0);
		\draw (x) -- (v\copy-8-\nextblock-0);
		\end{scope}
		\draw (v\copy-7-\block-6) -- (x);

                 \begin{scope}[edge]
                 \draw (v\copy-9-\block-1) -- (v\copy-8-\block-1);
                 \draw (v\copy-9-\block-3) -- (v\copy-8-\block-2);
                 \draw (v\copy-9-\block-3) -- (v\copy-8-\block-3);
                 \draw (v\copy-9-\block-4) -- (v\copy-8-\block-3);
                 \draw (v\copy-9-\block-4) -- (v\copy-8-\block-4);
                 \draw (v\copy-9-\block-6) -- (v\copy-8-\block-5);
                 \draw (v\copy-9-\block-7) -- (v\copy-8-\block-6);
                 \draw (v\copy-9-\block-7) -- (v\copy-8-\block-7);
                 \draw (v\copy-9-\block-8) -- (v\copy-8-\block-7);
                 \draw (v\copy-9-\block-8) -- (v\copy-8-\nextblock-0);
        
                 \draw(v\copy-9-\block-0) -- (v\copy-10-\block-0);
                 \draw(v\copy-9-\block-1) -- (v\copy-10-\block-1);
                 \draw(v\copy-9-\block-1) -- (v\copy-10-\block-2);
                 \draw(v\copy-9-\block-3) -- (v\copy-10-\block-3);
                 \draw(v\copy-9-\block-4) -- (v\copy-10-\block-3);
                 \draw(v\copy-9-\block-5) -- (v\copy-10-\block-3);
                 \draw(v\copy-9-\block-6) -- (v\copy-10-\block-5);
                 \draw(v\copy-9-\block-7) -- (v\copy-10-\block-7);
                 \draw(v\copy-9-\block-7) -- (v\copy-10-\block-8);
                 \draw(v\copy-9-\block-8) -- (v\copy-10-\block-9);
                 \draw[draw=none](v\copy-9-\block-8) -- node[vertex, pos=0.4, fill=red!50!blue] (x)  {} (v\copy-10-\nextblock-0);
                 \draw (v\copy-9-\block-8) -- (x);
                 \end{scope}
                 \draw (x) -- (v\copy-10-\nextblock-0);
}
}

\foreach \copy/\other/\nextblock in {1/2/1,2/1/1}{
\foreach  \block in {5}{
    		\draw (v\copy-6-\block-0) --node[vertex, pos=0.5, fill=red!80!black]  {}  (v\copy-7-\block-0);
		\draw (v\copy-6-\block-0) -- (v\copy-7-\block-1);
		\draw (v\copy-6-\block-2) -- (v\copy-7-\block-1);
		\draw (v\copy-6-\block-2) -- (v\copy-7-\block-3);
		\draw (v\copy-6-\block-3) -- (v\copy-7-\block-4);
		\draw (v\copy-6-\block-5) -- (v\copy-7-\block-5);
		\draw (v\copy-6-\block-5) -- (v\copy-7-\block-6);
	         \draw (v\copy-6-\block-7) -- (v\copy-7-\block-7);
		\draw (v\copy-6-\block-8) -- (v\copy-7-\block-7);

		\begin{scope}[edge]
		\draw (v\copy-7-\block-0) --   (v\copy-8-\block-1);
		\draw (v\copy-7-\block-2) --   (v\copy-8-\block-1);
		\draw (v\copy-7-\block-2) --   (v\copy-8-\block-3);
		\draw (v\copy-7-\block-5) --   (v\copy-8-\block-4);
		\draw (v\copy-7-\block-6) --   (v\copy-8-\block-6);
		\draw (v\copy-7-\block-7) --   (v\copy-8-\block-7);
		\draw (v\copy-7-\block-7) --   (v\copy-8-\block-8);
		\draw (v\copy-7-\block-8) --   (v\other-8-\nextblock-0);
                 
                 \draw (v\copy-9-\block-1) -- (v\copy-8-\block-1);
                 \draw (v\copy-9-\block-3) -- (v\copy-8-\block-2);
                 \draw (v\copy-9-\block-3) -- (v\copy-8-\block-3);
                 \draw (v\copy-9-\block-4) -- (v\copy-8-\block-3);
                 \draw (v\copy-9-\block-4) -- (v\copy-8-\block-4);
                 \draw (v\copy-9-\block-6) -- (v\copy-8-\block-5);
                 \draw (v\copy-9-\block-7) -- (v\copy-8-\block-6);
                 \draw (v\copy-9-\block-7) -- (v\copy-8-\block-7);
                 \draw (v\copy-9-\block-8) -- (v\copy-8-\block-7);
                 \draw (v\copy-9-\block-8) -- (v\other-8-\nextblock-0);
        
                 \draw(v\copy-9-\block-0) -- (v\copy-10-\block-0);
                 \draw(v\copy-9-\block-1) -- (v\copy-10-\block-1);
                 \draw(v\copy-9-\block-1) -- (v\copy-10-\block-2);
                 \draw(v\copy-9-\block-3) -- (v\copy-10-\block-3);
                 \draw(v\copy-9-\block-4) -- (v\copy-10-\block-3);
                 \draw(v\copy-9-\block-5) -- (v\copy-10-\block-3);
                 \draw(v\copy-9-\block-6) -- (v\copy-10-\block-5);
                 \draw(v\copy-9-\block-7) -- (v\copy-10-\block-7);
                 \draw(v\copy-9-\block-7) -- (v\copy-10-\block-8);
                 \draw(v\copy-9-\block-8) -- (v\copy-10-\block-8);
                 \draw[draw=none] (v\copy-9-\block-8) -- node[vertex, pos=0.3, fill=red!50!blue] (x) {} (v\other-10-\nextblock-0);
                 \draw (v\copy-9-\block-8) -- (x);
                 \end{scope}
                 \draw (x) -- (v\other-10-\nextblock-0);
}
}

\foreach \copy/\copyoffset in {
	1/0,
	2/18%
}{
	\foreach \i/\j in {%
		32/1,
		35/2,
		1/3,
		2/6,
		3/7,
		5/9,
		6/11,
		6/14,
		8/15,
		10/17,
		11/17,
		12/18,
		13/20,
		15/21,
		18/23,
		19/26,
		20/28,
		22/28,
		22/30,
		24/30%
	}{
		\pgfmathsetmacro{\innerblock}{int(Mod(ceil((\i+\copyoffset)/5)-1, 7)+1)}
		\pgfmathsetmacro{\inneroffset}{int(mod(\i+\copyoffset-1, 5))}
		\pgfmathsetmacro{\outerblock}{int(ceil(\j/6))}
		\pgfmathsetmacro{\outeroffset}{int(mod(\j-1, 6))}
		\draw (v\copy-5-\innerblock-\inneroffset) -- (v\copy-6-\outerblock-\outeroffset);
	}
}

\node[vertex,blue] (r) at (-10,0) {};
\node[above=4pt] at (r) {\large$r$};
\draw (r) -- (v1-10-5-2);
\draw (r) -- (v1-10-5-4);
\draw (r) -- (v1-10-5-6);
\draw (r) -- (v1-10-5-8);

\draw (v1-6-5-5) -- (v2-5-2-4);
\draw (v1-6-5-6) -- (v2-5-3-0);
\draw (v2-6-5-5) -- (v1-5-6-4);
\draw (v2-6-5-5) -- (v1-5-6-2);
\draw (v2-6-5-6) -- (v1-5-7-0);
\draw (v1-5-5-3) -- (v2-5-2-1);

\begin{scope}[yshift=9mm]
\node[blue, right] at (-11,14.5) {$U_0 \setminus U_1$};
\node[blue!50!red,right] at (-11,13.8) {$U_0 \cap U_1$};
   \node[red!80!black, right] at (-11,13.1) {$U_1\setminus (U_0\cup U_2)$};
\node[orange,right] at (-11,12.4) {$U_1 \cap U_2$};
\node[yellow!70!black, right] at (-11,11.7) {$U_2\setminus U_1$};
\end{scope}
\end{tikzpicture}
\end{center}
\caption{\label{fig:safe_cover}
Example construction of a $k$-vertex-safe cover for $k=2$, $a^\ast =1$, and $M=3$ in the proof of \Cref{theorem:compute_k_vertex_safe_cover_with_vertex_weights}.
Edges in $E_{\mathcal{U}}$ are drawn bold.
}
\end{figure}
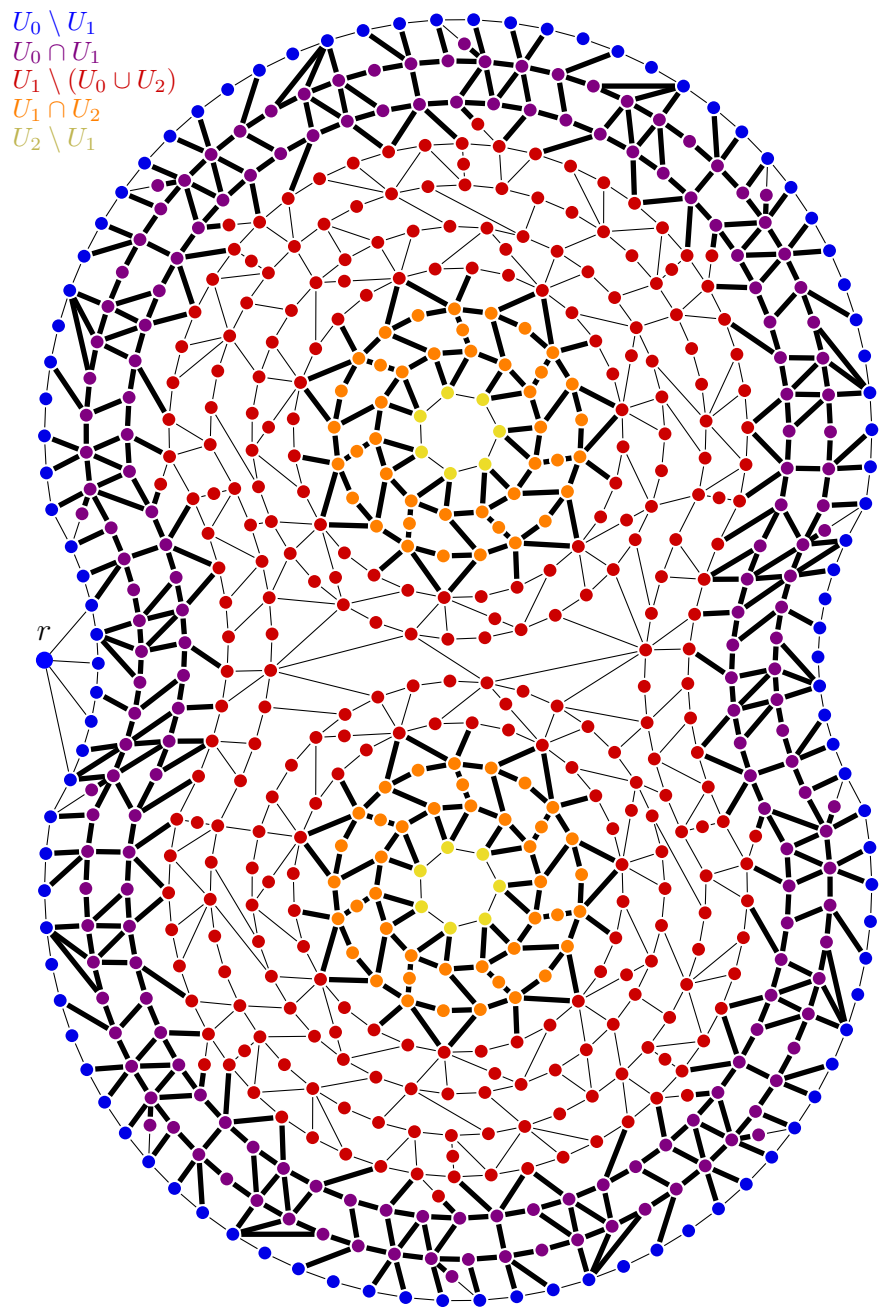

\begin{proof}[Proof of \Cref{theorem:compute_k_vertex_safe_cover_with_vertex_weights}] Let $G=(V,E)$ be a planar connected graph, equipped with vertex weights $w\colon V\rightarrow\mathbb{R}_{\ge 0}$ and edge costs $c\colon E\rightarrow\mathbb{R}_{\ge 0}$. Let further $\delta\in (0,1)$ and define $M\coloneqq\lceil\frac{3}{\delta}\rceil$.

We partition the vertex set into rings $R_j \coloneqq V_{[kj,\ k (j+1))}$ of width $k$. 
Every set $U_i$ in our $k$-vertex-safe cover will be the union of (up to) $M+1$ consecutive rings $R_j$, where every $M$-th ring will be included in two of the sets $U_i$, and will thus belong to $V_{\mathcal{U}}$.
In order to ensure \cref{theorem:compute_k_vertex_safe_cover_with_vertex_weights}~\ref{k-vertex_safe_edge_and_vertex_costs}, we will choose an appropriate offset $a^*\in \{0,\dots,M-1\}$ such that the first set $U_0$ is the union of the rings $R_1,\dots, R_{a^*}$, the next set $U_1$ is the union of the rings $R_{a^*}, \dots, R_{a^*+M}$, and so on.

For an offset $a\in \{0,\dots,M-1\}$, we consider the set of vertices in every $M$-th ring, starting with the ring $R_a$, that is,
\[
    W_a \coloneqq \bigcup_{t\in\mathbb{Z}_{\ge 0}} R_{a+tM}.
\] 
The sets $W_a$ with $a\in \{0,\dots, M-1\}$ form a partition of the vertex set $V$ into $M$ sets.
We further observe that $\Gamma(W_a)\subseteq W_{a-1}\cup W_{a}\cup W_{a+1}$, where $W_{-1}\coloneqq W_{M-1}$.
To see this, let $v\in R_{a+tM}$. For every neighbor $w$ of $v$, we have
\begin{align*}
2k\cdot (a-1+t\cdot M)&\le 2k\cdot (a+t\cdot M)-2\le\mathrm{dist}_{D_G}(r,v)-2\\
                      &\le \mathrm{dist}_{D_G}(r,w)\le \mathrm{dist}_{D_G}(r,v)+2\\
                      &< 2k\cdot (a+1+t\cdot M)+2\le 2k\cdot ((a+1)+1+t\cdot M),
\end{align*}
because $v$ and $w$ are both connected via an edge in $D_G$ to each of the incident faces of the edge $\{v,w\}$ of $G$. Hence, $w\in R_{a-1+tM}\cup R_{a+tM}\cup R_{a+1+tM}\subseteq W_{a-1}\cup W_{a}\cup W_{a+1}$. In particular, for a vertex $v\in V$, there exist at most three different offsets $a\in\{0,\dots,M-1\}$ such that $v\in W_{a}\cup\Gamma(W_{a})$.
Moreover, every edge has an endpoint in at most two different sets $W_{a}$.

Hence, we can in time $\mathcal{O}(|V|+|E|)$ compute an offset $a^*\in\{0,\dots,M-1\}$ such that
\begin{align}
w(W_{a^*}\cup \Gamma(W_{a^*}&))+c\big(\{e\in E\colon e\text{ has an endpoint in }W_{a^*}\}\big)\notag \\
        &\le\ \frac{3\cdot w(V)}{M}+\frac{2\cdot c(E)}{M} \le\ \delta\cdot (w(V)+c(E)).\label{eq:costs_E_U}
\end{align}
We define $U_0 \coloneqq \bigcup_{j=0}^{a^*} R_j$ and for $i\in \mathbb{Z}_{>0}$, we define
\[
U_i \coloneqq \bigcup_{j=a^*+(i-1)M}^{a^*+iM} R_j.
\]
In other words, $U_i = V_{[\alpha_i,\beta_i)}$ for  $\alpha_i\coloneqq \max\{0,k\cdot(a^*+(i-1)M)\}$ and  $\beta_i\coloneqq k\cdot (a^*+1+iM)$.
Let $I\coloneqq \{i\in\mathbb{Z}_{\ge 0}\colon U_i\ne \emptyset\}$. 
We can compute $I$ and the collection $(U_i)_{i\in I}$ in time $\mathcal{O}(|V|+|E|)$. 
At the end of the proof, we briefly discuss why this implies that we can also compute the graphs $G_i$ for $i\in I$ in time $\mathcal{O}(|V|+|E|)$.
Before doing so, we will show that $\mathcal{U}\coloneqq (U_i)_{i\in I}$ is a $k$-vertex-safe cover meeting properties \ref{k-vertex_safe_edge_and_vertex_costs}-\ref{item:noEdgesAcross_stronger} of \cref{theorem:compute_k_vertex_safe_cover_with_vertex_weights}.

For \ref{k-vertex_safe_edge_and_vertex_costs}, we note that $V_{\mathcal{U}}=W_{a ^*}$ and, as such, $w(V_{\mathcal{U}}\cup \Gamma(V_{\mathcal{U}}))+c(E_{\mathcal{U}})\le \delta\cdot (w(V)+c(E))$ by \eqref{eq:costs_E_U}.
As far as \ref{k-vertex_safe_stronger_tree_width} is concerned, \Cref{lemma:treewidth_G_i} tells us that for every $i\in I$, the graph $G_i$ has treewidth at most $3\cdot(\beta_i-\alpha_i)+5\le 3k(M+1)+5\le \frac{9k}{\delta}+12k+5\le \frac{26k}{\delta}$.
For \ref{item:inAtMostTwo_stronger}, we note that every vertex in $V_{\mathcal{U}}=W_{a^*}$ is contained in two consecutive $U_i$'s (namely, vertices in $R_{a^*+tM}$ are contained in $U_t$ and in $U_{t+1}$), whereas all other vertices appear in exactly one $U_i$.
Concerning \ref{item:noEdgesAcross_stronger}, recall that we established above that for a vertex $v\in R_j$, every neighbor of $v$ must be contained in $R_{j-1}\cup R_j\cup R_{j+1}$. Now, let $i,j\in I$ with $i < j$. If $j\ge i+2$, then the minimum index $a^*+(j-1)\cdot M$ of a ring in $U_j$ is by at least $M\ge 2$ larger than the maximum index $a^*+i\cdot M$ of a ring in $U_i$, meaning that no vertex of $U_i$ is adjacent to a vertex in $U_j$. On the other hand, if $j=i+1$, then $U_i\cap U_j=R_{a^*+iM}$ and the minimum index of a ring in $U_j\setminus U_i$ is $a^*+iM+1$, which is by $2$ larger than $a^*+iM-1$, which is the maximum index of a ring in $U_i\setminus U_j$.
Hence, no vertex of $U_i\setminus U_j$ is adjacent to a vertex in $U_j\setminus U_i$.

We now show that $\mathcal{U}$ is a $k$-vertex-safe cover.
Let $S\subseteq V$ be a connected cut. Then by \Cref{obs:cuts_dual_cycle}, we know that $\delta(S)$ corresponds to a cycle $C$ in the planar dual of $G$. 
Let $F_S$ be the set of faces that appear as vertices of this dual cycle, i.e., the set of faces incident to the edges in $\delta(S)$. 
We denote by $2\alpha+1$ and $2\beta+1$ the minimum and maximum among the values $\mathrm{dist}_{D_G}(r,f)$ with $f\in F_S$, respectively. 
Moreover, let $i\in\mathbb{Z}_{\ge 0}$ be maximum with $\alpha_i\le \alpha$. 
We show $|\Gamma_{E_\mathcal{U}}(S)|\ge k$ or $\delta(S)\subseteq E[U_i]$.

If $\beta<\beta_i-1$, then $\delta(S)\subseteq E[U_i]$ because the endpoints of every $e\in\delta(S)$ are incident with some face $f\in F_S$ and as such, their distances to $r$ can only range between $2\alpha\ge 2\alpha_i$ and $2 \beta+2<2 \beta_i$.
Now suppose $\beta\ge \beta_i-1$.
Then we choose faces $f_1,f_2\in F_S$ with $\mathrm{dist}_{D_G}(r,f_1)=2\alpha+1$ and $\mathrm{dist}_{D_G}(r,f_2)=2\beta+1$. 
The cycle $C$ consists of two edge-disjoint $f_1$-$f_2$ paths in the planar dual of~$G$. 
Let $P$ be one of these paths. 
For every edge $e\in E$ whose dual edge appears on $P$, its two incident faces share two vertices (the endpoints of $e$) and as such, their distances to $r$ in $D_G$ are either the same or differ by exactly two. 
In particular, on the path $P$, for every $\gamma\in [\alpha_{i+1}-1,\beta_i-2]\cap\mathbb{Z}\subseteq[\alpha,\beta-1]\cap\mathbb{Z}$, there is an edge $e_\gamma\in\delta(S)$ whose incident faces have distances $2\gamma+1$ and $2\gamma+3$ to the root $r$ in $D_G$. 
Then both endpoints of the edge $e_{\gamma}$ must have distance exactly $2\gamma+2$ to $r$ in $D_{G}$.
In particular, both endpoints of $e_{\gamma}$ are contained in $V_{[\alpha_{i+1}, \beta_{i})} = U_i\cap U_{i+1}\subseteq V_{\mathcal{U}}$.
Hence, because $e_{\gamma}\in \delta(S)$, the endpoint $w_{\gamma}\in V\setminus S$ is contained in $\Gamma_{E_{\mathcal{U}}}(S)$. 
As the vertices $w_{\gamma}$ for $\gamma\in [\alpha_{i+1}-1,\beta_i-2]\cap\mathbb{Z}$ are pairwise distinct (because their distances to $r$ in $D_G$ are distinct), we get $|\Gamma_{E_{\mathcal{U}}(S)}|\ge (\beta_i-2-(\alpha_{i+1}-1)+1)=\beta_i-\alpha_{i+1}=k$, as desired.

It remains to discuss that we can compute all graphs $G_i\coloneqq (V,E)/E[V\setminus U_i]$ for $i\in I$ in time $\mathcal{O}(|V|+|E|)$, once we know the sets $U_i$.
To this end, we first construct in linear time all graphs $\overline{G}_i\coloneqq G[U_i]$ for $i\in I$ by first starting with edge-free graphs $\overline{G}_i$ over the vertex sets $U_i$, and iterating over all edges $E$ and adding each edge $\{u,v\}\in E$ to all graphs $\overline{G}_i$ with $u,v\in U_i$.
For $i\in I$, to obtain the graph $G_i$ from $\overline{G}_i$, we need to add a vertex for each connected component of $G[V\setminus U_i]$ and connect it to all neighbors of this component in $U_i$.
We recall that $U_i\coloneqq V_{[\alpha_i,\beta_i)}$ and, by construction, there is no edge between $V_{<\alpha_i}$ and $V_{\geq \beta_i}$.
Hence, the connected components of $G[V\setminus U_i]$ are precisely the connected components of $G[V_{<\alpha_i}]$ and of $G[V_{\geq \beta_i}]$.
Moreover, by \Cref{lem:G_alpha_connected}, the graph $G[V_{<\alpha_i}]$ is connected and thus, there is only one connected component to consider here.
Moreover, among the vertices in $V_{<\alpha_i}$, only the vertices in $V_{\alpha_i-1}$ can have neighbors in $U_i$.
Thus, we can find all neighbors of the connected component $G[V_{<\alpha_i}]$ in $U_i$ by iterating over all vertices in $V_{\alpha_i-1}$.
As these vertex sets are disjoint for $i\in I$, this can be done for all $i\in I$ in total time $\mathcal{O}(|V|+|E|)$.
It remains to consider the connected components of $G[V_{\geq \beta_i}]$ and their neighbors in $U_i$.
An analogous argument to the $V_{<\alpha_i}$ case is not possible here, because $G[V_{\geq \beta_i}]$ may have multiple connected components. (For an example, consider the set $U_2 = V_{\geq 8}$ in \cref{fig:safe_cover}, which consists of two connected components.)

Nevertheless, also this can be done in linear time as follows.
We go through the indices $i\in I$ from largest to smallest and construct the connected components of $G[V_{\geq \beta_i}]$ (together with their neighbors in $U_i$) from the connected components of $G[V_{\geq \beta_{i+1}}]$ (which we have already computed in the previous iteration).
More precisely, instead of computing all vertices in each connected component of $G[V_{\geq \beta_i}]$, we only need to compute which vertices in $V_{\beta_i}$ belong to the same connected component of $G[V_{\geq \beta_i}]$.
This suffices to determine how the connected components of $G[V_{\geq \beta_i}]$ are connected to $U_i$, because among the vertices of $V_{\geq \beta_i}$, only the vertices in $V_{\beta_i}$ can have neighbors in $V_{< \beta_i}$.
Indeed, having computed this connected component structure for $G[V_{\geq \beta_{i+1}}]$, we can first construct the graph $G[V_{\geq \beta_i}]/G[V_{\geq \beta_{i+1}}]$, i.e., the graph obtained from $G[V_{\geq \beta_i}]$ by contracting all connected components of $G[V_{\geq \beta_{i+1}}]$. 
This can be done in time linear in $|V_{[\beta_i, \beta_{i+1}+1)}|$, by starting with the graph $G[V_{[\beta_i, \beta_{i+1})}]$ and then adding a vertex for each connected component of $G[V_{\geq \beta_{i+1}}]$ and connecting it to the neighbors of this component in $V_{[\beta_i, \beta_{i+1})}$.
Because the vertices in $V_{\beta_{i+1}}$ are the only ones among the vertices in $V_{\geq \beta_{i+1}}$ with neighbors in $V_{[\beta_i, \beta_{i+1})}$, only they are relevant for determining how the connected components of $G[V_{\geq \beta_{i+1}}]$ are connected to the other vertices of $G[V_{\geq \beta_i}]/G[V_{\geq \beta_{i+1}}]$.
Moreover, by the previous iteration, we know which vertices in $V_{\beta_{i+1}}$ belong to the same connected component of $G[V_{\geq \beta_{i+1}}]$.
Hence, adding a vertex per connected component of $G[V_{\geq \beta_{i+1}}]$ and connecting it to the neighbors of this component in $V_{[\beta_i, \beta_{i+1})}$ can be done in time linear in the number of edges incident to vertices in $V_{\beta_{i+1}}$.
Once $G[V_{\geq \beta_i}]/G[V_{\geq \beta_{i+1}}]$ is constructed, we can find its connected components, which correspond to the connected components of $G[V_{\geq \beta_{i}}]$, in time linear in its size.
Again, we only need to know which vertices of $V_{\beta_i}$ belong to the same connected component of $G[V_{\geq \beta_i}]$.
Finally, for each connected component of $G[V_{\geq \beta_i}]$, we can find its neighbors in $U_i$ by iterating over the edges incident to vertices in $V_{\beta_i}$, thus completing the construction of $G_i$.
Altogether, over all indices $i\in I$, this procedure takes therefore $\mathcal{O}(|V|+|E|)$ time, concluding the proof.

\end{proof}

\section{PTAS for planar $k$-WECSS and $k$-WVCSS with bounded cost ratio}\label{sec:WECSS_and_WVCSS}

In this section we provide the omitted details of the proofs of \Cref{thm:PTASs_k-WECSS_k-WVCSS_constant_k,cor:PTASs_k-WECSS}. We note that for constant $k$, one can check in linear time whether a planar graph is $k$-edge-connected/$k$-vertex-connected; see~\cite{Eppstein1999}. Hence, we may assume in the following that the instances of planar $k$-WECSS/$k$-WVCSS we work with are \emph{feasible}, i.e., the input graph is $k$-edge-connected/$k$-vertex-connected.

\subsection{Edge Connectivity}

This section contains the omitted parts of the proofs of \Cref{thm:PTASs_k-WECSS_k-WVCSS_constant_k,cor:PTASs_k-WECSS} for $k$-WECSS.
We restate the lemmas from Section~\ref{sec:summary_k_safe_covers} for convenience.

\GluingSolutions*
\begin{proof}
Let $\emptyset\neq T\subsetneq V$ and let $\emptyset\neq S\subsetneq V$ be a minimal cut in $G$ with $\delta_G(S)\subseteq \delta_G(T)$. Then $\delta_F(S)\subseteq \delta_F(T)$, so it suffices to show that $|\delta_F(S)|\ge k$. As $G$ is connected, the equivalence of connected cuts and minimal cuts in connected graphs implies that $S$ is a connected cut in $G$. In particular, \cref{def:k-edge-safe_cover}~\ref{def:k_edge_safe_1} or \cref{def:k-edge-safe_cover}~\ref{def:k_edge_safe_2} applies to $S$.
If there is some $i\in I$ such that $\delta_G(S)\subseteq E[U_i]$, then $S$ corresponds to a cut in $(V,F_i)/E[V\setminus U_i]$, implying $|\delta_G(S) \cap F_i\cap E[U_i]|=|\delta_G(S) \cap F_i|\ge k$. 
Otherwise, we have $|\delta_G(S) \cap E_\mathcal{U}|\ge k$.

\end{proof}

\CostCombinedSolution*
\begin{proof}
By \Cref{lemma:glue_solutions_together}, $F$ is a feasible solution to $k$-WECSS in $G$. Moreover, for each $i\in I$,
\[
T_i\ \coloneqq\ \OPT \setminus E[V\setminus U_i]\ \subseteq\ (\OPT\cap E[U_i])\cup (E(G_i)\setminus E[U_i])
\]
is a feasible solution to $k$-WECSS in $(G_i,c_i)$, implying $c_i(\OPT_i)\le c_i(T_i)=c((\OPT\cap E[U_i])\setminus E_{\mathcal{U}})$.
Hence,
 \begin{equation*}
     \sum_{i\in I}c_i(\OPT_i) \leq \sum_{i\in I} c((\OPT\cap E[U_i])\setminus E_{\mathcal{U}}) \le c(\OPT),
 \end{equation*}
because $E[U_i]\cap E[U_j]=E[U_i\cap U_j]\subseteq E_{\mathcal{U}}$ for every $i,j\in I$ with $i\neq j$. Finally,
\begin{align*}
    c(F)&= c(E_{\mathcal{U}}) + c(F\setminus E_\mathcal{U}) \le c(E_{\mathcal{U}}) + \sum_{i\in I} c((\OPT_i\cap E[U_i])\setminus E_{\mathcal{U}})\\
        &= c(E_{\mathcal{U}}) + \sum_{i\in I} c_i(\OPT_i) \le c(E_\mathcal{U}) + c(\OPT).
\end{align*}

\end{proof}

\begin{lemma}\label{lem:costs_of_all_edges}
Let $G=(V,E)$ be a planar $k$-edge-connected graph with at most $k$ parallel edges between any pair of vertices, and let $c\colon E\to \mathbb{R}_{>0}$.
Then
\begin{equation*}
c(E) \leq 6\cdot \Delta\cdot c(\OPT),
\end{equation*}
where $\OPT \subseteq E$ is an optimum $k$-WECSS solution in $(G,c)$, and $\Delta$ is the cost ratio of $c$.
\end{lemma}
\begin{proof}
Let $c_{min}$ and $c_{max}$ be the minimum and maximum edge costs, respectively. 
In the graph $(V,\OPT)$, every vertex has degree at least $k$, so $c(\OPT)\ge \frac{k}{2}\cdot c_{min} \cdot |V|$. 
As $G$ is planar, the underlying simple graph of $G$ contains at most $3\cdot |V|-6$ many edges, yielding
\[c(E)\le c_{max}\cdot |E|\le c_{max}\cdot k\cdot 3\cdot |V|\le (6\cdot \Delta)\cdot \frac{k}{2}\cdot c_{min} \cdot |V| \le 6\cdot \Delta\cdot c(\OPT).\]

\end{proof}

Combining these results yields our PTAS for planar $k$-WECSS with fixed $k$ and bounded cost ratio, i.e., \Cref{thm:PTASs_k-WECSS_k-WVCSS_constant_k} for $k$-WECSS.

\begin{proof}[Proof of \Cref{thm:PTASs_k-WECSS_k-WVCSS_constant_k} for $k$-WECSS]
Let $G=(V,E)$ be a planar graph with edge costs $c\colon E \to \mathbb{R}_{\geq 0}$ of bounded cost ratio $\Delta$, and let $\varepsilon >0$.
As pointed out before, we can check in linear time whether $G$ is $k$-edge-connected~\cite{Eppstein1999}, so we will assume this in the following.
As discussed, we can assume without changing the optimal solution that $G$ has at most $k$ parallel edges between any pair of vertices.
This can be achieved by removing all edges that are not among the $k$ cheapest edges between any pair of vertices.

Using \Cref{theorem:compute_edge_safe_cover}, we compute in linear time a $k$-edge-safe cover $\mathcal{U}=(U_i)_{i\in I}$ of $G$, as well as the collection of graphs $G_i\coloneqq (V,E)/E[V\setminus U_i]$ for $i\in I$. 
We define edge costs $c_i$ on $G_i$ that we obtain from $c$ by setting the costs of all edges in $E_{\mathcal{U}}$ and in $E(G_i)\setminus E[U_i]$ to zero.
\Cref{theorem:compute_edge_safe_cover} guarantees that each $G_i$ has treewidth $\mathcal{O}(\frac{k}{\varepsilon})$, which is constant for fixed $k$.
Hence, we can use a linear time dynamic programming procedure (\Cref{thm:bounded_treewidth_kWECSS}) to compute an optimum $k$-WECSS solution $\OPT_i$ in each $(G_i,c_i)$.
By \Cref{theorem:compute_edge_safe_cover}~\ref{item:inAtMostTwo}, the sum of the sizes of all $G_i$ is linearly bounded in the size of $G$.
Hence, computing the cost functions $c_i$, as well as $k$-WECSS solutions for all $G_i$'s takes, altogether, linear time in the size of $G$.
Finally, we return $F\coloneqq E_{\mathcal{U}}\cup \bigcup_{i\in I} (\OPT_i\cap E[U_i])$.
By \Cref{lemma:cost_combined_solution}, $F$ is a $k$-WECSS solution in $G$ of cost $c(F)\le c(\OPT) + c(E_{\mathcal{U}})$.
Moreover, \Cref{theorem:compute_edge_safe_cover}~\ref{k-edge-safe_cover_costs} implies $c(E_{\mathcal{U}}) \leq \varepsilon \cdot c(E)$, and thus $c(F)\le (1+6\varepsilon \Delta)\cdot c(\OPT)$ by \Cref{lem:costs_of_all_edges}.
For any $\epsilon' > 0$, we can thus achieve a $(1+\epsilon')$-approximation by setting $\epsilon \coloneqq \frac{\epsilon'}{6\Delta}$.

\end{proof}

We further derive \cref{cor:PTASs_k-WECSS}.
\begin{proof}[Proof of \cref{cor:PTASs_k-WECSS}]
As we can check in polynomial time whether the instance is feasible, i.e., whether the input graph $G=(V,E)$ is $k$-edge-connected, we will assume $G$ to be $k$-edge-connected in the following.
We first observe how to obtain a $(1+\frac{6\Delta}{k})$-approximation for planar $k$-WECSS with cost ratio $\Delta$.
To this end, one can first compute an optimal vertex solution to the natural linear programming relaxation, and then simply round up all fractional values.
Because $G$ is planar and, among any set of parallel edges, there is at most one with fractional LP value, the total rounding cost is upper bounded by $3|V|\cdot c_{max}$.
Since an optimal solution needs at least $\frac{k}{2}|V|$ many edges, whose cost is thus at least $\frac{k}{2}|V|\cdot c_{min}$, the rounding cost is no more than $\frac{6\Delta}{k} \cdot c(\OPT)$, thus leading to the claimed $(1+\frac{6\Delta}{k})$-approximation for planar $k$-WECSS with cost ratio $\Delta$.%
\footnote{A $(1+\mathcal{O}(\frac{\Delta}{k}))$-approximation for general (not necessarily planar) $k$-WECSS with cost ratio $\Delta$ also follows from known techniques on $k$-ECSS~\cite{gabowApproximatingSmallestIt2009,gabowIteratedRoundingAlgorithms2012}.
	Moreover, the constant of the term $\frac{6\Delta}{k}$ in the approach we showed can be slightly improved by exploiting that vertex solutions to the LP are sparse.
}

Hence, if $k \geq \frac{6\Delta}{\varepsilon}$, we can run the above-mentioned $(1+\frac{6\Delta}{k})$-approximation to get a $(1+\varepsilon)$-approximation.
Otherwise, we have $k \leq \frac{6\Delta}{\varepsilon}$, i.e., $k$ is constant.
In this case, \cref{thm:PTASs_k-WECSS_k-WVCSS_constant_k} concludes the proof.
\end{proof}

\subsection{Vertex Connectivity}

In this section, we prove \Cref{thm:PTASs_k-WECSS_k-WVCSS_constant_k} for $k$-WVCSS.
First, we provide a useful characterization of $k$-vertex-connectivity. Recall that a graph $G$ is $k$-vertex-connected if and only if it has at least $k+1$ vertices and for every subset $X$ of the vertices of size at most $k-1$, $G-X$ is connected.
In the following, we will assume that the input graphs $G$ we deal with have at least $k+1$ vertices; this property can be checked in linear time.

\begin{lemma}\label{lem:characterization_k_vertex_connectivity}
 Let $G=(V,E)$ be a connected graph with at least $k+1$ vertices. 
 Then $G$ is $k$-vertex-connected if and only if $|\Gamma(S)|\ge \min\{k,|V\setminus S|\}$ for every connected cut $\emptyset\ne S\subsetneq V$.
\end{lemma}
\begin{proof}
If $G$ is $k$-vertex-connected, then for every connected cut $S$ with $\Gamma(S)\subsetneq V\setminus S$, we must have $|\Gamma(S)|\ge k$ because $G-\Gamma(S)$ is disconnected.

 Next, assume that $G$ is not $k$-vertex-connected and let $X\subseteq V$ be of minimum cardinality such that $G-X$ is disconnected. 
 Then, $G-X$ has at least $2$ connected components and because $G$ is not $k$-vertex-connected, $|X|\le k-1$. 
Moreover, because $G$ is connected, $|X|\geq 1$.
 Let $V_1,\dots,V_\ell$ be the vertex sets of the connected components of $G-X$. 
 Then $\Gamma(V_i)\subseteq X$ for each $i\in[\ell]$ and by minimality of $X$, we actually have $\Gamma(V_i)=X$ for each $i\in [\ell]$. 
 Hence, $G[V_1]$ and $G[V\setminus V_1]=G[X\cup V_2\cup\dots\cup V_\ell]$ are connected, $\emptyset\ne V_1\subsetneq V$, and $|\Gamma(V_1)|=|X|<\min\{k,|V\setminus V_1|\}$.   
 
\end{proof}

A $k$-vertex-connected graph will not necessarily remain $k$-vertex-connected when we contract some edges, which is in contrast to edge-connectivity.
This means that when we contract some part of our instance, the resulting instance might not even be feasible anymore.
In order to handle this issue, we introduce a different contraction operation, which preserves $k$-vertex-connectivity.

\begin{definition}[$k$-vertex-safe contraction]
Let $G=(V,E)$ be a graph, let $U\subsetneq V$ and let $F\subseteq E$. The $k$-vertex-safe contraction $(V,F)\cvs G[U]$ of $U$ is the graph that arises from $(V,F)$ by 
\begin{itemize}
\item contracting every connected component of $G[U]$ into a single vertex and then
\item replacing each of the contracted vertices by a clique of size $k$, where each vertex of the clique is connected to each neighbor of the original, contracted vertex.
\end{itemize}
\end{definition}

Note that this operation does not preserve planarity, but this will not be an issue for our approach, as we will still be able to bound the treewidth of the resulting graph.
First, we show that $k$-vertex-safe contractions indeed preserve $k$-vertex-connectivity.

\begin{lemma}\label{lemma:safe_contraction}
Let $G=(V,E)$ be a graph and let $F\subseteq E$ such that $(V,F)$ is $k$-vertex-connected. Let further $U\subsetneq V$ and $G'=(V',F')\coloneqq (V,F)\cvs G[U]$. Then $G'$ is $k$-vertex-connected.
\end{lemma}
\begin{proof}
As $(V,F)$ is $k$-vertex-connected, $|V|\ge k+1$. 
Moreover, because $U\subsetneq V$, we have $|V'|\ge k+1$. 
To show that $G'$ is $k$-vertex-connected, we need to prove that for every $X\subseteq V'$ with $|X|\le k-1$, the graph $G'-X$ is connected. 

Let $X\subseteq V'$ with $|X|\le k-1$ and let $Y\subseteq X\cap (V\setminus U)$ be the set of ``non-clique-vertices'' in $X$. 
Then $|Y|\le k-1$ and as $(V,F)$ is $k$-vertex-connected, $(V,F)-Y$ is connected. 
Given that the cliques that we added for the connected components of $G[U]$ contain $k$ vertices each, at least one vertex per clique is not contained in $X$. 
Hence, $G'-X$ arises from the connected graph $(V,F)-Y$ by contracting the connected components of $G[U]$ and replacing each connected component by a nonempty clique (which may, however, contain fewer than $k$ vertices). 
As none of these operations destroys connectivity, $G'-X$ is indeed connected.

\end{proof}

Next, we observe that the treewidth of the graph resulting from a $k$-vertex-safe contraction cannot be much larger than the treewidth of the graph that we would obtain from the standard (not $k$-vertex-safe) contraction of edges.

\begin{lemma}\label{lemma:contraction_treewidth}
Let $G=(V,E)$ be a graph and let $U\subsetneq V$. Then 
\[
\mathrm{treewidth}(G\cvs G[U])\le k\cdot \mathrm{treewidth}(G/E[U])+k-1.
\]
\end{lemma}
\begin{proof}
Given a tree decomposition for the graph $G/E[U]$ with maximum bag size $\mathrm{treewidth}(G / E[U])+1$, we obtain a tree decomposition for $G\cvs G[U]$ by replacing each contracted vertex with its $k$ copies. In doing so, the maximum bag size increases to at most $k\cdot\mathrm{treewidth}(G\cvs G[U])+k$.

\end{proof}

We now show that, analogously to the edge-connectivity setting, buying the edges in $E_{\mathcal{U}}$ decomposes the problem into independent parts.
To this end, we prove a statement analogous to \Cref{lemma:glue_solutions_together}:

\begin{lemma}\label{lemma:glue_solutions_together_vertex}
Let $G=(V,E)$ be a graph, and let $\mathcal{U}=(U_i)_{i\in I}$ be a $k$-vertex-safe cover of $G$. Let $(F_i)_{i\in I}$ be edge sets such that $(V,F_i)\cvs G[V\setminus U_i]$ is $k$-vertex-connected. 
Let $F\coloneqq E_\mathcal{U}\cup \bigcup_{i\in I} \big( F_i\cap E[U_i] \big)$. Then $(V,F)$ is $k$-vertex-connected.
\end{lemma}
\begin{proof}
We first show that $(V,F)$ is connected. To this end, we note that $\mathcal{U}$ is a $1$-edge-safe cover and moreover, the fact that $(V,F_i)\cvs G[V\setminus U_i]$ is $k$-vertex-connected for every $i\in I$ implies that $(V,F_i)/ E[V\setminus U_i]$ is connected (i.e., $1$-edge-connected) for every $i\in I$. By \cref{lemma:glue_solutions_together}, $(V,F)$ is connected.

By \Cref{lem:characterization_k_vertex_connectivity}, it suffices to show that $|\Gamma_{(V,F)}(S)|\ge \min\{k,|V\setminus S|\}$ for every connected cut $\emptyset\ne S\subsetneq V$ of $(V,F)$.
Let $S$ be such a connected cut. Then $S$ is also a connected cut of $G$.
If there is no $i\in I$ such that $\delta(S)\subseteq E[U_i]$, then 
\[|\Gamma_{(V,F)}(S)|\ge |\Gamma_{E_{\mathcal{U}}}(S)|\ge k\ge\min\{k,|V\setminus S|\}.\]
Otherwise, there is an $i\in I$ such that $\delta(S)\subseteq E[U_i]$. 
Consider the graphs $G'_i\coloneqq (V,F_i)\cvs G[V\setminus U_i]$ and $G_i\coloneqq G\cvs G[V\setminus U_i]$, and observe that $G'_i$ is a spanning subgraph of $G_i$ because $(V,F_i)$ is a spanning subgraph of $G$.
Note that $S$ corresponds to a connected cut $S'$ in the graph $G'_i$ with $\delta_{G'_i}(S')\subseteq F_i\cap E[U_i]$. 
Moreover, all neighbors of $S'$ in the graph $G_i$ and in the graph $G'_i$ are contained in $U_i$. 
Hence, 
\[
|\Gamma_{(V,F)}(S)|\ge |\Gamma_{G'_i}(S')|\ge \min\{k,|V(G'_i)\setminus S'|\},
\]
where the second inequality holds by \Cref{lem:characterization_k_vertex_connectivity} and because $G'_i$ is $k$-vertex-connected by assumption.
If the minimum is attained by $k$, we get $|\Gamma_{(V,F)}(S)|\ge k\ge\min\{k,|V\setminus S|\}$. In case the minimum is attained by $|V(G'_i)\setminus S'|$, we know that $\Gamma_{G'_i}(S') = V(G'_i)\setminus S'$.
Using that $G'_i$ is a spanning subgraph of $G_i$, this implies
\[
V(G'_i)\setminus S'= \Gamma_{G'_i}(S')\subseteq \Gamma_{G_i}(S')\subseteq V(G_i)\setminus S' = V(G'_i)\setminus S',
\] 
implying $ |\Gamma_{G'_i}(S')| = |\Gamma_{G_i}(S')|$.
We conclude
\[
|\Gamma_{(V,F)}(S)| \ge |\Gamma_{G'_i}(S')| = |\Gamma_{G_i}(S')| = |\Gamma_G(S)| \ge \min\{k,|V\setminus S|\},
\] 
where the first inequality and the second equality follow from $\delta(S)\subseteq E[U_i]$, and the last inequality follows from the fact that $G$ is $k$-vertex-connected.

\end{proof}

Moreover, again analogously to the edge-connectivity setting (see \Cref{lemma:cost_combined_solution}), we can bound the total cost of optimum solutions to our independent subinstances.

\begin{lemma}\label{lemma:cost_combined_solution_vertex}
Let $G=(V,E)$ be a $k$-vertex-connected graph with edge costs $c\colon E \to \mathbb{R}_{\geq 0}$, and let $\mathcal{U}=(U_i)_{i\in I}$ be a $k$-vertex-safe cover of $G$.
For $i\in I$, let $G_i\coloneqq (V,E)\cvs G[V\setminus U_i]$ and let $c_i$ be the cost function that assigns costs of zero to all edges that are incident to a clique vertex or contained in $E_{\mathcal{U}}$, and a cost of $c(e)$ to every other edge $e$. Let $\OPT_i$ be an optimum $k$-WVCSS solution for $(G_i,c_i)$.

Then, $F\coloneqq E_\mathcal{U}\cup \bigcup_{i\in I} \big(\OPT_i\cap E[U_i]\big)$ is a solution to the $k$-WVCSS in $(G,C)$ of cost 
\[
c(F) \leq c(\OPT) + c(E_{\mathcal{U}}),
\]
where $\OPT$ is an optimal $k$-WVCSS solution of $(G,c)$.
\end{lemma}
\begin{proof}
By \Cref{lemma:glue_solutions_together_vertex}, $F$ is a feasible solution to $k$-WVCSS in $(G,c)$. \Cref{lemma:safe_contraction} tells us that $(V,\OPT)\cvs G[V\setminus U_i]$ is $k$-vertex-connected. In particular, the edges in $\OPT\cap E[U_i]$, together with all edges incident to clique vertices in $G_i$, form a feasible solution $F_i$ to the $k$-WVCSS instance $(G_i,c_i)$. This implies $c_i(\OPT_i)\le c_i(F_i)=c(\OPT\cap E[U_i]\setminus E_{\mathcal{U}})$.
Consequently,
 \begin{equation*}
     \sum_{i\in I}c_i(\OPT_i) \leq \sum_{i\in I} c(\OPT\cap E[U_i]\setminus E_{\mathcal{U}}) \le c(\OPT),
 \end{equation*}
because $E[U_i]\cap E[U_j] = E[U_i\cap U_j]\subseteq E_{\mathcal{U}}$ for $i\ne j\in I$. 
We conclude
\begin{align*}
c(F)\ =&\ c(F\setminus E_\mathcal{U})+c(E_\mathcal{U}) \ \le\  \sum_{i=1} c(\OPT_i\cap E[U_i]\setminus E_{\mathcal{U}})+c(E_\mathcal{U}) \\
 =&\  \sum_{i\in I} c_i(\OPT_i)+c(E_\mathcal{U})\ \le\ c(\OPT)+c(E_\mathcal{U}).
\end{align*}

\end{proof}

Finally, we obtain a PTAS for planar $k$-WVCSS for bounded cost ratio instances:

\begin{proof}[Proof of \Cref{thm:PTASs_k-WECSS_k-WVCSS_constant_k} for $k$-WVCSS]
Let $G=(V,E)$ be a planar $k$-vertex-connected graph with edge costs $c\colon E \to \mathbb{R}_{\geq 0}$ of bounded cost ratio, and let $\varepsilon >0$.
As we can check whether $G$ is $k$-vertex-connected in linear time~\cite{Eppstein1999}, we will assume that this is the case in the following.
We may further assume that $G$ is a simple graph because adding parallel edges cannot increase the vertex-connectivity of a graph.
If this assumption is not satisfied, we remove all but the cheapest edge between any pair of vertices, which we can do in linear time.

Using \Cref{theorem:compute_k_vertex_safe_cover}, we compute in linear time a $k$-vertex-safe cover $\mathcal{U}=(U_i)_{i\in I}$ of $G$, as well as each of the graphs $G_i\coloneqq (V,E) \cvs G[V\setminus U_i]$ for $i\in I$. Note that $G_i$ can be computed from $(V,E)/E[V\setminus U_i]$ in time linear in the size of $G_i$.We further define cost functions $c_i$ from $c$ by setting the costs of all edges in $E_{\mathcal{U}}$ and in $E(G_i)\setminus E[U_i]$ to zero.
\Cref{theorem:compute_k_vertex_safe_cover} and \Cref{lemma:contraction_treewidth} guarantee that each $G_i$ has treewidth $\mathcal{O}(\frac{k^2}{\varepsilon})$, which is constant for fixed $k$.
Hence, by \cref{thm:bounded_treewidth_kWECSS}, we can use a linear time dynamic programming procedure to compute an optimum $k$-WVCSS solution $\OPT_i$ in each $(G_i,c_i)$.
By \Cref{theorem:compute_k_vertex_safe_cover}~\ref{item:inAtMostTwo}, the sum of the sizes of all $G_i$ is linearly bounded in the size of $G$.
Hence, computing $k$-WVCSS solutions for all $G_i$'s takes, altogether, linear time in the size of $G$.
Finally, we return $F\coloneqq E_{\mathcal{U}}\cup \bigcup_{i\in I} (\OPT_i\cap E[U_i])$.
By \Cref{lemma:cost_combined_solution_vertex}, $F$ is a $k$-WVCSS solution in $G$ of cost $c(F)\le c(\OPT) + c(E_{\mathcal{U}})$.

By \Cref{theorem:compute_k_vertex_safe_cover}~\ref{k-vertex_safe_edge_costs} and \Cref{lem:costs_of_all_edges} (using that $k$-vertex-connectivity implies $k$-edge-connectivity), we have 
\[
c(E_{\mathcal{U}})\le \epsilon\cdot c(E)\le \epsilon\cdot 6\Delta \cdot c(\mathrm{OPT}).
\] 
For any $\epsilon' > 0$, we can thus achieve a $(1+\epsilon')$-approximation by setting $\epsilon \coloneqq \frac{\epsilon'}{6\Delta}$.

\end{proof}
\section{A PTAS for planar $k$-WCAP with bounded cost ratio}\label{sec:augmentation}
In this section, we expand on the ideas given in \cref{sec:augmentation_short} and provide a full proof of \cref{theorem:PTAS_WCAP}. The remainder of this section is organized as follows: In \cref{sec:snug_vertices_and_chains}, we provide the deferred proofs of \cref{lemma:chain_graph_paths,lemma:chain_graph_subgraph,lemma:only_cut_intersecting_snug_edge} and prove further structural results about snug paths and snug shores that will be crucial to develop our PTAS for planar $k$-WCAP with bounded cost ratio.
In \cref{sec:augmentation_safe_covers}, we formally introduce the notion of a $k$-augmentation-safe cover and discuss the interplay between $(k+1)$-edge-safe covers, snug path contractions, and $k$-augmentation-safe covers. \cref{sec:thinning} provides the tools that we need to thin out the link set and bound the cost of our solution. Finally, we prove \cref{theorem:PTAS_key_theorem} and \cref{theorem:PTAS_WCAP} in \cref{sec:PTAS_WCAP}.

For convenience, we will assume throughout this section that a given instance $(G,L,c)$ of (planar) $k$-WCAP is feasible, i.e., $G+L$ is $(k+1)$-edge-connected. Note that we can check for (in)feasibility in polynomial time.
\subsection{Snug vertices and snug paths\label{sec:snug_vertices_and_chains}}
For this section, let $k\in \mathbb{Z}_{> 0}$, let $G=(V,E)$ be a $k$-edge-connected graph, and let $r\in V$ be a fixed root.
We recap the definitions of snug vertices and snug shores from \cref{sec:augmentation_short}.
\begin{definition}[snug vertex]
	A vertex $v\in V\setminus\{r\}$ is called \emph{snug} if $\mathrm{deg}(v)\ge k+1$ and there exist $k$-cuts $S_1\subseteq S_2\subseteq V\setminus\{r\}$ with $S_2\setminus S_1=\{v\}$. We denote the set of snug vertices by $V_{\rm{snug}}$.
\end{definition}

\begin{definition}
	Let $u\in V$ be a snug vertex.
	A pair $(S_1, S_2)$ of $k$-cuts is called \emph{snug shores for $u$} if $S_2 = S_1\cup \{u\}$ and $u\not\in S_1$. 
\end{definition}
We begin with \cref{lemma:interaction_with_snug_shores} below, which tells us that no $k$-cut $S\subseteq V\setminus\{r\}$ can cross snug shores. It helps us to prove the uniqueness of snug shores (\cref{lemma:snug_shores_unique}). \cref{lemma:interaction_with_snug_shores} can further be used to establish structural properties of the chain graph and prove \cref{lemma:only_cut_intersecting_snug_edge}.
\begin{lemma}\label{lemma:interaction_with_snug_shores}
	Let $(S_1,S_2)$ be snug shores for $u\in V$ and let $S\subseteq V\setminus \{r\}$, $S\neq \emptyset$ be a $k$-cut.  Then $S\subseteq S_1$, $S_2\subseteq S$, or $S_2\cap S=\emptyset$.
\end{lemma}
\begin{proof}
	If $u\in S$, but $S_2\not\subseteq S$, we have $S_1\setminus S\ne \emptyset$. As $u\in S\cap S_2$ and $r\notin S\cup S_2$, $Q\coloneqq S\cap S_2$ is a $k$-cut. We have $S_1\setminus Q\ne \emptyset$ and $Q\setminus S_1=\{u\}$, so both must be $k$-cuts. This implies $\deg(u)=k$, a contradiction.
	
	If $u\notin S$, but $S\not\subseteq S_1$ and $S_2\cap S\ne \emptyset$, then $S\setminus S_2\ne\emptyset$ and $S_1\cap S\ne \emptyset$. As $r\notin S\cup S_1$, $Q\coloneqq S\cup S_1$ is a $k$-cut. We have $Q\setminus S_2\ne \emptyset$ and $S_2\setminus Q=\{u\}$, so both must be $k$-cuts. But again, this yields $\deg(u)=k$.
    
\end{proof}
\begin{lemma}\label{lemma:snug_shores_unique}%
	Let $u$ be a snug vertex. Then there are unique snug shores for $u$.
\end{lemma}

\begin{proof}
	Let $(S_1,S_2)$ and $(T_1,T_2)$ be snug shores for $u$. As $u\in S_2\cap T_2$, but $u\notin S_1\cup T_1$, \Cref{lemma:interaction_with_snug_shores} implies $S_2\subseteq T_2$ (when applied to $(S_1,S_2)$ and $T_2$) and $T_2\subseteq S_2$ (when applied to $(T_1,T_2)$ and $S_2$), so $S_2=T_2$.
	Because $T_1 = T_2 \setminus \{u\}$, as $(T_1,T_2)$ are snug shores for $u$, we also have $T_1=S_2\setminus \{u\}=S_1$.

\end{proof}

Recall the definition of the chain graph from \cref{sec:augmentation_short}.

\begin{definition}[chain graph]
	The \emph{chain graph} $G_{\rm{chain}}=(V_{\rm{snug}},A)$ contains an arc $(u,v)$ if and only if $u$ and $v$ are snug vertices with shores $(S^u_1,S^u_2)$ and $(S^v_1,S^v_2)$, respectively, with $S^u_2=S^v_1$.
\end{definition}
\cref{lemma:chain_graph_paths}, which we restate here for convenience, follows from \cref{prop:at_most_one_outgoing,prop:at_most_one_incoming,prop:no_directed_cycle}, which we state and prove next.
\lemmachaingraphpaths*

\begin{proposition}\label{prop:at_most_one_outgoing}
	For every $k$-cut $S\subseteq V\setminus \{r\}$, there is at most one vertex $v\in V\setminus S$ such that $(S,S\cup \{v\})$ form snug shores for $v$.
	In particular, each vertex has at most one outgoing arc in $G_{\rm{chain}}$. 
\end{proposition}
\begin{proof}
	Assume towards a contradiction that $(S,S\cup\{v\})$ are the snug shores for $v$ and $(S,S\cup\{v'\})$ are the snug shores for $v'\ne v$.
  However, this contradicts \cref{lemma:interaction_with_snug_shores} for the snug shores $(S,S\cup \{v\})$ and the $k$-cut $S\cup\{v'\}$.
	Now, for a snug vertex $u$ with snug shores $(S_1,S_2)$, there is at most one vertex $v\in V\setminus S_2$ such that $(S_2,S_2\cup\{v\})$ form snug shores for $v$. Hence, $u$ has at most one outgoing arc in $G_{\rm{chain}}$.
\end{proof}
\begin{proposition}\label{prop:at_most_one_incoming}
	For every $k$-cut $S\subseteq V\setminus \{r\}$, there is at most one vertex $v\in S$ such that $(S\setminus \{v\},S)$ form snug shores for $v$.
	In particular, each vertex has at most one incoming arc in $G_{\rm{chain}}$. 
\end{proposition}
\begin{proof}
	Assume towards a contradiction that $(S\setminus \{v\},S)$ are the snug shores for $v$ and $(S\setminus \{v'\},S)$ are the snug shores for $v'\ne v$.
  However, this contradicts \cref{lemma:interaction_with_snug_shores} for the snug shores $(S\setminus \{v\},S)$ and the $k$-cut $S\setminus \{v'\}$.
	Now, for a snug vertex $u$ with snug shores $(S_1,S_2)$, there is at most one vertex $v\in S_1$ such that $(S_1\setminus\{v\},S_1)$ are snug shores for $v$. Hence, $u$ has at most one incoming arc in $G_{\rm{chain}}$.
\end{proof}
\begin{proposition}\label{prop:no_directed_cycle}
	$G_{\rm{chain}}$ does not contain any directed cycle.   
\end{proposition}
\begin{proof}
	Assume towards a contradiction that $u_1,(u_1,u_2),u_2,\dots,u_k,(u_k,u_1),u_1$ were a directed cycle. Then there are $k$-cuts $S_1,\dots,S_k$ such that $(S_i,S_{i+1})$ are the snug shores for $u_i$, where $S_{k+1}\coloneqq S_1$. But then $S_1\subsetneq S_2\subsetneq \dots\subsetneq S_k\subsetneq S_1$, a contradiction.
\end{proof}

We call the maximal paths in $G_{\rm{chain}}$ the \emph{snug paths of $G$} and denote the family of snug paths by $\mathcal{P}^G_{\rm{chain}}$. 
For a snug path $P=(u_0,\dots,u_t)$, the \emph{snug shores for $P$} are the $k$-cuts $(S_0,\dots,S_{t+1})$, where $(S_i,S_{i+1})$ are the snug shores for $u_i$ for $i\in\{0,\dots,t\}$. We call $(S_0,\dots,S_{t+1})$ a \emph{snug chain}.

\Cref{prop:at_most_one_outgoing,prop:at_most_one_incoming} imply \cref{lemma:shore_for_at_most_one_path}.
\begin{lemma}
	A $k$-cut $S$ is a snug shore for at most one snug path.\label{lemma:shore_for_at_most_one_path}
\end{lemma}

As remarked in \cref{sec:augmentation_short}, we can efficiently compute the chain graph and all snug shores:

\begin{restatable}{lemma}{lemmacomputechaingraph}\label{lemma:compute_chain_graph}
	We can compute $G_{\rm{chain}}$, as well as the snug shores for every snug path, in polynomial time.
\end{restatable} 
\begin{proof}
	We can enumerate all minimum cuts of $G$ in polynomial time (see, e.g., \cite{doi:10.1137/S0895480194271323}). This allows us to determine the set of snug vertices, their snug shores, and the chain graph in polynomial time.
\end{proof}
Next, we prove \cref{lemma:chain_graph_subgraph}, which we restate here for convenience.
\lemmachaingraphsubgraph*

\begin{proof}
	Let $(u,v)$ be an arc in $G_{\rm{chain}}$ and let $(S_1^u,S_2^u)$ and $(S_1^v,S_2^v)$ with $S_2^u=S_1^v$ be the snug shores for $u$ and $v$, respectively. Assume towards a contradiction that $u$ and $v$ do not share an edge in $G$. Then $\delta_G(u)\subseteq \delta_G(S_1^u)\cup \delta_G(S_2^v)$ because for $e=\{u,w\}\in \delta_G(u)$, either $w\in S_1^u$ and $e\in \delta_G(S_1^u)$, or $w\in V\setminus (S_1^u\cup\{u,v\})=V\setminus S_2^v$, implying $e\in \delta_G(S_2^v)$. Similarly, $\delta_G(v)\subseteq \delta_G(S_1^u)\cup\delta_G(S_2^v)$. Hence,
	\[2\cdot(k+1)\le |\delta_G(u)|+|\delta_G(v)|=|\delta_G(u)\cup \delta_G(v)|\le |\delta_G(S_1^u)|+|\delta_G(S_2^v)|\le 2k,\] a contradiction.
\end{proof}
 In a slight abuse of notation, for a collection $\mathcal{Q}\subseteq \mathcal{P}^G_{\rm{chain}}$, we denote by $G/\mathcal{Q}$ the graph that arises from $G$ by contracting the undirected paths in $G$ corresponding to $\mathcal{Q}$. Given a set of links $L$ and links costs $c\colon L\rightarrow\mathbb{R}_{>0}$, we further write $L/\mathcal{Q}$ for the (multi)set of links arising from the contraction, $c/\mathcal{Q}$ for the induced cost function, $(G,L)/\mathcal{Q}\coloneqq (G/\mathcal{Q},L/\mathcal{Q})$, and $(G,L,c)/\mathcal{Q}\coloneqq (G/\mathcal{Q},L/\mathcal{Q},c/\mathcal{Q})$. By \cref{lemma:chain_graph_subgraph}, if $G+L$ is planar, then so is $(G+L)/\mathcal{Q}=G/\mathcal{Q}+L/\mathcal{Q}$.

We further establish \cref{lemma:only_cut_intersecting_snug_edge}, which tells us that contracting snug paths does not destroy any $k$-cuts except for snug shores.
\lemmaonlycutintersectingsnugedge*
\begin{proof}
	Let $S$ be a $k$-cut in $G$ with $(u,v)\in \delta^+(S)$.
  By \cref{lemma:interaction_with_snug_shores} applied to the snug shores $(S_1^u,S_2^u)$ for $u$ and the $k$-cut $S$, we get $S_2^u \subseteq S$.
 Similarly, by applying \cref{lemma:interaction_with_snug_shores} to the snug shores $(S_1^v,S_2^v)$ for $v$ and the $k$-cut $S$, we get $S\subseteq S_1^v = S_2^u$.
 Thus, $S=S_2^u$, as claimed.
\end{proof}

A key property that we will exploit to handle snug chains is that, as long as we make sure to cover every cut of a snug chain by some link, we can focus on the instance that arises by contraction of the corresponding snug path. 

\begin{definition}
	We say that a set $F$ of links \emph{covers} a snug path $P\in\mathcal{P}^G_{\rm{chain}}$ if for every snug shore $S$ of~$P$, we have $\delta_F(S)\ne \emptyset$. We say that $F$ covers a collection $\mathcal{Q}\subseteq \mathcal{P}^G_{\rm{chain}}$ if $F$ covers every snug path in $\mathcal{Q}$.
\end{definition}

\begin{lemma}\label{lemma:solution_contract_snug_chains}
	Let $\mathcal{Q}\subseteq \mathcal{P}^G_{\rm{chain}}$.
	Let $F\subseteq L$ cover $\mathcal{Q}$ and let $(V,E\dot{\cup} F)/\mathcal{Q}$ be $(k+1)$-edge-connected. Then $(V,E\dot{\cup} F)$ is $(k+1)$-edge-connected.
\end{lemma}
\begin{proof}
	We need to show that for every $k$-cut $S$, $\delta_F(S)\ne \emptyset$. This is clear if $S$ is a snug shore (or its complement) for a snug path in $\mathcal{Q}$. On the other hand, if $S$ is not of this form, then by \Cref{lemma:only_cut_intersecting_snug_edge}, $S$ does not cross any snug path in $\mathcal{Q}$, so $S$ corresponds to a $k$-cut in $G/\mathcal{Q}$.
\end{proof}
\cref{lemma:link_furthest_out} and \cref{lemma:link_deepest_in} will be useful to thin out the link set in \cref{sec:thinning}.
\begin{lemma}\label{lemma:link_furthest_out}
	Let $P=(u_0,\dots,u_t)$ be a snug path with snug shores $(S_0,\dots,S_{t+1})$. Let $v\in S_0$ and let $\ell=\{v,u_i\}$ and $\ell'=\{v,u_j\}$ be two links with $0\le i < j\le t$. Every $k$-cut in $G$ that is crossed by $\ell$ is also crossed by $\ell'$.
\end{lemma}
\begin{proof}
	By \Cref{lemma:only_cut_intersecting_snug_edge}, if $S$ is a $k$-cut with $\{S,V\setminus S\}\cap\{S_{i+1},\dots,S_j\}=\emptyset$, then $\ell$ covers $S$ if and only if $\ell'$ does because none of the edges on the $u_i$-$u_j$ subpath of $P$ crosses $S$, meaning that $u_i$ and $u_j$ are either both contained in $S$ or both contained in $V\setminus S$. Moreover, $\ell'$ crosses all of $S_{i+1},\dots,S_j$ and their complements.
\end{proof}
\begin{lemma}\label{lemma:link_deepest_in}
	Let $P=(u_0,\dots,u_t)$ be a snug path with snug shores $(S_0,\dots,S_{t+1})$. Let $v\in V\setminus S_{t+1}$ and let $\ell=\{v,u_i\}$ and $\ell'=\{v,u_j\}$ be two links with $0\le j < i\le t$. Every $k$-cut in $G$ that is crossed by $\ell$ is also crossed by $\ell'$.
\end{lemma}
\begin{proof}
	By \Cref{lemma:only_cut_intersecting_snug_edge}, if $S$ is a $k$-cut with $\{S,V\setminus S\}\cap\{S_{j+1},\dots,S_i\}=\emptyset$, then $\ell$ covers $S$ if and only if $\ell'$ does because none of the edges on the $u_j$-$u_i$ subpath of $P$ crosses $S$, meaning that $u_i$ and $u_j$ are either both contained in $S$ or both contained in $V\setminus S$. Moreover, $\ell'$ crosses all of $S_{j+1},\dots,S_i$ and their complements.
\end{proof}
\cref{lemma:link_crosses_outer_shore,prop:minimum_cut_connected,lemma:outer_cuts} will be needed when constructing a $k$-augmentation-safe cover in \cref{sec:augmentation_safe_covers}.
\begin{lemma}\label{lemma:link_crosses_outer_shore}
Let $P=(u_0,\dots,u_t)$ be a snug path with snug shores $(S_0,\dots,S_{t+1})$. Let $\ell$ be a link that has at least one endpoint outside $P$ that crosses some snug shore of $P$. Then $\ell$ crosses $S_0$ or $S_{t+1}$.
\end{lemma} 
\begin{proof}
Let $\ell=\{u,v\}$ with $v\in V\setminus \{u_0,\dots,u_t\}$. If $v\in S_0$, then we must have $u\in V\setminus S_0$ and $\ell$ crosses $S_0$. If $v\in V\setminus S_{t+1}$, then we must have $u\in S_{t+1}$ and $\ell$ crosses $S_{t+1}$.
\end{proof}
\begin{proposition}\label{prop:minimum_cut_connected}
	Let $S$ be a $k$-cut in $G$. Then $S$ is connected.
\end{proposition}
\begin{proof}
	As $S$ is a $k$-cut and $G$ is $k$-edge-connected, $S$ must be a minimal cut, i.e., there is no cut $T\subsetneq V, T\neq \emptyset$ with $\delta(T)\subsetneq \delta(S)$. The equivalence of minimal cuts and connected cuts yields the desired statement.
\end{proof}
\begin{lemma}\label{lemma:outer_cuts}
	Let $P=(u_0,\dots,u_t)$ be a snug path with snug shores $(S_0,\dots,S_{t+1})$. Then $S_0$ and $S_{t+1}$ are connected $k$-cuts in $G$ and no snug path crosses $S_0$ or $S_{t+1}$.
\end{lemma}
\begin{proof}
	$S_0$ and $S_{t+1}$ are connected cuts by \cref{prop:minimum_cut_connected}. The second part of the statement follows by maximality of snug paths and \cref{lemma:only_cut_intersecting_snug_edge}, using that no snug shore contains $r$.
\end{proof}
\subsection{Augmentation-safe covers}\label{sec:augmentation_safe_covers}
In this section, we formally introduce the notion of a $k$-augmentation-safe cover alluded to in \cref{sec:augmentation_short}.
\begin{definition}\label{def:augmentation_safe_cover}
	Let $G$ be a $k$-edge-connected graph and $F\subseteq \binom{V}{2}$ be a (multi)set of links. We call a collection $\mathcal{U}=(U_i)_{i\in I}$ of nonempty subsets of $V$ covering $V$ a \emph{$k$-augmentation-safe cover of $(G,F)$} if for every $k$-cut $S\subseteq V$ of $G$
	\begin{enumerate}[label=(\roman*)]
		\item \label{def:augmentation_safe:1} $\delta_{E\dot{\cup} F} (S)\subseteq (E\dot{\cup} F)[U_i]$ for some $i\in I$, or
		\item \label{def:augmentation_safe:2}$\delta_{F_{\mathcal{U}}}(S)\ne \emptyset$, where $F_{\mathcal{U}}$ denotes the links in $F$ that are incident to a vertex in $V_{\mathcal{U}}\coloneqq \bigcup_{i, j\in I: i \neq j} U_i\cap U_j$.
	\end{enumerate}
\end{definition}

\cref{lemma:glue_solutions_together_augmentation} tell us how, given a $k$-augmentation-safe cover, we can combine solutions to the subinstances to a solution to the original instance. It is analogous to \cref{lemma:glue_solutions_together}, except that we use $F_i$ instead of $F_i\cap L[U_i]$ to construct the final solution (because this suffices to prove \cref{theorem:PTAS_key_theorem}).
\begin{lemma}\label{lemma:glue_solutions_together_augmentation}
	Let $(G=(V,E),L,c)$ be an instance of $k$-WCAP and let $\mathcal{U}=(U_i)_{i\in I}$ be a $k$-augmentation-safe cover of $(G,L)$. 
	
	Let $(F_i)_{i\in I}$ be link sets such that $(G+F_i)/(E\dot{\cup} L)[V\setminus U_i]$ is $(k+1)$-edge-connected for $i\in I$.
	Let $F\coloneqq L_{\mathcal{U}}\cup \bigcup_{i\in I} F_i$.
	 Then $G+F$ is $(k+1)$-edge-connected.
\end{lemma}
\begin{proof}
	It suffices to prove that for every $k$-cut $S$ of $G$, $\delta_F(S)\ne \emptyset$. If there is $i\in I$ with $\delta_{E\dot{\cup} L} (S)\subseteq (E\dot{\cup} L)[U_i]$, then $S$ corresponds to a $k$-cut in $G/(E\dot{\cup} L)[V\setminus U_i]$, so $\delta_{F_i}(S)\ne \emptyset$.
	On the other hand, if case \ref{def:augmentation_safe:2} from \cref{def:augmentation_safe_cover} applies, then $\delta_F(S)\supseteq \delta_{L_{\mathcal{U}}}(S)\ne \emptyset$.
\end{proof}
\cref{lemma:edge_safe_is_augmentation_safe} shows that $(k+1)$-edge-safe covers are $k$-augmentation-safe covers.
Thus we can obtain a $k$-augmentation-safe cover by constructing a $(k+1)$-edge-safe cover.
Nevertheless, it is helpful to use the notion of $k$-augmentation-safe covers, as it cleanly highlights the properties we need in this context.
\begin{lemma}\label{lemma:edge_safe_is_augmentation_safe}
Let $(G=(V,E),L,c)$ be an instance of $k$-WCAP and let $\mathcal{U}=(U_i)_{i\in I}$ be a $(k+1)$-edge-safe cover of $G+L$. Then $\mathcal{U}$ is a $k$-augmentation-safe cover of $(G,L)$. 
\end{lemma}
\begin{proof}
Let $S$ be a $k$-cut of $G$. By \cref{prop:minimum_cut_connected}, $S$ is connected in $G$, and, hence, also in $G+L$. As $\mathcal{U}$ is a $(k+1)$-edge-safe cover of $G+L$, we either have $\delta_{E\dot{\cup} L}(S)\subseteq (E\dot{\cup}L)[U_i]$ for some $i\in I$, or $|\delta_{E_{\mathcal{U}}}(S)|+|\delta_{L_\mathcal{U}}(S)|\ge k+1$. As $S$ is a $k$-cut in $G$, we must have $\delta_{L_{\mathcal{U}}}(S)\ne \emptyset$ in the latter case.
\end{proof}
For technical reasons, we will require that the $k$-augmentation-safe covers we work with satisfy another condition, which we call being \emph{well-separated}. By \cref{theorem:compute_k_vertex_safe_cover_with_vertex_weights} and \cref{lemma:edge_safe_is_augmentation_safe}, we can, in polynomial time, compute $k$-augmentation-safe covers with this property.
\begin{definition}
Let $\mathcal{U}=(U_i)_{i\in I}$ be a $k$-augmentation-safe cover of $(G,F)$, where $G$ is a $k$-edge-connected graph and $F$ is a set of links. We call $\mathcal{U}$ \emph{well-separated} if for every $i,j\in I$, $i\ne j$, there is no edge or link with one endpoint in $U_i\setminus U_j$ and the other endpoint in $U_j\setminus U_i$.
\end{definition}
\subsubsection{Lifting and contracting covers}
When contracting and uncontracting snug paths, we will transform a $k$-augmen\-ta\-tion-safe cover in the contracted instance into a $k$-augmentation-safe cover in the uncontracted instance, and vice versa. For this, we require the following notation.
\begin{definition}
Let $G$ be a $k$-edge-connected graph and let $\mathcal{Q}\subseteq \mathcal{P}^G_{\rm{chain}}$.
\begin{enumerate}[(i)]
	\item Let $\mathcal{U}'=(U'_i)_{i\in I}$ be a collection of subsets of $V(G/\mathcal{Q})$. The \emph{lift of $\mathcal{U}'$ to $G$} is the collection $\mathcal{U}=(U_i)_{i\in I}$, where $U_i$ is the preimage of $U'_i$ under the map $V(G)\rightarrow V(G/\mathcal{Q})$.
	\item Let $\mathcal{U}=(U_i)_{i\in I}$ be a collection of subsets of $V(G)$. The \emph{contraction of $\mathcal{U}$} is the collection $\mathcal{U}/\mathcal{Q}=(U'_i)_{i\in I}$, where $U'_i$ is the image of $U_i$ under the map $V(G)\rightarrow V(G/\mathcal{Q})$.
\end{enumerate}
\end{definition}
\cref{lemma:lift_cuts_not_crossed_by_snug_path_1} below tells us that lifting a $k$-augmentation-safe cover preserves the defining property for all $k$-cuts that are not crossed by any snug path. By \cref{lemma:only_cut_intersecting_snug_edge}, the only $k$-cuts that a snug path crosses are its ``inner'' snug shores (and their complements).
\begin{lemma}\label{lemma:lift_cuts_not_crossed_by_snug_path_1}
Let $G$ be a $k$-edge-connected graph and let $F$ be a set of links. Let $\mathcal{U}'=(U'_i)_{i\in I}$ be a $k$-augmentation-safe cover for $(G,F)/\mathcal{P}^G_{\rm{chain}}$ and let $\mathcal{U}=(U_i)_{i\in I}$ be the lift of $\mathcal{U}'$ to $G$. Then for every $k$-cut $S$ of $G$ that is not crossed by $\mathcal{P}^G_{\rm{chain}}$, one of \cref{def:augmentation_safe_cover}~\ref{def:augmentation_safe:1} or \cref{def:augmentation_safe_cover}~\ref{def:augmentation_safe:2} applies. 
\end{lemma}
\begin{proof}
Let $S$ be a $k$-cut in $G$ that is not crossed by $\mathcal{P}^G_{\rm{chain}}$. Then $S$ corresponds to a $k$-cut $T$ in $G/\mathcal{P}^G_{\rm{chain}}$. Now \cref{def:augmentation_safe_cover}~\ref{def:augmentation_safe:1} or \cref{def:augmentation_safe_cover}~\ref{def:augmentation_safe:2}  holds for $T$ and $\mathcal{U}'$, which implies that \cref{def:augmentation_safe_cover}~\ref{def:augmentation_safe:1} or \cref{def:augmentation_safe_cover}~\ref{def:augmentation_safe:2} holds for $S$ and $\mathcal{U}$. 
\end{proof}
We continue by investigating the interactions between lifted covers and snug shores. 
\cref{lemma:interaction_of_snug_shores_with_cover} tells us we can lift a well-separated $k$-augmentation-safe cover for $(G,F)/\mathcal{P}^G_{\rm{chain}}$ as long as we do not uncontract the snug paths that are ``close to the boundaries of the pieces of the cover'' (i.e., those for which \ref{item:u_P_in_V_U_or_neighborhood} of \cref{lemma:interaction_of_snug_shores_with_cover} applies). The proof of \cref{lemma:interaction_of_snug_shores_with_cover} is deferred to the end of the section.
\begin{lemma}\label{lemma:interaction_of_snug_shores_with_cover}
	Let $G$ be a $k$-edge-connected graph, let $F$ be a set of links and let $\mathcal{U}'=(U'_i)_{i\in I}$ be a well-separated $k$-augmentation-safe cover of $(G,F)/\mathcal{P}^G_{\rm{chain}}$. Denote the lift of $\mathcal{U}'$ to $G$ by $\mathcal{U}=(U_i)_{i\in I}$. Let $P\in\mathcal{P}^G_{\rm{chain}}$ and denote the contracted vertex corresponding to $P$ in $G/\mathcal{P}^G_{\rm{chain}}$ by $u_P$. Then one of the following holds:
	\begin{enumerate}
		\item \label{item:u_P_in_V_U_or_neighborhood} $u_P\in V'_{\mathcal{U'}}\cup \Gamma_{(G+F)/\mathcal{P}^G_{\rm{chain}}}(V'_{\mathcal{U'}})$, where $V'_{\mathcal{U'}}\coloneqq \bigcup_{i,j \in I: i\neq j} U'_i\cap U'_j$.
		\item \label{item:all_snug_cuts_in_U_i} There is a unique $i\in I$ such that $u_P\in U'_i$. Moreover, we have $\delta_{E\dot{\cup} F}(S)\subseteq (E\dot{\cup}F)[U_i]$ for every snug shore $S$ of $P$.
		\item \label{item:all_snug_cuts_covered} We have $\delta_{F_\mathcal{U}}(S)\ne \emptyset$ for every snug shore $S$ of $P$.
	\end{enumerate}
\end{lemma}
As outlined in \cref{sec:augmentation_short}, our goal is to lift a $k$-augmentation-safe cover of $(G,F)/\mathcal{P}^G_{\rm{chain}}$ (for some ``cheap'' link set $F$) to a $k$-augmentation-safe cover such that all snug paths that remain contracted can be cheaply covered, and such that all snug paths that we uncontract form snug paths in the pieces they appear in. To achieve this property, we will keep the snug paths satisfying properties \ref{item:u_P_in_V_U_or_neighborhood} or \ref{item:all_snug_cuts_covered} from \cref{lemma:interaction_of_snug_shores_with_cover} contracted. We will use \cref{theorem:compute_k_vertex_safe_cover_with_vertex_weights} with appropriate vertex weights for the snug paths to ensure that we can cheaply cover the paths for which \ref{item:u_P_in_V_U_or_neighborhood} applies. The paths to which \ref{item:all_snug_cuts_covered} applies can be covered by $F_{\mathcal{U}}$.
 
To state our choice of the link set $F$, we need to introduce the following notation.
\begin{definition}\label{def:L_snug}
Given an instance $(G,L,c)$ of $k$-WCAP, we define $L_{\rm{snug}}\coloneqq L\cap\binom{V_{\rm{snug}}}{2}$ to be the collection of links with two snug endpoints. We further define the collection $L^\circ_{\rm{snug}}\subseteq L_{\rm{snug}}$, which, for every pair of snug paths $P$ and $Q$ such that there exists a link with one endpoint in $P$ and one endpoint in $Q$, contains a cheapest one (breaking ties arbitrarily).
\end{definition}
\begin{proposition}\label{prop:L_snug_edge_connectivity}
Let $(G,L,c)$ be an instance of $k$-WCAP. \\Then $(G+(L\setminus L_{\rm{snug}} \cup L^\circ_{\rm{snug}}))/\mathcal{P}^G_{\rm{chain}}$ is $(k+1)$-edge-connected.
\end{proposition}
\begin{proof}
$(G+(L\setminus L_{\rm{snug}} \cup L^\circ_{\rm{snug}}))/\mathcal{P}^G_{\rm{chain}}$ arises from $(G+L)/\mathcal{P}^G_{\rm{chain}}$ by removing all but one copy of certain parallel links (recall that we always delete loops arising from edge contractions). The facts that $G$ is $k$-edge-connected and we assumed $k$-WCAP instances to be feasible conclude the proof.
\end{proof}
In order to guarantee that $F_{\mathcal{U}}$ is cheap, we will choose $F=L\setminus L_{\rm{snug}}\cup L^\circ_{\rm{snug}}$ (after pre-processing the instance to ensure that $L\setminus L_{\rm{snug}}$ is cheap). To solve $k$-WCAP within the pieces, we will, however, have to add additional links to ensure feasibility.
\begin{lemma}\label{lemma:augmentation_safe_cover_respecting_snug_paths}
	Let $(G,L,c)$ be an instance of $k$-WCAP.
	Let $\mathcal{U}'=(U'_i)_{i\in I}$ be a well-separated $k$-augmentation-safe cover of $(G,L\setminus L_{\rm{snug}}\cup L^\circ_{\rm{snug}})/\mathcal{P}^G_{\rm{chain}}$ and let $\mathcal{U}$ denote its lift to $G$. Let $\mathcal{Q}$ be the set of snug paths to which \ref{item:u_P_in_V_U_or_neighborhood} or \ref{item:all_snug_cuts_covered} from \cref{lemma:interaction_of_snug_shores_with_cover} applies (with $F=L\setminus L_{\rm{snug}}\cup L^\circ_{\rm{snug}}$). Finally, let $L'\supseteq L\setminus L_{\rm{snug}}\cup L^\circ_{\rm{snug}}$ be any link set. Then:
	\begin{enumerate}
		\item \label{item:augmentation_safe_cover_in_contracted_instance} $\mathcal{U}/\mathcal{Q}\eqqcolon (U^\mathcal{Q}_i)_{i\in I}$ is a $k$-augmentation-safe cover of $(G_\mathcal{Q}=(V_\mathcal{Q},E_{\mathcal{Q}}),L'_{\mathcal{Q}})\coloneqq (G,L')/\mathcal{Q}$.
		\item \label{item:snug_paths_survive} For every $P\in\mathcal{P}^G_{\rm{chain}}\setminus \mathcal{Q}$, there is a unique $i\in I$ such that $V(P)\subseteq U_i$. Moreover, $P$ is a subpath of a snug path in $G_{\mathcal{Q}}/(E_{\mathcal{Q}}\dot{\cup} L'_{\mathcal{Q}})[V_{\mathcal{Q}}\setminus U^\mathcal{Q}_i]$.
	\end{enumerate}
	\end{lemma}
\subsubsection{Proving \cref{lemma:interaction_of_snug_shores_with_cover,lemma:augmentation_safe_cover_respecting_snug_paths}}

\begin{proof}[Proof of \cref{lemma:interaction_of_snug_shores_with_cover}]
	As $\mathcal{U}'$ is well-separated, we know that
	\begin{equation}
		(E\dot{\cup} F)[U_i\setminus U_j,U_j\setminus U_i]=\emptyset \text{ for $i\ne j\in I$.}\label{eq:no_edges_between_different_sets}
	\end{equation}
	Assume that \ref{item:u_P_in_V_U_or_neighborhood} does not hold. Then there is a unique $i\in I$ such that $u_P\in U'_i$. Let the vertex sequence of $P$ be $(u_0,\dots,u_t)$ and let the snug shores of $P$ be $(S_0,\dots,S_{t+1})$. 
	As  \ref{item:u_P_in_V_U_or_neighborhood} does not hold, we have $u_j\in U_i\setminus (V_{\mathcal{U}}\cup \Gamma_{E\dot{\cup} F}(V_{\mathcal{U}}))$ for $j\in \{0,\dots,t\}$.
	In particular, 
	\begin{align}
		&\delta_{E\dot{\cup} F}(u_j)\subseteq (E\dot{\cup}F)[U_i] \text{ and }\label{eq:edges_u_j_1}\\
			&\delta_{E\dot{\cup} F}(u_j)\cap (E\dot{\cup} F)[U_{i'}]=\emptyset \text{ for each $i'\in I\setminus \{i\}$} \text{ and }\label{eq:edges_u_j_2}\\
			&\delta_{F_{\mathcal{U}}}(u_j)=\emptyset\label{eq:edges_u_j_3}
	\end{align} by \eqref{eq:no_edges_between_different_sets} and because $u_j$ is neither contained in nor adjacent to a vertex in $V_{\mathcal{U}}$. By \cref{lemma:outer_cuts}, we know that no snug path crosses the cuts $S_0$ or $S_{t+1}$. Hence, $S_0$ and $S_{t+1}$ correspond to cuts $S'_0$ and $S'_{t+1}$ in $(G+F)/\mathcal{P}^G_{\rm{chain}}$ with $S'_{t+1}=S'_0\cup \{u_P\}$. Moreover, \[\delta_{(E\dot{\cup}F)/\mathcal{P}^G_{\rm{chain}}}(u_P)\subseteq \delta_{(E\dot{\cup}F)/\mathcal{P}^G_{\rm{chain}}}(S'_0)\cup \delta_{(E\dot{\cup}F)/\mathcal{P}^G_{\rm{chain}}}(S'_{t+1}).\] As $(G+F)/\mathcal{P}^G_{\rm{chain}}$ is connected, we may assume that \begin{align}\delta_{(E\dot{\cup}F)/\mathcal{P}^G_{\rm{chain}}}(u_P)\cap \delta_{(E\dot{\cup}F)/\mathcal{P}^G_{\rm{chain}}}(S'_0)&\neq \emptyset\text{ and hence, }\notag\\\delta_{E\dot{\cup} F}(\{u_0,\dots,u_t\})\cap \delta_{E\dot{\cup} F}(S_0)&\neq \emptyset;\label{eq:edges_in_S_0_incident_to_snug_vertices}\end{align} the subsequent arguments work analogously if we replace $S'_0$ and $S_0$ by $S'_{t+1}$ and $S_{t+1}$, respectively.
	By \cref{lemma:lift_cuts_not_crossed_by_snug_path_1,lemma:outer_cuts}, we know that
	\begin{enumerate}[(a)]
		\item \label{item:S_0_in_U_i_prime} There is $i'\in I$ such that $\delta_{E\dot{\cup} F} (S_0)\subseteq (E\dot{\cup} F)[U_{i'}]$, or
		\item \label{item:S_0_covered} $\delta_{F_{\mathcal{U}}} (S_0)\ne \emptyset$.
	\end{enumerate}
	If \ref{item:S_0_in_U_i_prime} applies, then \eqref{eq:edges_u_j_2} and \eqref{eq:edges_in_S_0_incident_to_snug_vertices} imply $i=i'$.
	For $j\in \{0,\dots,t+1\}$, we have
	\[\delta_{E\dot{\cup} F}(S_j)\setminus \delta_{E\dot{\cup} F}(S_0)\subseteq \delta_{E\dot{\cup} F}(S_0\Delta S_j)=\delta_{E\dot{\cup} F}(\{u_0,\dots,u_{j-1}\})\subseteq (E\dot{\cup} F)[U_i]\] by \eqref{eq:edges_u_j_1}, where $S_0\Delta S_j=(S_0\setminus S_j)\cup (S_j\setminus S_0)$ denotes the symmetric difference of $S_0$ and $S_j$. Hence, $\delta_{E\dot{\cup} F}(S_j)\subseteq (E\dot{\cup} F)[U_i]$ for all $j\in\{0,\dots,t+1\}$, i.e., \ref{item:all_snug_cuts_in_U_i} holds.
	
	On the other hand, if \ref{item:S_0_covered} applies, then for $j\in\{0,\dots,t+1\}$, we have
	\[\delta_{F_{\mathcal{U}}}(S_j)\Delta \delta_{F_{\mathcal{U}}}(S_0)= \delta_{F_{\mathcal{U}}}(S_0\Delta S_j)=\delta_{{F_{\mathcal{U}}}}(\{u_0,\dots,u_{j-1}\})=\emptyset\] by \eqref{eq:edges_u_j_3}. This yields \ref{item:all_snug_cuts_covered}.
\end{proof}
For the proof of \cref{lemma:augmentation_safe_cover_respecting_snug_paths}, we first establish \cref{lemma:contraction}, which allows us to work with $\mathcal{U}$ instead of $\mathcal{U}/\mathcal{Q}$, and \cref{lemma:lift_cuts_not_crossed_by_snug_path_2}, which takes care of the $k$-cuts not crossed by snug paths.
\begin{lemma}\label{lemma:contraction}
	Let $G=(V,E)$ be a $k$-edge-connected graph, let $F$ be a set of links and let $\mathcal{Q}\subseteq \mathcal{P}^G_{\rm{chain}}$.
	Let $\mathcal{U}=(U_i)_{i\in I}$ be a collection of subsets of $V$ covering $V$ such that for every $k$-cut $S$ of $G$ that is not crossed by $\mathcal{Q}$, one of \cref{def:augmentation_safe_cover}~\ref{def:augmentation_safe:1} or \cref{def:augmentation_safe_cover}~\ref{def:augmentation_safe:2} holds. Then $\mathcal{U}/\mathcal{Q}$ is a $k$-augmentation-safe cover of $(G,F)/\mathcal{Q}$.
\end{lemma}
\begin{proof}
	Let $S$ be a $k$-cut in $G/\mathcal{Q}$. Then $S$ corresponds to a $k$-cut $T$ in $G$ such that for every $Q\in\mathcal{Q}$, we have $V(Q)\subseteq T$ or $V(Q)\cap T=\emptyset$. Hence, one of \cref{def:augmentation_safe_cover}~\ref{def:augmentation_safe:1} or \cref{def:augmentation_safe_cover}~\ref{def:augmentation_safe:2} applies to $T$. In the first case, let $i\in I$ with $\delta_{E\dot{\cup} F}(T)\subseteq (E\dot{\cup}F)[U_i]$. As every edge or link crossing $S$ is the image of an edge or link crossing $T$, and $U'_i$ is the image of $U_i$, \cref{def:augmentation_safe_cover}~\ref{def:augmentation_safe:1} also applies to $S$ (for $\mathcal{U}/\mathcal{Q}$).
	
	Next, assume that \cref{def:augmentation_safe_cover}~\ref{def:augmentation_safe:2} applies to $T$. Let $\ell=\{v,w\}\in\delta_{F_{\mathcal{U}}}(T)$ and let $v\in V_{\mathcal{U}}$. Let $i,j\in I$, $i \ne j$ with $v\in U_i\cap U_j$. Let $v'$ be the image of $v$ in $V/\mathcal{Q}$ and let $\ell'\in F/\mathcal{Q}$ be the image of $\ell$. Note that $\ell'$ cannot have become a loop (and been removed) since $\ell$ crosses $T$, but no path in $\mathcal{Q}$ does. Then $\ell'\in \delta_{F/\mathcal{Q}}(S)$ and $v'\in U'_i\cap U'_j$, so $\ell'\in (F/\mathcal{Q})_{\mathcal{U'}}$. Hence, \cref{def:augmentation_safe_cover}~\ref{def:augmentation_safe:2} applies to $S$.
\end{proof}
\begin{lemma}\label{lemma:lift_cuts_not_crossed_by_snug_path_2}
	Let $(G,L)$ be an instance of $k$-CAP, let $\mathcal{U'}$ be a well-separated $k$-augmentation-safe cover of $(G,L\setminus L_{\rm{snug}} \cup L^\circ_{\rm{snug}})/\mathcal{P}^G_{\rm{chain}}$ and let $\mathcal{U}$ be the lift of $\mathcal{U'}$ to $G$. Let $L'\subseteq L$ be any link set with $L\setminus L_{\rm{snug}}\cup L^\circ_{\rm{snug}}\subseteq L'$. Then for any $k$-cut $S$ of $G$ that is not crossed by $\mathcal{P}^G_{\rm{chain}}$, one of \cref{def:augmentation_safe_cover}~\ref{def:augmentation_safe:1} or \cref{def:augmentation_safe_cover}~\ref{def:augmentation_safe:2} holds for the collection $\mathcal{U}$ and $F= L'$.
\end{lemma}
\begin{proof} By \cref{lemma:lift_cuts_not_crossed_by_snug_path_1}, we know that one of \cref{def:augmentation_safe_cover}~\ref{def:augmentation_safe:1} or \cref{def:augmentation_safe_cover}~\ref{def:augmentation_safe:2} holds for $S$ with respect to $\mathcal{U}$ and $F=L\setminus L_{\rm{snug}} \cup L^\circ_{\rm{snug}}$. If \ref{def:augmentation_safe:2} holds, then it also holds for $F=L'\supseteq L\setminus L_{\rm{snug}}\cup L^\circ_{\rm{snug}}$. Next, assume that \ref{def:augmentation_safe:1} holds for $S$. Let $i\in I$ such that 
	\begin{equation}\delta_{E\dot{\cup}(L\setminus L_{\rm{snug}}\cup L^\circ_{\rm{snug}})}(S)\subseteq(E\dot{\cup}(L\setminus L_{\rm{snug}}\cup L^\circ_{\rm{snug}}))[U_i].\label{eq:old_edges_and_links_are_good}\end{equation} It suffices to prove that for every link in $\ell\in (L_{\rm{snug}}\setminus L^\circ_{\rm{snug}})\cap \delta_L(S)$, both endpoints of $\ell$ are contained in $U_i$. 
	
	No link $\ell\in L_{\rm{snug}}\setminus L^\circ_{\rm{snug}}$ with both endpoints on the same snug path $P$ crosses $S$ because no snug path crosses $S$ by assumption. Next, consider a link $\ell\in L_{\rm{snug}}\setminus L^\circ_{\rm{snug}}$ with one endpoint in one snug path, say, $P$, and the other endpoint in another snug path, say, $Q$. Then $L^\circ_{\rm{snug}}$ contains a link $\ell'$ with one endpoint in $P$ and one endpoint in $Q$. As no snug path crosses $S$, $\ell'$ crosses $S$ if and only if $\ell$ does. Moreover, $\ell'$ has both endpoints in $U_i$ if and only if $\ell$ does. By \eqref{eq:old_edges_and_links_are_good}, $\ell'$, and, hence, $\ell$, has both endpoints in $U_i$.
\end{proof}
\begin{proof}[Proof of \cref{lemma:augmentation_safe_cover_respecting_snug_paths}]
We first prove \ref{item:augmentation_safe_cover_in_contracted_instance}. By \cref{lemma:contraction}, it suffices to show that for every $k$-cut $S$ of $G$ that is not crossed by $\mathcal{Q}$, \cref{def:augmentation_safe_cover}~\ref{def:augmentation_safe:1} or \cref{def:augmentation_safe_cover}~\ref{def:augmentation_safe:2} holds. For a $k$-cut $S$ that is not crossed by any snug path, this follows from \cref{lemma:lift_cuts_not_crossed_by_snug_path_2}.
Next, consider a $k$-cut $S$ that is crossed by a snug path in $P\in\mathcal{P}^G_{\rm{chain}}\setminus \mathcal{Q}$. Let $(S_0,\dots,S_{t+1})$ be the snug shores of $P$. By \cref{lemma:only_cut_intersecting_snug_edge} and since no snug shore contains $r$, we have $S\in\{S_1,\dots,S_t\}$. By \cref{lemma:interaction_of_snug_shores_with_cover}, we know that there exists a unique $i\in I$ such that $u_P$, the contracted vertex corresponding to $P$, is contained in $U'_i$. Moreover, \begin{equation}\delta_{E\dot{\cup} (L\setminus L_{\rm{snug}}\cup L^\circ_{\rm{{snug}}})}(S_j)\subseteq (E\dot{\cup} (L\setminus L_{\rm{snug}}\cup L^\circ_{\rm{{snug}}}))[U_i]\quad\forall j\in\{0,\dots,t+1\}. \label{eq:cut_good_for_old_links_and_edges}\end{equation} Consider a link $\ell\in L_{\rm{snug}}\setminus L^\circ_{\rm{snug}}$. If both endpoints of $\ell$ are contained in one snug path $Q$, then $\ell$ can cross $S$ only if $Q$ crosses $S$. By \cref{lemma:only_cut_intersecting_snug_edge}, the latter is only the case if $P=Q$ and in this case, $\ell\in L[U_i]$ because $u_P\in U'_i$.

Next, consider the case where the endpoints of $\ell$ are contained in different snug paths $Q_1$ and $Q_2$, and assume that $\ell$ crosses $S$. By \cref{lemma:link_crosses_outer_shore}, $\ell$ also crosses $S_0$ or $S_{t+1}$. By definition, $L^\circ_{\rm{snug}}$ contains a link $\ell'$ with one endpoint in $Q_1$ and one endpoint in $Q_2$. By \cref{lemma:only_cut_intersecting_snug_edge}, $\ell'$ also crosses $S_0$ or $S_{t+1}$ because neither $Q_1$ nor $Q_2$ crosses $S_0$ or $S_{t+1}$.
By \eqref{eq:cut_good_for_old_links_and_edges}, both endpoints of $\ell'$ are contained in $U_i$. In particular, the contracted vertices $u_{Q_1}$ and $u_{Q_2}$ corresponding to $Q_1$ and $Q_2$ are contained in $U'_i$. Hence, both endpoints of $\ell$ are contained in $U_i$ as well.

Finally, we prove \ref{item:snug_paths_survive}. Let $P\in\mathcal{P}^G_{\rm{chain}}\setminus \mathcal{Q}$. As we have seen in the proof of \ref{item:augmentation_safe_cover_in_contracted_instance}, there is a unique $i\in I$ such that $u_P\in U'_i$. We have $V(P)\subseteq U_i$ and $V(P)\cap U_j=\emptyset$ for $j\ne i$. Moreover, for every snug shore $S$ of $P$, we have $\delta_{E\dot{\cup} L'}(S)\subseteq (E\dot{\cup} L')[U_i]$, so $S$ corresponds to a $k$-cut in $G/(E\dot{\cup} L')[V\setminus U_i]$. By \cref{lemma:only_cut_intersecting_snug_edge}, we further know that no path in $\mathcal{Q}$ crosses $S$, so $S$ corresponds to a $k$-cut in $G_{\mathcal{Q}}/(E_{\mathcal{Q}}\dot{\cup} L'_{\mathcal{Q}})[V_{\mathcal{Q}}\setminus U^\mathcal{Q}_i]$. Moreover, as no edge or link crossing a snug shore of $P$ is contracted when constructing $G_{\mathcal{Q}}/(E_{\mathcal{Q}}\dot{\cup} L'_{\mathcal{Q}})[V_{\mathcal{Q}}\setminus U^\mathcal{Q}_i]$, none of the snug vertices of $P$ is merged with any other vertex. Hence, the snug chain corresponding to $P$ yields a subchain of a snug chain in $G_{\mathcal{Q}}/(E_\mathcal{Q}\dot{\cup} L'_{\mathcal{Q}})[V_\mathcal{Q}\setminus U^\mathcal{Q}_i]$, i.e., $P$ is a subpath of a snug path in $G_{\mathcal{Q}}/(E_\mathcal{Q}\dot{\cup} L'_{\mathcal{Q}})[V_\mathcal{Q}\setminus U^\mathcal{Q}_i]$.
\end{proof}
\subsection{Thinning the link set}\label{sec:thinning}
In this section, we explain how to thin the link set of an instance of planar $k$-WCAP with bounded cost ratio to obtain good bounds on the total cost of links we need to buy in the process of decomposing our instance into subinstances of bounded snug-treewidth.
All of these results crucially rely on \cref{lemma:bound_non_snug_vertices_and_snug_paths}, which we prove in this section, as well as the following proposition, which we already observed in \cref{sec:augmentation_short}.
\begin{proposition}\label{prop:lower_bound_opt_n_k}
Let $(G,L,c)$ be an instance of $k$-WCAP. Let $c_{min}$ denote the minimum cost of a link and let $n_k(G)$ denote the number of vertices of $G$ of degree $k$. Then	$c(\OPT)\ge \frac{n_k(G)}{2}\cdot c_{min}$, where $\OPT$ denotes an optimum solution to $(G,L,c)$.
\end{proposition}
\begin{proof}
	A link $\ell=\{u,v\}$ crosses a $k$-cut of the form $\{w\}$ only if $u=w$ or $v=w$. Hence, we need at least $\frac{n_k(G)}{2}$ links to cross all of the singleton $k$-cuts $\{w\}$, where $w$ is a vertex of degree $k$.
\end{proof}

We further prove \cref{lemma:reduction_to_sparse_instances}, which allows us to reduce to a setting where the total cost of non-snug links can be bounded in terms of the cost of an optimum solution.

Another key result that we prove in this section is \cref{lemma:link_sets_covering_snug_chains}, which helps us to cheaply cover the snug paths close to the boundary of the $k$-augmentation-safe cover we will compute (see \cref{lemma:interaction_of_snug_shores_with_cover}~\ref{item:u_P_in_V_U_or_neighborhood}).

\begin{lemma}\label{lemma:reduction_to_sparse_instances}
Let $\lambda\in (0,1)$ and $\Delta \ge 1$. Given an instance $(G,L,c)$ of planar $k$-WCAP with cost ratio $\Delta$ and such that $G$ is minimally $k$-edge-connected, we can, in polynomial time, compute a link set $L'$ with $L_{\rm{snug}}\subseteq L'\subseteq L$ satisfying the following properties:
\begin{enumerate}[(i)]
	\item \label{item:sparse_instance_cheap} $G+L'$ is $(k+1)$-edge-connected and $c(\OPT')\le (1+\lambda)\cdot c(\OPT)$, where $\OPT$ and $\OPT'$ denote optimum solutions to $(G,L,c)$ and $(G,L',c)$, respectively.
	\item \label{item:sparse_instance} $c(L'\setminus L_{\rm{snug}})\le 108\cdot \frac{\Delta^2}{\lambda}\cdot c(\OPT)$.
	
\end{enumerate}
\end{lemma}
\begin{lemma}\label{lemma:link_sets_covering_snug_chains}
Given an instance $(G,L,c)$ of planar $k$-WCAP with cost ratio $\Delta$ and such that $G$ is minimally $k$-edge-connected, we can, in polynomial time, compute link sets $(L_P)_{P\in\mathcal{P}^G_{\rm{chain}}}$ with the following properties:
	\begin{enumerate}[(i)]
		\item \label{item:sets_cover_paths} For $P\in\mathcal{P}^G_{\rm{chain}}$, $L_P$ covers $P$.
		\item \label{item:cost_bound_L_Ps} $\sum_{P\in\mathcal{P}^G_{\rm{chain}}} c(L_P)\le (8\Delta +1)\cdot c(\OPT)$.
	\end{enumerate}
\end{lemma}

\subsubsection{Proving \cref{lemma:bound_non_snug_vertices_and_snug_paths}}
Fix an instance $(G=(V,E),L,c)$ of planar $k$-WCAP, where $G$ is minimally $k$-edge-connected. Further fix a root $r\in V$.
To derive \cref{lemma:bound_non_snug_vertices_and_snug_paths}, we will analyze the structure of certain laminar families of $k$-cuts.
\begin{lemma}\label{lemma:laminar_family}
	There exists a laminar family $\mathcal{L}\subseteq 2^{V\setminus\{r\}}$ of $k$-cuts such that every edge crosses at least one cut in $\mathcal{L}$.
\end{lemma}
\begin{proof}
	Let $T=(V,F)$ be a Gomory-Hu tree for $G$. Recall that a Gomory-Hu tree for $G$ is a tree on the vertex set $V$ such that for every edge $f=\{u,v\}\in F$, the vertex set of any of the two connected components of $T-f$ defines a minimum cut separating $u$ and $v$ in $G$. (See, e.g., \cite{korteCombinatorialOptimizationTheory2018}, for more information on Gomory-Hu trees.)
	Let $\mathcal{L}$ be the family of all such cuts that do not contain the root (to avoid including the same cut twice, once for each of its shores) and that correspond to $k$-cuts of $G$. Clearly, $\mathcal{L}\subseteq 2^{V\setminus\{r\}}$ is a laminar family of $k$-cuts. For $e=\{v,w\}\in E$, there is a $k$-cut containing $e$ because $G$ is minimally $k$-edge-connected. Hence, the minimum size of a cut in $G$ separating $v$ and $w$ is $k$. By definition of a Gomory-Hu tree, one of the cuts in $\mathcal{L}$ is a $k$-cut separating $v$ and $w$, i.e., it is crossed by $e$.
\end{proof}
In the following, fix a \emph{maximal} laminar family $\mathcal{L}\subseteq 2^{V\setminus\{r\}}$ of $k$-cuts in $G$ with the property that each edge crosses at least one of the $k$-cuts. We remark that we do not need to compute this family; we only need it for the analysis.
\begin{lemma}\label{lemma:snug:shores_in_L}
	Let $u$ be a snug vertex with snug shores $(S_1,S_2)$. Then $S_1,S_2\in\mathcal{L}$.
\end{lemma}
\begin{proof}
	By \cref{lemma:interaction_with_snug_shores}, no $k$-cut $S\in\mathcal{L}$ crosses $S_1$ or $S_2$. Hence, the desired statement follows by maximality of $\mathcal{L}$.
\end{proof}
We call the inclusionwise maximal proper subsets of $S\in\mathcal{L}$ contained in $\mathcal{L}$ the \emph{children of $S$}. If $S$ does not have any children, we call $S$ a \emph{leaf}.
\begin{lemma}\label{lemma:properties_laminar_family}\
	\begin{enumerate}[label=(\roman*)]
		\item Every leaf of $\mathcal{L}$ is a singleton, i.e., a set of size $1$.
		\item Let $S\in\mathcal{L}$ such that $S$ has exactly one child $S'$. Then  there is a snug vertex $v\in V\setminus \{r\}$ such that $(S',S)$ form snug shores for $v$.
	\end{enumerate}
\end{lemma}
\begin{proof}
	Let $S\in \mathcal{L}$ be a leaf. Then for every $v\in S$, we have $\delta(v)\subseteq \delta(S)$ because every edge in $G$ crosses some $k$-cut in $\mathcal{L}$. 
	As $S$ is a $k$-cut and every vertex in $G$ has degree at least $k$ (because $G$ is $k$-edge-connected), $S$ must be a singleton.
	
	Next, let $S$ be a set with exactly one child $S'$ and let $v\in S\setminus S'$. Then $\deg(v)\ge k+1$ because otherwise, we could have added $\{v\}$ to $\mathcal{L}$, contradicting its maximality. Moreover, we must have $\delta(v)\subseteq \delta(S)\cup \delta(S')$ because every edge in $G$ crosses some $k$-cut in $\mathcal{L}$. As both $S$ and $S'$ are $k$-cuts, $S\setminus S'$ contains exactly one vertex.
\end{proof}
We are now ready to prove \cref{lemma:number_snug_paths,lemma:non_snug_vertices}, which, together, yield \cref{lemma:bound_non_snug_vertices_and_snug_paths}.
\begin{lemma}\label{lemma:number_snug_paths}
	We have $|\mathcal{P}_{\rm{chain}}^G|< 2\cdot n_k(G)$.
\end{lemma}
\begin{proof}
	For a snug path $P$ with snug shores $(S_0,\dots,S_{t+1})$, let $S_P\coloneqq S_0$. Note that $S_P\in\mathcal{L}$ by \Cref{lemma:snug:shores_in_L}. The sets $(S_P)_{P\in\mathcal{P}_{\rm{chain}}^G}$ are pairwise distinct by \Cref{lemma:shore_for_at_most_one_path}.
	Moreover, by \Cref{lemma:properties_laminar_family}, every set $S_P$ is either a leaf of $\mathcal{L}$, or has at least two children in $\mathcal{L}$. All leaves of $\mathcal{L}$ are singleton $k$-cuts by \Cref{lemma:properties_laminar_family}. Moreover, there are strictly fewer sets with at least two children than leaves. (This can be observed by looking at the tree representation of $\mathcal{L}$.) Hence, the claimed bound follows. 
\end{proof}
\begin{lemma}\label{lemma:non_snug_vertices}
	We have $|V\setminus V_{\rm{snug}}|<4\cdot n_k(G)$.
\end{lemma}
\begin{proof}
	Let $\mathcal{L}_s$ denote the collection of sets in $\mathcal{L}$ with exactly one child.
	For $L\in \mathcal{L}_s$, denote the unique child of $L$ by $\child(L)$ and let $v_L$ be the unique vertex in $L\setminus \child(L)$ (see \cref{lemma:properties_laminar_family}). We further denote by 
	\begin{equation*}
		E_L \coloneqq \delta(L) \cap \delta(\child(L))
	\end{equation*}
	all edges that cross both $L$ and its child.
	For each $e\in E$, we have
	\begin{equation}\label{eq:bound_crossing_edges_using_E_L}
		|\{L\in \mathcal{L} \colon e\in \delta(L)\}| \geq 1 + |\{L\in \mathcal{L}_s \colon e\in E_L\}|,
	\end{equation}
	which can be derived by observing that $e$ crosses each $L\in \mathcal{L}_s$ with $e\in E_L$, and, on top of that, crosses at least one more set, for example the smallest one in $\mathcal{L}$ that is crossed by $e$.
	(Recall that every edge crosses at least one set in $\mathcal{L}$.)
	Moreover, for $L\in\mathcal{L}_s$, we have
	\begin{equation}\label{eq:bound_E_L}
		\begin{aligned}
			|E_L| &= |\delta(L) \cap \delta(\child(L))| = \frac{1}{2}\left(|\delta(L)| + |\delta(\child(L))| - |\delta(v_L)|\right)
			= \frac{2k - \deg(v_L)}{2}.
		\end{aligned}
	\end{equation}
	
	By combining these relations, we get
	\begin{align*}
		k \cdot | \mathcal{L} | &= \sum_{L\in\mathcal{L}} |\delta(L)|=\sum_{e \in E} |\{L \in \mathcal{L} : e \in \delta(L)\}| \\
		&\geq |E| + \sum_{e \in E} |\{L \in \mathcal{L}_s : e \in E_L \}| && \text{by \eqref{eq:bound_crossing_edges_using_E_L}}\\
		&= |E| + \sum_{L \in \mathcal{L}_s} |E_L| \\
		&= |E| + \sum_{L \in \mathcal{L}_s} \frac{2k - \deg(v_L)}{2} && \text{by \eqref{eq:bound_E_L}}\\
		&= \frac{1}{2} \sum_{v \in V} \deg(v) + k \cdot |\mathcal{L}_s| - \frac{1}{2} \sum_{L \in \mathcal{L}_s} \deg(v_L) \\
		&= \frac{1}{2} \sum_{v \in V\setminus V_{\rm{snug}}} \underbrace{\deg(v)}_{\geq k} + \ k \cdot |\mathcal{L}_s| \\
		&\geq \frac{1}{2} k \cdot |V\setminus V_{\rm{snug}}| + k \cdot |\mathcal{L}_s|,
	\end{align*}
	where the fourth equality uses $|E|=\frac{1}{2}\sum_{v\in V}\deg(v)$ and the fifth one follows from $\{v_L \colon L \in \mathcal{L}_s\} = V_{\rm{snug}}$ by \cref{lemma:snug:shores_in_L,lemma:properties_laminar_family}.
	
	By reordering terms, we get
	\begin{equation}\label{eq:bound_non_snug_vertices}
		|V\setminus V_{\rm{snug}}| \leq 2 \cdot |\mathcal{L} \setminus \mathcal{L}_s|.
	\end{equation}
	
	Note that $\mathcal{L}\setminus \mathcal{L}_s$ is the number of sets in $\mathcal{L}$ that have either at least two children or are leaves.
	Moreover, the number of sets in $\mathcal{L}$ that have at least two children is strictly less than the number of leaves in $\mathcal{L}$.
	Using \cref{lemma:properties_laminar_family}, these observations imply
	\begin{equation*}
		|\mathcal{L}\setminus \mathcal{L}_s| < 2 \cdot |\{L\in \mathcal{L}\colon L \text{ is a leaf }\}|\le 2\cdot n_k(G),
	\end{equation*}
	and the desired result follows by combining this relation with \eqref{eq:bound_non_snug_vertices}.
\end{proof}
\begin{corollary}\label{cor:cost_bound_L_prime_snug}
We have $c(L^\circ_{\rm{snug}})\le 12\cdot \Delta\cdot c(\OPT)$, where $\OPT$ denotes an optimum solution to $(G,L,c)$ and $\Delta$ denotes the cost ratio of the instance.
\end{corollary}
\begin{proof}
By \cref{lemma:chain_graph_subgraph}, $(V_{\rm{snug}},L^\circ_{\rm{snug}})/\mathcal{P}^G_{\rm{chain}}$ is a loopless simple planar graph. As a simple planar graph with $n$ vertices can have at most $3n-6$ edges, \cref{lemma:number_snug_paths} yields $|L^\circ_{\rm{snug}}|\le 6\cdot n_k(G)$. 
Using \cref{prop:lower_bound_opt_n_k}, we obtain
\[c(L^\circ_{\rm{snug}})\le 6\cdot c_{max}\cdot n_k(G)\le 12\Delta\cdot c(\OPT).\]
\end{proof}
\subsubsection{Proof of \cref{lemma:reduction_to_sparse_instances}}
\begin{proof}[Proof of \cref{lemma:reduction_to_sparse_instances}]
Let $(G=(V,E),L,c)$ be an instance of planar $k$-WCAP such that $G$ is minimally $k$-edge-connected and let $r\in V$ be a fixed root. Recall that we assume all $k$-WCAP instances we deal with to be feasible, i.e., $G+L$ is $(k+1)$-edge-connected. Let $c_{min}$ and $c_{max}$ be the minimum and maximum cost of a link and let $\Delta = \frac{c_{max}}{c_{min}}$ be the cost ratio of the instance.
	Let $N\coloneqq \lfloor\log_{1+\lambda}(\Delta)\rfloor+1$. The \emph{cost class} of a link $\ell\in L$ is the integer $i\in\{0,\dots,N-1\}$ with $c(\ell)\in [c_{min}\cdot (1+\lambda)^i,c_{min}\cdot (1+\lambda)^{i+1})$. 
We construct $L'$ as follows:
\begin{itemize}
	\item Add all links in $L_{\rm{snug}}$ to $L'$.
	\item Let $v,w\in V\setminus V_{\rm{snug}}$. We keep a cheapest link connecting $v$ and $w$, if it exists.
	\item Let $v\in V\setminus V_{\rm{snug}}$ and let $P=(u_0,\dots,u_t)$ be a snug path with snug shores $(S_0,\dots,S_{t+1})$. If $v\in S_0$, then for every cost class for which there is a link $\ell=\{v,u_i\}$, we keep one where $i$ is maximum. If $v\in V\setminus S_{t+1}$, then for every cost class for which there is a link $\ell=\{v,u_i\}$, we keep one where $i$ is minimum.
\end{itemize} 
By \cref{lemma:link_deepest_in,lemma:link_furthest_out}, $G+L'$ is $(k+1)$-edge-connected. By \cref{prop:lower_bound_opt_n_k}, we have 
\begin{equation}c(\OPT)\ge c_{min}\cdot \frac{n_k(G)}{2}\label{eq:bound_opt_prime}.\end{equation}
To obtain an upper bound on $c(L'\setminus L_{\rm{snug}})$, we observe that by \cref{lemma:chain_graph_subgraph}, $(V,L'\setminus L_{\rm{snug}})/\mathcal{P}^G_{\rm{chain}}$ is a loopless planar graph with $|V\setminus V_{\rm{snug}}| + |\mathcal{P}_{\rm{chain}}^G|\le 6\cdot n_k(G)$ vertices (by \cref{lemma:bound_non_snug_vertices_and_snug_paths}) and at most $N$ parallel copies per edge. As a simple planar graph with $n$ vertices cannot have more than $3n-6$ edges, this yields $|L'\setminus L_{\rm{snug}}|<18\cdot N\cdot n_k(G)$. 
Hence,
 \begin{equation}c(L'\setminus L_{\rm{snug}})\le 18\cdot N\cdot n_k(G)\cdot c_{max}\stackrel{\eqref{eq:bound_opt_prime}}{\le} 36\cdot N\cdot \Delta\cdot c(\OPT).\label{eq:bound_non_snug_links}\end{equation}
Using $\lambda\in (0,1)$, we obtain
\[\ln(1+\lambda)=\int_{1}^{1+\lambda}\frac{1}{x}dx \ge \frac{\lambda}{2}.\]
Using $\Delta \ge 1$ and $\ln (\Delta)\le \Delta$, we bound
 \begin{equation}
 N\le \log_{1+\lambda}(\Delta)+1 =\frac{\ln(\Delta)}{\ln(1+\lambda)}+1\le \frac{2\Delta}{\lambda}+\Delta\le \frac{3\Delta}{\lambda}.\label{eq:bound_N}
 \end{equation}Plugging \eqref{eq:bound_N} into \eqref{eq:bound_non_snug_links} proves \ref{item:sparse_instance}.

For \ref{item:sparse_instance_cheap}, we explain how to transform $\OPT$ into a feasible solution $F\subseteq L'$ of cost at most $(1+\lambda)\cdot c(\OPT)$.
\begin{itemize}
	\item We add all links in $L_{\rm{snug}}\cap \OPT$ to $F$.
	\item Let $v,w\in V\setminus V_{\rm{snug}}$. If $\OPT$ contains a link connecting $v$ and $w$, add a cheapest one to $F$.
	\item Let $v\in V\setminus V_{\rm{snug}}$ and let $P=(u_0,\dots,u_t)$ be a snug path with snug shores $(S_0,\dots,S_{t+1})$. If $v\in S_0$, then for every cost class for which there is a link $\ell=\{v,u_i\}\in\OPT$, we add one where $i$ is maximum to $F$. If $v\in V\setminus S_{t+1}$, then for every cost class for which there is a link $\ell=\{v,u_i\}\in\OPT$, we add one where $i$ is minimum to $F$.
\end{itemize}
The cost bound is clear by construction. Feasibility follows from \Cref{lemma:link_furthest_out,lemma:link_deepest_in}.
\end{proof}
\subsubsection{Proof of \cref{lemma:link_sets_covering_snug_chains}}
\begin{proof}[Proof of \cref{lemma:link_sets_covering_snug_chains}]
	Let $P=(u_0,\dots,u_t)\in\mathcal{P}_{\rm{chain}}^G$ and let $(S_0,\dots,S_{t+1})$ be the snug shores of $P$. We define \begin{align*}a_P&\coloneqq
		\max\{i\in\{0,\dots,t+1\}\colon \text{there is } \{u,v\}\in L \text{ with } u\in S_0 \text{ and } v\in V\!\setminus\! S_i\};\\
		b_P&\coloneqq
		\min\{i\in\{0,\dots,t+1\}\colon \text{there is } \{u,v\}\in L \text{ with } u\in V\!\setminus\! S_{t+1} \text{ and } v\in S_i\}.
	\end{align*}
	Note that we assume $(G,L,c)$ to be feasible, so $\delta_L(S_0)\ne \emptyset$ and $\delta_L(S_{t+1})\ne \emptyset$, meaning that the maximum/minimum in the definition of $a_P$/$b_P$ is taken over a nonempty set.
	Select a link $\ell_1^P$ with one endpoint in $S_0$ and one endpoint in $V\setminus S_{a_P}$ and a link $\ell_2^P$ with one endpoint in $S_{b_P}$ and the other endpoint in $V\setminus S_{t+1}$. Observe that $\ell_1^P$ covers all of the cuts $S_0,\dots,S_{a_P}$, while $\ell_2^P$ covers all of the cuts $S_{b_P},\dots,S_{t+1}$.
	Let $L'_P$ be a cheapest set of links covering all of the cuts $S_i$ for $i\in [a_P+1,b_P-1]\cap \{0,\dots,t+1\}$. Note that since $S_0\subsetneq S_1\subsetneq \dots \subsetneq S_{t+1}$, finding such a link set reduces to an interval covering problem, which can be solved efficiently by standard techniques, e.g., dynamic programming.
	Define $L_P\coloneqq L'_P\cup \{\ell_1^P,\ell_2^P\}$. 
	
	\cref{lemma:link_sets_covering_snug_chains}~\ref{item:sets_cover_paths} is satisfied by construction. It remains to show \cref{lemma:link_sets_covering_snug_chains}~\ref{item:cost_bound_L_Ps}. Using \cref{prop:lower_bound_opt_n_k} and \cref{lemma:number_snug_paths}, we obtain
	\begin{equation}
		\sum_{P\in \mathcal{P}_{\rm{chain}}^G} c(\ell_1^P)+c(\ell_2^P)\le 2\cdot | \mathcal{P}_{\rm{chain}}^G|\cdot c_{max}< 4\cdot n_k(G)\cdot \Delta\cdot c_{min}\le 8\cdot \Delta \cdot c(\OPT). \label{eq:cost_bound_ell_1_2}
	\end{equation}
	To bound the total costs of the link sets $L'_P$, we let $F_P$ denote the set of links in $\OPT$ with both endpoints in $P$. We show that for a snug path $P=(u_0,\dots,u_t)\in \mathcal{P}_{\rm{chain}}^G$, $F_P$ covers all of the cuts $S_i$ with $i\in [a_P+1,b_P-1]\cap \{0,\dots,t+1\}$. By optimality of $L'_{P}$, this will imply
	\begin{equation}
		\sum_{P\in \mathcal{P}_{\rm{chain}}^G} c(L'_P)\le \sum_{P\in \mathcal{P}_{\rm{chain}}^G} c(F_P)\le c(\OPT)	\label{eq:cost_bound_internal_links}
	\end{equation}
	and combining \eqref{eq:cost_bound_ell_1_2} and \eqref{eq:cost_bound_internal_links} completes the proof.
	
	Let $P=(u_0,\dots,u_t)\in \mathcal{P}_{\rm{chain}}^G$, let $i\in [a_P+1,b_P-1]\cap \{0,\dots,t+1\}$ and let $\ell\in \OPT$ be a link that covers $S_i$. Then $\ell=\{u,v\}$ with $u\in S_i$ and $v\in V\setminus S_i$. If $v\notin \{u_{i},\dots,u_t\}$, then $v\in V\setminus S_{t+1}$, contradicting the minimality of $b_P$. Conversely, if $u\notin \{u_0,\dots,u_{i-1}\}$, then $u\in S_0$, contradicting the maximality of $a_P$. So $\ell\in F_P$. Hence, $F_P$ covers all of the cuts $S_i$ with $i\in [a_P+1,b_P-1]\cap \{0,\dots,t+1\}$.
\end{proof}
\subsection{The PTAS}\label{sec:PTAS_WCAP}
In this section, we prove \cref{theorem:PTAS_key_theorem}, which, combined with \cref{thm:WCAP_bounded_snug_treewidth}, implies \cref{theorem:PTAS_WCAP}.
\subsubsection{Proof of \cref{theorem:PTAS_key_theorem}}

For convenience, we restate \cref{theorem:PTAS_key_theorem} below before proving it.

\theoremPTASkeytheorem*
\begin{proof}
Recall that in \cref{sec:augmentation_short}, we made the assumption that a given instance of $k$-WCAP is feasible and $G$ is minimally $k$-edge-connected.

We apply \cref{lemma:reduction_to_sparse_instances} with $\lambda=\frac{\epsilon}{3}$ to compute $\overline{L}\subseteq L$ with $L_{\rm{snug}}\subseteq \overline{L}$ such that $G+\overline{L}$ is $(k+1)$-edge-connected, \begin{equation}
c(\overline{L}\setminus L_{\rm{snug}})\le\frac{324\cdot\Delta^2}{\epsilon}\cdot c(\OPT) \text{ and } c(\overline{\OPT})\le \left(1+\frac{\epsilon}{3}\right)\cdot c(\OPT),\label{eq:bound_non_snug_links_bar_L}
\end{equation} where $\OPT$ and $\overline{\OPT}$ denote optimum solutions to $(G,L,c)$ and $(G,\overline{L},c)$, respectively. Note that $(G,\overline{L},c)$ has cost ratio at most $\Delta$. Moreover, by construction, we have $\overline{L}_{\rm{snug}}=L_{\rm{snug}}$ and $\overline{L}^\circ_{\rm{snug}}=L^\circ_{\rm{snug}}$ (assuming we use the same tie-breaking rule to define these sets).

We apply \cref{lemma:link_sets_covering_snug_chains} to obtain link sets $(L_P)_{P\in\mathcal{P}^G_{\rm{chain}}}$ such that $L_P$ covers $P$ for $P\in\mathcal{P}^G_{\rm{chain}}$ and $\sum_{P\in\mathcal{P}^G_{\rm{chain}}} c(L_P)\le (8\Delta+1)\cdot c(\OPT)$.

We use these link sets to define vertex weights $w\colon V(G/\mathcal{P}^G_{\rm{chain}})\rightarrow\mathbb{R}_{\ge 0}$ by setting $w(u_P)=c(L_P)$ for $P\in\mathcal{P}^G_{\rm{chain}}$, where $u_P$ denotes the contracted vertex corresponding to $P$, and defining $w(v)=0$ for all remaining vertices.

We let $\tilde{L}\coloneqq\overline{L}\setminus L_{\rm{snug}}\cup L^\circ_{\rm{snug}}$ and note that by \cref{cor:cost_bound_L_prime_snug} and \eqref{eq:bound_non_snug_links_bar_L}, we have
\begin{equation}
	c(\tilde{L})\le \frac{336\cdot \Delta^2}{\epsilon}\cdot c(\OPT).\label{eq:bound_cost_tilde_L}
\end{equation}
We extend the cost function $c$ to $E$ by setting $c(e)=0$ for every $e\in E$, and apply \cref{theorem:compute_k_vertex_safe_cover_with_vertex_weights} with 
\[\delta\coloneqq \frac{\epsilon^2}{6\cdot 345\cdot \Delta^2}\] to compute a $(k+1)$-edge-safe cover $\mathcal{U}'=(U'_i)_{i\in I}$ in $(G+\tilde{L})/\mathcal{P}^G_{\rm{chain}}$.  By \cref{lemma:edge_safe_is_augmentation_safe} and \cref{theorem:compute_k_vertex_safe_cover_with_vertex_weights}~\ref{item:noEdgesAcross_stronger}, we have that $\mathcal{U}'$ is a well-separated $k$-augmentation-safe cover of $(G,\tilde{L})/\mathcal{P}^G_{\rm{chain}}$. Denote the lift of $\mathcal{U}'$ to $G$ by $\mathcal{U}^\star=(U^\star_i)_{i\in I}$. We define $\mathcal{Q}_1$ and $\mathcal{Q}_3$ to consist of all snug paths of $G$ for which \ref{item:u_P_in_V_U_or_neighborhood} or \ref{item:all_snug_cuts_covered} of \cref{lemma:interaction_of_snug_shores_with_cover} applies (for $F=\tilde{L}$), respectively, and let $\mathcal{Q}\coloneqq \mathcal{Q}_1\cup\mathcal{Q}_3$. We define $L^{\star\star}\coloneqq \bigcup_{P\in\mathcal{Q}_1} L_P\cup  \tilde{L}_{\mathcal{U}^\star}$. 
By \cref{theorem:compute_k_vertex_safe_cover_with_vertex_weights}~\ref{k-vertex_safe_edge_and_vertex_costs}, \cref{lemma:link_sets_covering_snug_chains} and \eqref{eq:bound_cost_tilde_L}, we have
\begin{align}
c(L^{\star\star})&=\sum_{P\in\mathcal{Q}_1} c(L_P)+c(\tilde{L}_{\mathcal{U}^\star})\le \delta\cdot\left(\sum_{P\in\mathcal{P}^G_{\rm{chain}}} c(L_P)+c(\tilde{L})\right)\notag\\
&\le \delta\cdot\left(8\Delta +1 + \frac{336\Delta^2}{\epsilon}\right)\cdot c(\OPT)\le \delta\cdot\frac{345\Delta^2}{\epsilon}\cdot c(\OPT)\notag\\
&=\frac{\epsilon}{6}\cdot c(\OPT),\label{eq:bound_L_double_star}
	\end{align} where we used $\Delta \ge 1$ and $\epsilon\in (0,1)$ for the third inequality. By definition, $L^{\star\star}$ covers all snug paths in $\mathcal{Q}$. Let $\mathcal{U}=(U_i)_{i\in I}$ be the lift of $\mathcal{U}'$ to $G/\mathcal{Q}$. By \cref{theorem:compute_k_vertex_safe_cover_with_vertex_weights}~\ref{item:inAtMostTwo_stronger}, we have $|I|\in\mathcal{O}(|V(G)|)$, establishing \cref{theorem:PTAS_key_theorem}~\ref{item:size_of_cover}.
	
	By \cref{lemma:augmentation_safe_cover_respecting_snug_paths} (applied to $(G,\overline{L},c)$), $\mathcal{U}$ is a $k$-augmentation-safe cover of $(G,L')/\mathcal{Q}$ for any $\tilde{L}\subseteq L'\subseteq \overline{L}$. We will define a link set $L^\star$ with $\tilde{L}\subseteq L^\star\subseteq\overline{L}$ that guarantees the existence of good $k$-WCAP solutions in the subinstances given by $\mathcal{U}$, without increasing their snug-treewidth. We obtain $L^\star$ as follows:
	\begin{itemize}
		\item Add all links in $\tilde{L}$ to $L^\star$.
		\item For a snug path $P\in\mathcal{P}^G_{\rm{chain}}\setminus \mathcal{Q}$, add all links with both endpoints in $P$ to $L^\star$.
		\item For a pair of snug paths $P_1$ and $P_2$ with $\{P_1,P_2\}\not\subseteq \mathcal{Q}$, we add all links with one endpoint in $P_1$ and the other endpoint in $P_2$ to $L^\star$.
	\end{itemize}
	\begin{claim}
	We have $L^\star_{\mathcal{U}^\star}=\tilde{L}_{\mathcal{U}^\star}$.
	\end{claim}
	\begin{proof}[Proof of claim]
	For a snug path $P\in\mathcal{P}^G_{\rm{chain}}\setminus \mathcal{Q}$, we have $V(P)\cap V_{\mathcal{U}^\star}=\emptyset$, so no link with both endpoints in $P$ is contained in $L^\star_{\mathcal{U}^\star}$. Next, let $P_1$ and $P_2$ be two snug paths with $\{P_1,P_2\}\not\subseteq \mathcal{Q}$; assume w.l.o.g.\ $P_1\notin \mathcal{Q}$. Then in particular, $P_1\notin\mathcal{Q}_1$, so no vertex of $P_1$ is contained in or adjacent to $V_{\mathcal{U}^\star}$. Hence, no link connecting $P_1$ and $P_2$ is contained in $L^\star_{\mathcal{U}^\star}$.
	\end{proof}
	\begin{claim}
	$(G+L^\star)/\mathcal{Q}$ is $(k+1)$-edge-connected.
	\end{claim}
	\begin{proof}[Proof of claim]
	Our choice of $L^\star$, $L^\circ_{\rm{snug}}\subseteq\tilde{L}\subseteq L^\star$ and $\overline{L}\setminus \tilde{L}=L_{\rm{snug}}\setminus L^\circ_{\rm{snug}}$ imply that $(G+L^\star)/\mathcal{Q}$ arises from $(G+\overline{L})/\mathcal{Q}$ by removing parallel copies of links (recall that edge contractions, for us, automatically delete loops). As $G$ is $k$-edge-connected and $G+\overline{L}$ is $(k+1)$-edge-connected, the claim follows.
	 \end{proof}
		For the remainder of the proof, we introduce some notation. As in the statement of the theorem, we let \begin{align*}(G_\mathcal{Q}=(V_\mathcal{Q},E_{\mathcal{Q}}),L^\star_\mathcal{Q},c_\mathcal{Q})&\coloneqq (G,L^\star,c)/\mathcal{Q}\text{ and }\\
		(G^i_\mathcal{Q},L^i_\mathcal{Q},c^i_\mathcal{Q})&\coloneqq (G_\mathcal{Q},L^\star_\mathcal{Q},c_\mathcal{Q})/(E_\mathcal{Q}\dot{\cup}L^\star_\mathcal{Q})[V_\mathcal{Q}\setminus U_i].\end{align*} In addition, we define \[(G_\mathcal{P}=(V_\mathcal{P},E_{\mathcal{P}}),\tilde{L}_\mathcal{P},c_\mathcal{P})\coloneqq (G,\tilde{L},c)/\mathcal{P}^G_{\rm{chain}}.\]
	\begin{claim}
	For every $i\in I$, $(G^i_\mathcal{Q},L^i_\mathcal{Q},c^i_\mathcal{Q})$ has snug-treewidth $\mathcal{O}(\frac{k\Delta^2}{\epsilon^2})$, i.e., \cref{theorem:PTAS_key_theorem}~\ref{item:bounded_snug_treewidth_pieces} holds.
	\end{claim}
	\begin{proof}[Proof of claim]
	By \cref{theorem:compute_k_vertex_safe_cover_with_vertex_weights}~\ref{k-vertex_safe_stronger_tree_width}, \[\widetilde{G}^i_\mathcal{P}\coloneqq (G_\mathcal{P}+\tilde{L}_\mathcal{P})/(E_\mathcal{P}\dot{\cup}\tilde{L}_\mathcal{P})[V_\mathcal{P}\setminus U'_i]\] has treewidth at most $\frac{26k}{\delta}=\frac{53820k\Delta^2}{\epsilon^2}$. We will show that the graph $H^i_{\mathcal{Q}}$ that we obtain from $G^i_\mathcal{Q}+L^i_{\mathcal{Q}}$ by contracting all snug paths of $G^i_{\mathcal{Q}}$ arises from $\widetilde{G}^i_\mathcal{P}$ by contracting edges and links and adding parallel copies of links (again, recall that for us, edge contractions delete arising loops). As none of these operations increases the treewidth, this will imply the desired bound on the snug-treewidth of $(G^i_\mathcal{Q},L^i_\mathcal{Q},c^i_\mathcal{Q})$.

	$\widetilde{G}^i_\mathcal{P}$ arises from $G+\tilde{L}$ by
	\begin{itemize}
		\item contracting all snug paths, and
		\item contracting all edges and links in $(E\dot{\cup}\tilde{L})[V\setminus U^\star_i]$.
	\end{itemize} 
	Moreover, $H^i_{\mathcal{Q}}$ arises from $G+L^\star$ by 
	\begin{itemize}
		\item contracting all snug paths in $\mathcal{Q}$,
		\item contracting all edges and links in $(E\dot{\cup}L^\star)[V\setminus U^\star_i]$, which includes the edges of all snug paths $P\in \mathcal{P}^G_{\rm{chain}}\setminus \mathcal{Q}$ with $V(P)\subseteq V\setminus U^\star_i$ by \cref{lemma:chain_graph_subgraph} and the links in $\tilde{L}[V\setminus U^\star_i]$ because $\tilde{L}\subseteq L^\star$, and
		\item contracting all snug paths of $G^i_\mathcal{Q}$.
	\end{itemize}
	As $\mathcal{U}^\star$ arises by lifting the cover $\mathcal{U}'$ of $(G_{\mathcal{P}},\tilde{L}_\mathcal{P})$, every path $P\in \mathcal{P}^G_{\rm{chain}}$ satisfies $V(P)\subseteq U^\star_i$ or $V(P)\subseteq V\setminus U^\star_i$. By \cref{lemma:augmentation_safe_cover_respecting_snug_paths}, every $P\in \mathcal{P}^G_{\rm{chain}}\setminus \mathcal{Q}$ with $V(P)\subseteq U^\star_i$ corresponds to a subpath of a snug path in $G^i_\mathcal{Q}$. Hence, we obtain $H^i_{\mathcal{Q}}$ from $\tilde{G}^i_{\mathcal{P}}$ by 
	\begin{itemize}
		\item contracting edges and links and
		\item adding the (images under the contractions of the) links in $L^\star\setminus \tilde{L}$.
	\end{itemize}
	We have $L^\star\setminus \tilde{L}\subseteq L_{\rm{snug}}$ and $L^\circ_{\rm{snug}}\subseteq \tilde{L}$. Moreover, upon the contraction of all paths in $\mathcal{P}^G_{\rm{chain}}$, every link in $L_{\rm{snug}}$ either becomes a loop and is deleted, or it becomes parallel to a link in $L^\circ_{\rm{snug}}$. Hence, $H^i_{\mathcal{Q}}$ arises from $\widetilde{G}^i_{\mathcal{P}}$ by contracting edges and links and adding parallel copies of links as claimed.
	 \end{proof}
	It remains to prove \cref{theorem:PTAS_key_theorem}~\ref{item:combine_to_good_solution}. By \cref{lemma:solution_contract_snug_chains,lemma:glue_solutions_together_augmentation} and since $L^\star_{\mathcal{U}^\star}=\tilde{L}_{\mathcal{U}^\star}\subseteq L^{\star\star}$, we know that the solution $F$ we construct in \ref{item:combine_to_good_solution} is feasible. To prove the cost bound, we construct link sets $(F_i)_{i\in I}$ with $F_i\subseteq L^\star$ such that $F_i$ corresponds to a feasible solution for $(G^i_\mathcal{Q},L^i_\mathcal{Q},c^i_\mathcal{Q})$ for every $i\in I$, and such that 
	\[\sum_{i\in I} c(F_i)+c(L^{\star\star})\le (1+\epsilon)\cdot c(\OPT).\]
	Fix $i\in I$. We construct $F_i$ by using the optimum solution $\overline{\OPT}$ to $(G,\overline{L},c)$ as follows.
	\begin{itemize}
		\item Add all links in $\overline{\OPT}\setminus L_{\rm{snug}}$ with one endpoint in $U^\star_i$ to $F_i$.
		\item Add all links in $\overline{\OPT}\cap L_{\rm{snug}}$ that are incident to a vertex in $U^\star_i$ and that are incident to a snug path $P\in\mathcal{P}^G_{\rm{chain}}\setminus\mathcal{Q}$.
		\item For every link $\ell\in\overline{\OPT}\cap L_{\rm{snug}}$ that connects two distinct snug paths $Q_1,Q_2\in\mathcal{Q}$ with $\{u_{Q_1},u_{Q_2}\}\cap U_i\neq \emptyset$, add the link $\ell'\in L^\circ_{\rm{snug}}$ connecting $Q_1$ and $Q_2$ to $F_i$. Here $u_{Q_1}$ and $u_{Q_2}$ denote the contracted vertices corresponding to $Q_1$ and $Q_2$.
	\end{itemize}
	\begin{claim}
	$(G_\mathcal{Q}+F_i/\mathcal{Q})/(E_\mathcal{Q}\dot{\cup} L^\star_\mathcal{Q})[V_\mathcal{Q}\setminus U_i]$  is $(k+1)$-edge-connected, i.e., $F_i$ corresponds to a feasible solution to $(G^i_\mathcal{Q},L^i_\mathcal{Q},c^i_\mathcal{Q})$.
	\end{claim}
	\begin{proof}[Proof of claim]
		We note that for $i\in I$ and $\ell\in\overline{\OPT}$ without an endpoint in $U^\star_i$, the endpoints of $\ell$ are identified when contracting $\mathcal{Q}$ and $(E_\mathcal{Q}\dot{\cup} L^\star_\mathcal{Q})[V_\mathcal{Q}\setminus U_i]$: As $\ell\in\overline{\OPT}\subseteq \overline{L}$, we either have $\ell\in L^\star$, or $\ell\in L_{\rm{snug}}$. In the latter case, we know that the snug paths $P_1$ and $P_2$ containing the endpoints of $\ell$ satisfy $V(P_1)\cup V(P_2)\subseteq V\setminus U^\star_i$. (Note that we may have $P_1=P_2$.) By \cref{lemma:chain_graph_subgraph}, $P_1$ and $P_2$ are contracted when contracting $\mathcal{Q}$ and $E_{\mathcal{Q}}[V_\mathcal{Q}\setminus U_i]$. This implies that upon contracting $\mathcal{Q}$ and $E_{\mathcal{Q}}[V_\mathcal{Q}\setminus U_i]$, the endpoints of $\ell$ are identified, or $\ell$ becomes parallel to a link in $(L^\circ_{\rm{snug}}/\mathcal{Q})[V_\mathcal{Q}\setminus U_i]\subseteq L^\star_{\mathcal{Q}}[V_\mathcal{Q}\setminus U_i]$.

	Hence, we can conclude that $(G_\mathcal{Q}+F_i/\mathcal{Q})/(E_\mathcal{Q}\dot{\cup} L^\star_\mathcal{Q})[V_\mathcal{Q}\setminus U_i]$ arises from $(G_\mathcal{Q}+\overline{\OPT}/\mathcal{Q})/(E_\mathcal{Q}\dot{\cup} L^\star_\mathcal{Q})[V_\mathcal{Q}\setminus U_i]$ by deleting/exchanging parallel copies of links (recall that loops arising from edge contractions are always deleted). Hence, $(G_\mathcal{Q}+F_i/\mathcal{Q})/(E_\mathcal{Q}\dot{\cup} L^\star_\mathcal{Q})[V_\mathcal{Q}\setminus U_i]$ is $(k+1)$-edge-connected.	
	 \end{proof}

	In order to bound $\sum_{i\in I} c(F_i)$, we need to bound the number of times links in $\overline{\OPT}$ are considered when constructing the solutions $(F_i)_{i\in I}$.
	\begin{claim}
	Let $\ell\in\overline{\OPT}\setminus L_{\rm{snug}}$. Then there are at most four indices $i\in I$ such that $\ell$ has an endpoint in $U^\star_i$. Moreover, if $\ell$ has an endpoint in more than one $U^\star_i$, then $\ell\in\tilde{L}_{\mathcal{U}^\star}$.
	\end{claim}
	\begin{proof}[Proof of claim]
	By \cref{theorem:compute_k_vertex_safe_cover_with_vertex_weights}~\ref{item:inAtMostTwo_stronger}, every vertex of $G/\mathcal{P}^G_{\rm{chain}}$ is contained in at most two distinct $U'_i$s. Hence, each endpoint of $\ell$ is contained in at most two distinct $U^\star_i$s.
	
	By \cref{theorem:compute_k_vertex_safe_cover_with_vertex_weights}~\ref{item:noEdgesAcross_stronger}, $\mathcal{U}'$, and, hence, $\mathcal{U}^\star$, is well-separated. Hence, if $\ell$ has endpoints in at least two distinct $U^\star_i$s, then some endpoint of $\ell$ must be contained in $\bigcup_{i,j\in I: i\ne j} (U^\star_i\cap U^\star_j)$. As $\overline{\OPT}\setminus L_{\rm{snug}}\subseteq \overline{L}\setminus L_{\rm{snug}}=\tilde{L}\setminus L_{\rm{snug}}$, we obtain $\ell\in\tilde{L}_{\mathcal{U}^\star}$.
	\end{proof}
	\begin{claim}
	Let $\ell\in\overline{\OPT}\cap L_{\rm{snug}}$ be incident to $P\in \mathcal{P}^G_{\rm{chain}}\setminus \mathcal{Q}$. Then there is a unique $i\in I$ such that $\ell$ is incident to a vertex in $U^\star_i$.
	\end{claim}
	\begin{proof}
	By definition of $\mathcal{Q}$, the contracted vertex $u_P$ corresponding to $P$ is neither contained in nor adjacent to $\bigcup_{i,j\in I: i\ne j} (U'_i\cap U'_j)$ in $(G+\tilde{L})/\mathcal{P}^G_{\rm{chain}}$. If both endpoints of $\ell$ are contained in $P$, then $\ell$ is only incident to the unique $U^\star_i$ containing all vertices of $P$. Otherwise, as $L^\circ_{\rm{snug}}\subseteq \tilde{L}$, the graph $(G+\tilde{L})/\mathcal{P}^G_{\rm{chain}}$ contains a link parallel to $\ell$. Again, using that $\mathcal{U}'$ is well-separated, this implies that the other endpoint of $\ell$ is only contained in the unique $U^\star_i$ containing all vertices of $P$, and in no other $U^\star_j$.
	 \end{proof}
	\begin{claim}
		Let $\ell\in \overline{\OPT}\cap L_{\rm{snug}}$ connect two snug paths $Q_1,Q_2\in \mathcal{Q}$, and let $\ell'\in L^\circ_{\rm{snug}}$ be the link connecting $Q_1$ and $Q_2$. Then $c(\ell')\le c(\ell)$. Moreover, there are at most $4$ $U_i$s such that $\{u_{Q_1},u_{Q_2}\}\cap U_i\ne \emptyset$, and if there is more than one, then $\ell'\in \tilde{L}_{\mathcal{U}^\star}$.
	\end{claim}
	\begin{proof}
	$c(\ell')\le c(\ell)$ follows by definition of $L^\circ_{\rm{snug}}$. The fact that there are at most $4$ $U_i$s such that $\{u_{Q_1},u_{Q_2}\}\cap U_i\ne \emptyset$ follows from \cref{theorem:compute_k_vertex_safe_cover_with_vertex_weights}~\ref{item:inAtMostTwo_stronger}. The last point follows from the fact that $\mathcal{U}'$ and $\mathcal{U}$ are well-separated. 
	 \end{proof}
Hence, we obtain

\begin{align*}
\sum_{i\in I} c(F_i) +c(L^{\star\star})&\le c(\overline{\OPT})+3\cdot c(\tilde{L}_{\mathcal{U}^\star})+c(L^{\star\star})\stackrel{\eqref{eq:bound_L_double_star}}{\le} c(\overline{\OPT})+4\cdot c(L^{\star\star})\\
&\stackrel{\eqref{eq:bound_non_snug_links_bar_L},\eqref{eq:bound_L_double_star}}{\le} \left(1+\frac{\epsilon}{3}\right)\cdot c(\OPT)+\frac{4\epsilon}{6}\cdot c(\OPT)=(1+\epsilon)\cdot c(\OPT).
\end{align*}
\end{proof}	
\subsubsection{Proof of \cref{theorem:PTAS_WCAP}}
\begin{proof}[Proof of \cref{theorem:PTAS_WCAP}]
Let $\epsilon\in (0,1)$ and $\Delta \ge 1$.	Given an instance $(G,L,c)$ of planar $k$-WCAP with cost ratio at most $\Delta$, we apply \cref{theorem:PTAS_key_theorem} to compute link sets $L^\star$ and $L^{\star\star}$, $\mathcal{Q}\subseteq \mathcal{P}^G_{\rm{chain}}$ and a $k$-augmentation-safe cover $\mathcal{U}=(U_i)_{i\in I}$ for $(G,L^\star)/\mathcal{Q}$ as stated in the theorem. Using \cref{thm:WCAP_bounded_snug_treewidth}, we, for every $i\in I$, compute $\OPT_i\subseteq L^\star$ such that $\OPT_i/\mathcal{Q}$ is an optimum solution to the instance $(G^i_\mathcal{Q},L^i_\mathcal{Q},c^i_\mathcal{Q})$ defined in \cref{theorem:PTAS_key_theorem}. By \cref{theorem:PTAS_key_theorem}, $F\coloneqq \bigcup_{i\in I}\OPT_i\cup L^{\star\star}$ constitutes a $(1+\epsilon)$-approximate solution to the $k$-WCAP in $(G,L,c)$.
 \end{proof}

\subsection{Proof of \cref{thm:WCAP_bounded_snug_treewidth}\label{sec:bounded_snug_treewidth}}

We now provide a dynamic programming algorithm that shows \Cref{thm:WCAP_bounded_snug_treewidth}, which we repeat below for convenience.

\wcapBoundedSnugTreewidth*

Fix an instance $(G=(V,E),L,c)$ of $k$-WCAP (with root $r\in V$) that has snug-treewidth at most $\tau$.
Our dynamic program (DP) works with a tree decomposition of width $\tau$ of the graph $(V,E\cupp L)$ after contracting all snug paths $\mathcal{P}^G_{\rm{chain}}$. 
Such a decomposition exists because $(G,L,c)$ is a $k$-WCAP instance of snug-treewidth at most $\tau$.
Moreover, it is known that such a tree decomposition can be computed in linear time for constant $\tau$~\cite{bodlaenderLinearTimeAlgorithmFinding1996}.
Let $H=(W,E_H)$ be the graph obtained from $G$ after contracting the snug paths and let $L_H$ be the links obtained from $L$ after contracting the snug paths.
Hence, $E_H\subseteq E$ and $L_H\subseteq L$ are all edges and links, respectively, that do not have both endpoints in the same snug path.
To avoid confusion, we will call the vertices in $H$ \emph{nodes}.
A node $w\in W$ that corresponds to a snug path $P\in \mathcal{P}^G_{\rm{chain}}$ in $G$ is simply represented by the vertices of the snug path, i.e., $w=V(P)$.

The tree decomposition of $(W, E_H \cupp L_H)$ of width $\tau$ is a tree $(\mathcal{X},\mathcal{E})$ with the following properties.
The vertices $\mathcal{X} = \{X_i \colon i\in I\}\subseteq 2^W$ of this tree are subsets of nodes and are called \emph{bags}.
They have size at most $\tau+1$, i.e., $|X|\leq \tau+1$ for $X \in \mathcal{X}$, because the tree decomposition has width $\tau$.
Moreover, the following properties hold:
\begin{enumerate}
  \item Each node $w\in W$ is contained in at least one bag $X\in \mathcal{X}$.
  \item For each edge or link $\{w_1,w_2\}\in E_H\cup L_H$, there is at least one bag $X\in \mathcal{X}$ with $w_1\in X$ and $w_2 \in X$.
  \item For each node $w\in W$, the set of bags containing $w$ induces a subtree of $(\mathcal{X},\mathcal{E})$.
\end{enumerate}

Moreover, to simplify the presentation, we will work with what is often called a \emph{nice} tree decomposition, which has some additional properties.
Namely, there is a root bag $X_\rho$ in the tree decomposition, which has degree one in $(\mathcal{X},\mathcal{E})$.
This root bag together with the tree of the tree decomposition naturally defines a parent-child relation on the bags, where the parent of a non-root bag is the unique bag incident with it in $(\mathcal{X},\mathcal{E})$ that is closer to the root in the tree.
Additionally, each bag $X_i$ in the tree decomposition is either:
\begin{enumerate}[label=(\alph*)]
  \item a \emph{leaf bag} $X_i \in \mathcal{X}$, which has no children and contains precisely one node, i.e., $|X_i| = 1$;
  \item an \emph{introduce bag} $X_i \in \mathcal{X}$, which has precisely one child $X_j$ and its bag has one additional node compared to $X_j$, i.e., $X_i = X_j \cup \{w\}$ for some $w\in W\setminus X_j$;
  \item a \emph{forget bag} $X_i \in \mathcal{X}$, which has precisely one child $X_j$ and its bag has one node less than $X_j$, i.e., $X_i = X_j \setminus \{w\}$ for some $w\in X_j$;
  \item a \emph{join bag} $X_i \in \mathcal{X}$, which has two children $X_{i_1}$ and $X_{i_2}$ and $X_i = X_{i_1} = X_{i_2}$.
\end{enumerate}
It is known that any tree decomposition can be transformed in linear time into a nice tree decomposition of the same width and with an increase in size of at most a constant factor~\cite{kloksTreewidthComputationsApproximations1994}.

Moreover, one can assume that the root bag $X_\rho$ contains only the root $r$ of $G$ (with respect to which the snug paths are defined).
Indeed, \textcite{kloksTreewidthComputationsApproximations1994} shows how to transform any tree decomposition into a nice tree decomposition by a recursive procedure.
This procedure leads to a root bag with $\tau+1$ vertices, where $\tau$ is the treewidth.
Additionally, for any single fixed vertex $v\in V$, one can make sure that the root bag $X_\rho$ contains $v$.%
\footnote{This follows because \cite{kloksTreewidthComputationsApproximations1994} uses a step-by-step procedure based on a perfect elimination order of a graph over $V$.
The last $\tau+1$ vertices of this perfect elimination order will be the root bag of the resulting nice tree decomposition.
Because perfect elimination orders can be built greedily, and there are always at least two simplicial vertices to choose from when building them (as long as there are at least $2$ vertices left), one can choose $v$ to be among the $\tau+1$ last vertices.
}
Finally, to make sure that $X_\rho = \{r\}$, one can add a sequence of forget bags on top of the root bag obtained from the above procedure, forgetting all nodes except for $r$, and defining $X_\rho = \{r\}$ to be the new root bag.

For a bag $X_i\in \mathcal{X}$, let $V(X_i)\subseteq V$ be the set of all vertices that are contained in $X_i$, either as their own node, or as part of a snug path in $X_i$.
In other words, $V(X_i)$ is the set of all vertices in $X_i$ after uncontracting the snug paths in $X_i$.

As it is common for dynamic programs on tree decompositions, we will compute key properties for subtrees of the tree decomposition.
To this end, we use the following notation for $i\in I$:
\begin{align*}
  V_i^{\downarrow} &= \left(\bigcup_{j\in I \text{ is a descendant of }i} V(X_j)\right)\setminus V(X_i), \\
  L_i^{\downarrow} &= \left\{\ell\in L \colon \text{at least one endpoint of $\ell$ is in $V_i^{\downarrow}$}\right\},\\
  E_i^{\downarrow} &= \left\{e\in E \colon \text{at least one endpoint of $e$ is in $V_i^{\downarrow}$}\right\}.
\end{align*}

A key challenge in designing a dynamic program for our setting is the special type of nodes that correspond to snug paths.
One issue when contracting a snug path $P$ is that there are certain $k$-cuts in $G$ that do not exist anymore in $H$, namely those that cut through $P$, i.e., they contain some vertices of $P$ but not all of them.
Fortunately, snug paths are highly structured.
Namely, the only such $k$-cuts are the snug shores of $P$ that cross $P$, which we also call the \emph{internal snug shores} of $P$, and their complements (see \Cref{lemma:only_cut_intersecting_snug_edge}).
Moreover, these snug shores are all identical outside the snug path $P$.
More precisely, there is a set $W_P\subseteq W$ of nodes in $H$ so that the snug shores of $P=(u_1,\dots, u_q)$ are precisely the cuts $V(W_P), V(W_P)\cup \{u_1\}, V(W_P)\cup \{u_1,u_2\},\dots, V(W_P)\cup V(P)$.
(For a set of nodes $Y$, the set $V(Y)$ denotes the collection of vertices contained in them.)
To simplify notation, we define $V_P^j \coloneqq V(W_P) \cup \{u_1,\dots,u_j\}$ for $j\in \{0,\dots,q\}$.
In particular, $V_P^0 = V(W_P)$ and $V_P^q = V(W_P)\cup V(P)$.

The $k$-cuts $V(W_P)$ and $V(W_P)\cup V(P)$ are the only ones among the snug shores of $P$ that do not cross the snug path, and are therefore also present in the graph $H$.
Thus, in our dynamic program, we will keep track of which of the other shores, which cross $P$ (recall that we call them internal shores), are already covered by picked links.
Again, the structure of snug chains helps us to make sure that there is little to keep track of.
More precisely, any link with at least one endpoint outside $V(P)$ will cover either a prefix or suffix of the $k$-cuts $V_P^0, V_P^1,\dots, V_P^q$ (this prefix or suffix may also be all of these $k$-cuts).
Thus, there are only quadratically (in $q$) many things to keep track of per snug path in $X_i$.

\medskip

In the following, we formalize the precise information we update through the dynamic program.
To simplify the presentation, we assume that all nodes of $H$ correspond to contracted snug paths, except for the root $r$, which we used to define snug paths. Recall that $r$ is the unique vertex contained in the root bag $X_\rho$.
Nodes that do not correspond to snug paths are easier to handle, and allow for shortcutting steps in the dynamic program compared to when dealing with snug paths.
Also, often one can think of a node $w\in W$ that does not correspond to a snug path as a snug path of size one, i.e., $w=\{v\}$.
This way, we do not have to distinguish between nodes that correspond to snug paths and those that do not when describing the dynamic program, which simplifies the presentation.

When considering a bag $X_i\in \mathcal{X}$, we are interested in how a choice of links $F\subseteq L_i^{\downarrow}$ affects cuts that involve the nodes in $X_i$.
Then, as usual in dynamic programs, we will only keep track of this information, and, among all link sets that have the same impact on cuts involving the nodes in $X_i$, we will pick a link set that is cheapest.

\subsubsection{Relevant DP information for a link set}

To build up intuition and to understand what information we keep track of in our DP, consider a link set $F\subseteq L_i^{\downarrow}$ that we would like to extend to a solution to the $k$-WCAP instance $(G,L,c)$ by adding links from $L\setminus L_i^{\downarrow}$.
We now discuss what information we need from $F$ for extending it to a full solution.
This information reveals the type of DP entries we need to compute.
First, for $F$ to even admit an extension that is a $k$-WCAP solution, it has to cover all $k$-cuts in $G$ that are fully contained in $V_i^{\downarrow}$, because such cuts are not crossed by links in $L\setminus L_i^{\downarrow}$.

Moreover, to keep track of the interaction of $F$ with cuts involving vertices in $V(X_i)$, we want to save, for each node $w\in X_i$, the prefix and suffix of the $k$-cuts $V_w^1, V_w^2,\dots,V_w^{|w|-1}$ that are covered by $F$ (recall that $w\subseteq V$ corresponds to a snug path, or $|w|=1$ and this set of cuts is empty).
To represent such prefix/suffix information, we define, for any $q\in \mathbb{Z}_{\geq 1}$,
\begin{equation*}
  \mathcal{Q}(q) \coloneqq \Big\{\{1,\dots, q_1\}\cup \{q_2,\dots, q-1\} \colon q_1\in \{0,\dots, q-1\}, q_2\in \{1,\dots, q\}\Big\}.
\end{equation*}
Hence, the prefix and suffix of the $k$-cuts $V_w^1, V_w^2,\dots,V_w^{|w|-1}$ that are covered by $F$ can be represented by a set $Q^F_i(w) \in \mathcal{Q}(|w|)$, which contains the indices $j\in \{1,\dots,q-1\}$ of the cuts $V_w^j$ that are covered by $F$.
Thus, the tuple $(Q^F_i(w))_{w\in X_i}$ describes how $F$ interacts with all internal shores of snug paths in $X_i$.
The family of possible such tuples is therefore given by
\begin{equation*}
  \mathcal{Q}_i \coloneqq \prod_{w\in X_i} \mathcal{Q}(|w|).
\end{equation*}

Additionally, we also want to keep track of how $F$ interacts with cuts that involve vertices in $V(X_i)$ and that do not cross through any snug path.
To this end, we want to keep track of the following information for each nonempty $Z\subseteq X_i$:
\begin{equation*}
  \varphi_i^F(Z) \coloneqq \min\Big\{|\delta_{E_i^{\downarrow}}(S)|+|\delta_F(S)| \colon S \subseteq V_i^{\downarrow}\cup V(X_i) \text{ with }S\cap V(X_i) = V(Z) \Big\}.
\end{equation*}
In words, $\varphi_i^F(Z)$ is the smallest number of crossing edges plus links with at least one endpoint in $V_i^{\downarrow}$ of any cut in $V_i^{\downarrow}\cup V(X_i)$ that coincides with $V(Z)$ on $V(X_i)$.
One can think of these as the ``worst (or lightest) cuts'' for the current link set $F$.
We remark that we have to be a bit careful because towards the end of the DP, once $V_i^\downarrow\cup V(X_i)=V$, the vertex set $S$ attaining the minimum for $Z=V(X_i)$ might be $S=V$ and of course, for this set, we will never have any crossing edges or links. However, we point out that this issue can indeed only occur when $V_i^\downarrow\cup V(X_i)=V$ and $Z=V(X_i)$. But in this setting, for any cut $S$ with $V(X_i)\subseteq S$, its complement is fully contained in $V_i^\downarrow$ and has, thus, already been taken care of. As such, once $V_i^\downarrow\cup V(X_i)=V$, we may simply ignore the DP entries for $Z=V(X_i)$.

As we will observe in the following, knowing $(Q_i^F(w))_{w\in X_i}$  and $(\varphi_i^F(Z))_{Z\subseteq X_i: Z\neq \emptyset}$ is sufficient to determine how $F$ can be extended to a solution of the $k$-WCAP instance $(G,L,c)$.

Note that the tuple $(\varphi_i^F(Z))_{Z\subseteq X_i: Z\neq \emptyset}$ is an element of
\begin{equation*}
  \Pi_i \coloneqq \{0,\dots, |E|+|L|\}^{2^{|X_i|}\setminus \{\emptyset\}}.
\end{equation*}

For $\pi \in \Pi_i$, we say that a link set $F\subseteq L_i^{\downarrow}$ is \emph{compatible} with $\pi$ if $F$ covers all $k$-cuts contained in $V^{\downarrow}_i$ and $\pi = (\varphi_i^F(Z))_{Z\subseteq X_i: Z\neq \emptyset}$.
Analogously, for $Q\in \mathcal{Q}_i$, we say that a link set $F\subseteq L_i^{\downarrow}$ is \emph{compatible} with $Q$ if $F$ covers all $k$-cuts contained in $V^{\downarrow}_i$ and $Q = (Q^F_i(w))_{w\in X_i}$.
Finally, we say that $F$ is compatible with $(Q,\pi)$ if it is compatible with both $Q$ and $\pi$.

\subsubsection{Formal definition of DP entries}

The goal of our dynamic program is to compute, starting from the (non-root) leaves up to the root bag, the following information.
Let $X_i\in \mathcal{X}$ be a bag in the tree decomposition $(\mathcal{X},\mathcal{E})$.
Then we want to know for each tuple $(Q,\pi)\in \mathcal{Q}_i\times \Pi_i$, whether there is a link set $F\subseteq L_i^{\downarrow}$ that is compatible with $(Q,\pi)$ and, if so, we want to save a minimum cost link set $F_{(Q,\pi)}$ compatible with $(Q,\pi)$.
We call $F_{(Q,\pi)}$ a link set that \emph{realizes} the DP entry $(Q,\pi)$.
We will denote by $\nu(Q,\pi) = c(F_{(Q,\pi)})$ the cost of such a minimum cost link set $F_{(Q,\pi)}$.
If no such link set exists, we set $\nu(Q,\pi) = \infty$.

Note that the number of DP entries for a bag $X_i\in \mathcal{X}$ is $|\mathcal{Q}_i| \cdot |\Pi_i|$.
This is at most
\begin{equation*}
  |\mathcal{Q}(n)|^{|X_i|} \cdot (|E|+|L|+1)^{2^{|X_i|}-1} \le n^{2(\tau+1)} \cdot (|E|+|L|+1)^{2^{\tau+1}},
\end{equation*}
where $n=|V|$, which is polynomially bounded in the input for constant snug-treewidth $\tau$.
This sanity check shows that the data we store in our DP is polynomially bounded in the input size, as desired.

\subsubsection{From DP entries to optimal solution}

Before discussing how to compute the DP entries, we first observe that these DP entries, once computed, allow for obtaining an optimal solution to the $k$-WCAP instance $(G,L,c)$.
To this end, we look at the root bag $X_\rho = \{r\}$ of the tree decomposition.
Here we consider all the link sets that realize the tuples in $\mathcal{Q}_\rho \times \Pi_\rho$.
Recall that any such realizing link set $F\subseteq L$ covers all $k$-cuts of $G$ that are fully contained in $V_\rho^{\downarrow} = V\setminus \{r\}$.
However, these are all $k$-cuts in $G$ (in the sense that for each $k$-cut $S$, either $S\subseteq V\setminus \{r\}$ or $V\setminus S\subseteq V\setminus\{r\}$). 
Hence, it suffices to return the cheapest of these realizing link sets.
(The only reason why there may be several realizing link sets is because they have a different number of links that are incident with $r$.)

Thus, it remains to discuss how to compute, for each bag $X_i\in \mathcal{X}$, a link set realizing $(Q,\pi)$ for each $Q\in \mathcal{Q}_i$ and $\pi\in \Pi_i$.

\subsubsection{Propagation of DP entries}

We now describe how to compute the DP entries for a bag $X_i\in \mathcal{X}$ from the DP entries of its children.
First note that it is easy to compute the DP entries for a leaf bag $X_i$, which contains only a single node and has no children.
We now discuss how to deal with the other three types of bags, namely introduce, forget, and join bags.

\subsubsection*{Introduce bags}

Let $X_i$ be an introduce bag with child $X_j$.
Then $X_i = X_j \cup \{w\}$ for some node $w\in W\setminus X_j$.
Because we propagate from children to parents, we assume that we already know the DP entries for $X_j$.
We claim that all that has to be done, is to consider all link sets that realize the DP entries of $X_j$ and keep them as the link sets that realize the DP entries of $X_i$.
To this end, we first note that $V_j^\downarrow=V_i^\downarrow$, so $L^{\downarrow}_i = L^{\downarrow}_j$ and $E^{\downarrow}_i=E^{\downarrow}_j$. In particular, the collection of candidate link sets, i.e., link sets $F\subseteq L^{\downarrow}_i = L^{\downarrow}_j$ that cover all $k$-cuts in $V_j^\downarrow=V_i^\downarrow$, is the same for both $i$ and $j$. Hence, it remains to argue how the additional entries of tuples $(Q,\pi)\in \mathcal{Q}_i\times\Pi_i$ can be derived from the entries for tuples $(Q',\pi')\in \mathcal{Q}_j\times \Pi_j$ that they extend. We first handle $\pi$, which is easier to deal with.
Because $w$ was just added to $X_j$ to obtain $X_i$, there are no edges in $E_i^\downarrow$ or links in $L_i^{\downarrow}$ that have an endpoint in $w$.
Indeed, by the properties of a tree decomposition, for an edge or link to connect a vertex $v$ in $w$ to a vertex $v'\in V_i^{\downarrow}$ in some node $w'$, there must be a bag that is a descendant of $X_i$ and contains both $w$ and $w'$.
However, $w$ does not appear in any descendant of $X_i$.
Thus, there cannot be any edge or link between $v$ and $v'$.
In particular, for $F\subseteq L_i^\downarrow$ and $Z\subseteq X_j$, we have
\[|\delta_{E_i^{\downarrow}}(S)|+|\delta_F(S)| = |\delta_{E_i^{\downarrow}}(S\cup w)|+|\delta_F(S\cup w)| \text{ for every } S\subseteq V_i^\downarrow\cup V(X_j) \text{ with } S\cap V(X_j)=V(Z).\]
Hence, $\varphi_i^F(Z\cup \{w\})=\varphi_i^F(Z)=\varphi_j^F(Z)$ for every (nonempty) $Z\subseteq X_j$ and $F\subseteq L_i^\downarrow=L_j^\downarrow$.
Next, we discuss how to compute the entry $Q\in \mathcal{Q}_i$ realized by some link set $F$, assuming we know the tuple $(Q',\pi')\in\mathcal{Q}_j\times \Pi_j$ it realizes.
As there is no link in $L_i^\downarrow$ incident to a vertex in $w$, link sets in $L_i^{\downarrow}$ will either cover all snug shores of the snug path $w$ or no snug shore of it at all.
When a link covers all snug shores of a snug path, we say that it \emph{covers the snug path}.

(We recall that $W_w\subseteq W$ are the nodes in $H$ so that the snug shores of $w$ are precisely the cuts $V(W_w), V(W_w)\cup \{u_1\}, V(W_w)\cup \{u_1,u_2\},\dots, V(W_w)\cup w$, where $w=\{u_1, u_2, \dots, u_q\}$.)

\begin{lemma}\label{lem:introduce_bag}
Let $F=F_{(Q,\pi)}$ be a link set realizing some DP entry $(Q,\pi)$ for $X_j$.
Then the following two statements are equivalent:
\begin{enumerate}
  \item $F$ covers the snug path corresponding to $w$.
  \item $W_w\cap X_j\ne \emptyset$ and $\pi(W_w\cap X_j) > |\delta_{E_j^{\downarrow}}(V(W_w))|$.
\end{enumerate}
\end{lemma}

Before proving this lemma, let us interpret it.
In short, the lemma states that whether $F$ covers the snug path corresponding to $w$ is completely determined by the $\pi$ it realizes in $X_j$.
Thus, because the DP entries of $X_j$ already contain optimal/cheapest link sets for all $(Q,\pi) \in Q_j \times \Pi_j$ entries, these link sets will also be optimal for the DP entries of $X_i$.

It remains to prove the lemma.
To this end, we first show the following claims, where we denote by $W_j^{\downarrow}$ the set of all nodes in $W\setminus X_j$ in bags that are descendants of $X_j$.

\begin{claim}\label{claim:not_only_vertices_below}
We have $V(W_w\setminus W_j^{\downarrow})\ne \emptyset$.
\end{claim}
\begin{proof}
Let $w=\{u_1,\dots,u_q\}$ be the vertices in the snug path $(u_1,u_2,\dots, u_q)$ corresponding to $w$. By definition of a snug vertex, $|\delta_E(u_1)|\ge k+1>k=|\delta_E(V(W_w)\cup \{u_1\})|$, so there exists an edge from $V(W_w)$ to $u_1\in w$. As there are no edges from vertices in $V_j^\downarrow=V(W_j^\downarrow)$ to vertices in $w$, the claim follows.
\end{proof}
\begin{claim}\label{claim:unique_minimizer_below}
  $V(W_w \cap W_j^{\downarrow})$ is the unique minimizer of
  \begin{equation}\label{eq:ww_is_minimizer}
    \min_{S\subseteq V_j^{\downarrow}}|\delta_{E_j^{\downarrow}}(V(W_w\cap X_j)\cup S)|.
  \end{equation}
\end{claim}

\begin{proof}

  Note that for any $T\subseteq V$, we can partition the edges in $\delta_E(T)$ as follows into two parts, with one containing the edges in $E_j^{\downarrow}$ and one containing the edges in $E\setminus E_j^{\downarrow}$:
  \begin{equation}\label{eq:delta_partition}
    \begin{aligned}
      \delta_E(T) &= \delta_{E_j^{\downarrow}}(T) \cup \delta_{E\setminus E_j^{\downarrow}}(T)\\
                  &= \delta_{E_j^{\downarrow}}(T\cap (V_j^{\downarrow}\cup V(X_j))) \cup \delta_{E\setminus E_j^{\downarrow}}(T\setminus V_j^{\downarrow})\\
                  &= \delta_{E_j^{\downarrow}}((T\cap V(X_j))\cup (T\cap V_j^{\downarrow})) \cup \delta_{E\setminus E_j^{\downarrow}}(T\setminus V_j^{\downarrow}),
    \end{aligned}
  \end{equation}
  where the second equality holds due to the following.
  Because $E_j^{\downarrow}$ contains all edges that have at least one endpoint in $V_j^{\downarrow}$, the other endpoint must be in $V_j^{\downarrow}\cup V(X_j)$ because these are the only vertices that can appear (inside some node) in a common bag with a vertex of $V_j^{\downarrow}$. (This uses that, in a tree decomposition, the bags containing a node must be connected in the tree.)
  Hence, $\delta_{E_j^{\downarrow}}(T) = \delta_{E_j^{\downarrow}}(T\cap (V_j^{\downarrow} \cup V(X_j)))$.
  Moreover, $\delta_{E\setminus E_j^{\downarrow}}(T) =  \delta_{E\setminus E_j^{\downarrow}}(T\setminus V_j^{\downarrow})$ because the edges in $E\setminus E_j^{\downarrow}$ have both endpoints in $V\setminus V_j^{\downarrow}$.

  By counting the number of edges on each side of \Cref{eq:delta_partition}, we get
  \begin{equation}\label{eq:delta_partition_count}
    |\delta_E(T)| = |\delta_{E_j^{\downarrow}}((T\cap V(X_j))\cup (T\cap V_j^{\downarrow}))| + |\delta_{E\setminus E_j^{\downarrow}}(T\setminus V_j^{\downarrow})| \qquad \forall T\subseteq V.
  \end{equation}

  We now apply \Cref{eq:delta_partition_count} to $T = V(W_w)$, which is a $k$-cut because it is a shore of the snug path corresponding to $w$.
  Hence,
  \begin{equation}\label{eq:delta_partition_w}
    \begin{aligned}
      k &= |\delta_E(V(W_w))|\\
        &= |\delta_{E_j^{\downarrow}}(V(W_w\cap X_j)\cup V(W_w\cap W_j^{\downarrow}))| + |\delta_{E\setminus E_j^{\downarrow}}(V(W_w\setminus W_j^{\downarrow}))|.
    \end{aligned}
  \end{equation}

  To show that $V(W_w \cap W_j^{\downarrow})$ is a minimizer of \Cref{eq:ww_is_minimizer}, assume for the sake of deriving a contradiction that there is a set $S\subseteq V_j^{\downarrow}$ with
  \begin{equation*}
    |\delta_{E_j^{\downarrow}}(V(W_w\cap X_j)\cup S)| < |\delta_{E_j^{\downarrow}}(V(W_w\cap X_j)\cup V(W_w\cap W_j^{\downarrow}))|.
  \end{equation*}

  Further expanding \Cref{eq:delta_partition_w} using the above inequality, we obtain
  \begin{equation*}
    k > |\delta_{E_j^{\downarrow}}(V(W_w\cap X_j) \cup S)| + |\delta_{E\setminus E_j^{\downarrow}}(V(W_w\setminus W_j^{\downarrow}))|
    = |\delta_E(V(W_w\setminus W_j^{\downarrow})\cup S)|,
  \end{equation*}
  where the equality uses again \Cref{eq:delta_partition_count}, this time with $T=V(W_w\setminus W_j^{\downarrow})\cup S$.

  However, this implies that the cut $V(W_w\setminus W_j^{\downarrow})\cup S$ has strictly fewer than $k$ crossing edges in $G$, which is impossible because \[\emptyset\ne V(W_w\setminus W_j^{\downarrow})\subseteq V(W_w\setminus W_j^{\downarrow})\cup S\subseteq V\setminus\{u_1\}\subsetneq V\] and because $G$ is $k$-edge-connected.
  Hence, $V(W_w \cap W_j^{\downarrow})$ is indeed a minimizer of \Cref{eq:ww_is_minimizer}.
  We now discuss why it is the unique minimizer.

  For the sake of deriving a contradiction, assume that there is another minimizer $S\subseteq V_j^{\downarrow}$ of \Cref{eq:ww_is_minimizer} that is not equal to $V(W_w \cap W_j^{\downarrow})$.
  We now compare the set
  \begin{equation*}
    \overline{S} \coloneqq S \cup V(W_w\setminus W_j^{\downarrow})
  \end{equation*}
  to $V(W_w)$.
  Because both $S$ and $V(W_w\cap W_j^{\downarrow})$ are minimizers of \Cref{eq:ww_is_minimizer}, we have by \Cref{eq:delta_partition_count} that $\overline{S}$ and $V(W_w) = V(W_w\cap W_j^{\downarrow})\cup V(W_w\setminus W_j^{\downarrow})$ have the same number of edges crossing them, which is $k$ because $V(W_w)$ is a $k$-cut, i.e.,
  \begin{equation*}
    |\delta_{E}(\overline{S})| = |\delta_E(V(W_w))| = k.
  \end{equation*}
  Let $w=\{u_1,\dots,u_q\}$ be the vertices in the snug path $(u_1,\dots,u_q)$ that corresponds to $w$.
  Note that we have the following, where the second equality holds because $W_w\cup \{u_1\}$ is a snug shore of $w$:
  \begin{equation}\label{eq:alternative_shore}
    |\delta_E(\overline{S}\cup \{u_1\})| = |\delta_E(V(W_w)\cup \{u_1\})| = k.
  \end{equation}
  To see why (the first equality of) \Cref{eq:alternative_shore} holds, consider how the edge sets $\delta_E(\overline{S})$ and $\delta_E(V(W_w))$ change when adding $\{u_1\}$ to them.
  We claim that they change the same way, i.e., the edges leaving or entering the cut when adding $\{u_1\}$ are the same for both cuts. 
  Indeed, the neighbors of $u_1$ in $G$ are all contained in $V(W\setminus W_j^{\downarrow})$ because no vertex in $w$ is incident to an edge in $E_j^\downarrow$, and we have $\overline{S}\setminus V(W_j^{\downarrow}) = V(W_w \setminus W_j^{\downarrow})$ by construction. This proves \Cref{eq:alternative_shore}.
  It remains to observe that \Cref{eq:alternative_shore} immediately leads to a contradiction, because we now have two distinct $k$-cuts that contain $u_1$, but not  $u_2$, which contradicts \Cref{lemma:only_cut_intersecting_snug_edge}.

\end{proof}

We are now ready to prove \Cref{lem:introduce_bag}.

\begin{proof}[Proof of \Cref{lem:introduce_bag}]
First, consider the case where $W_w\cap X_j=\emptyset$. By \Cref{prop:minimum_cut_connected}, we know that $G[V(W_w)]$ is connected, and by \Cref{claim:not_only_vertices_below}, $V(W_w\setminus W_j^\downarrow)\ne \emptyset$. As, by the properties of a tree decomposition, there are no edges between $V(W\setminus (X_j\cup W_j^\downarrow))$ and $V(W_j^\downarrow)$, we have $V(W_w\cap W_j^\downarrow)=\emptyset$. Hence, $V(W_w)\subseteq V(W\setminus (X_j\cup W_j^\downarrow))$. But this means that no link in $L_j^\downarrow$ can be incident to $V(W_w)$, so no link set $F\subseteq L_j^\downarrow$ can cover the snug path corresponding to $w$.

Next, we assume that $W_w\cap X_j\ne \emptyset$. 
Recall that the $\pi(W_w \cap X_j)$ value realized by $F$ is given by
\begin{equation}\label{eq:pi_value_of_introduce_bag}
\begin{aligned}
  \pi(W_w \cap X_j) &= \varphi_j^F(W_w \cap X_j)\\
                    &= \min\Big\{|\delta_{E_j^{\downarrow}}(\overline{S})| + |\delta_F(\overline{S})| \colon \overline{S}\subseteq V_j^{\downarrow}\cup V(X_j)\\
                    &\hspace{20mm} \text{ with }\overline{S}\cap V(X_j) = V(W_w \cap X_j)\Big\}.
\end{aligned}
\end{equation}

Note that if we restrict ourselves to only the first term of the above objective, i.e., the term $|\delta_{E_j^{\downarrow}}(\overline{S})|$, then the minimization problem can be rewritten as
\begin{equation}\label{eq:pi_value_only_edges}
\begin{aligned}
   &\min\Big\{|\delta_{E_j^{\downarrow}}(\overline{S})| \colon \overline{S}\subseteq V_j^{\downarrow}\cup V(X_j) \text{ with }\overline{S}\cap V(X_j) = V(W_w \cap X_j)\Big\}\\
  =& 
\min\Big\{|\delta_{E_j^{\downarrow}}(V(W_w\cap X_j) \cup \widetilde{S})| \colon \widetilde{S}\subseteq V_j^{\downarrow} \Big\}
  ,
\end{aligned}
\end{equation}
where the equality follows by the substitution $\widetilde{S} = \overline{S}\setminus V(X_j)$, which can be done because the vertices $\overline{S}\cap V(X_j)$ are fixed to $V(W_w\cap X_j)$.

Hence, by \Cref{claim:unique_minimizer_below}, we know that $\widetilde{S} = V(W_w\cap W_j^{\downarrow})$ is the unique minimizer of the right-hand side of \Cref{eq:pi_value_only_edges}.
Thus, the unique minimizer of the left-hand side of \Cref{eq:pi_value_only_edges} is given by $\overline{S} = V(W_w\cap X_j) \cup V(W_w\cap W_j^{\downarrow})$,
and its value is
\begin{equation*}
  |\delta_{E_j^{\downarrow}}(V(W_w\cap X_j)\cup V(W_w \cap W_j^{\downarrow}))| = |\delta_{E_j^{\downarrow}}(V(W_w))|.
\end{equation*}

To summarize, the value of $\pi(W_w\cap X_j)$ (see \Cref{eq:pi_value_of_introduce_bag}) is always at least $|\delta_{E_j^{\downarrow}}(V(W_w))|$.
Moreover, the only candidate set $\overline{S}$ in the right-hand side of \Cref{eq:pi_value_of_introduce_bag} that can lead to the value $|\delta_{E_j}^{\downarrow}(V(W_w))|$ is $\overline{S} = V(W_w\cap X_j) \cup V(W_w\cap W_j^{\downarrow})$.
Thus, $|\pi(W_w\cap X_j)| \leq |\delta_{E_j^{\downarrow}}(V(W_w))|$ if and only if
\begin{equation*}
\delta_F(V(W_w\cap X_j) \cup V(W_w\cap W_j^{\downarrow}))=\emptyset.
\end{equation*}
The result now follows by observing that
\begin{equation*}
  \delta_F(V(W_w\cap X_j) \cup V(W_w\cap W_j^{\downarrow})) = \delta_F(V(W_w)),
\end{equation*}
and this set is empty if and only if $F$ does not cover the snug path corresponding to $w$.
(Recall that, because $F$ contains only links in $L_i^{\downarrow}$, it either covers all snug shores of the snug path corresponding to $w$ or none of them.)
Hence, $|\pi(W_w\cap X_j)| \leq |\delta_{E_j^{\downarrow}}(V(W_w))|$ if and only if $F$ does not cover the snug path corresponding to $w$.
\end{proof}

\subsubsection*{Forget bags}

Let $X_i$ be a forget bag with child $X_j$.
Then $X_i = X_j \setminus \{w\}$ for some node $w\in X_j$.
In this case, to construct link sets in $L_i^{\downarrow}$ that realize DP entries for $X_i$, we also have to consider links that connect vertices in $w$ (recall that $w$ represents the vertices of a snug path) to vertices in $X_i$, and also links between two vertices in $w$.
To this end, we would like to consider all possible extensions of realizing link sets in $X_j$ to link sets in $X_i$, to make sure that we have an optimal link set for each DP entry $(Q,\pi)$ in $X_i$.
Unfortunately, this is not possible computationally, because both $w$ and $X_i$ can have a large (linear in $|V|$) number of vertices, leading to exponentially many possible extensions of a link set $F$ realizing some DP entry of $X_j$.
However, we do not have to go through all possible extensions.
Consider a node $u\in X_i$.
The links from $w$ to $u$ are all equivalent in terms of the $k$-cuts they cover, except for the $k$-cuts corresponding to internal shores of the snug paths corresponding to $w$ and $u$.
By \Cref{lemma:link_furthest_out,lemma:link_deepest_in}, we need to consider at most two links between $w$ and $u$; namely one covering the largest prefix or suffix of its snug shores for each $u$ and $w$.
Hence, because $|X_i| = |X_j|-1 \leq \tau$ is constant, we need to add at most a constant number of links between $w$ and vertices in $X_i$ in total to $F$.
There are only polynomially many such extensions of $F$.
For each such extension, we also consider a further extension of it that covers the snug path corresponding to $w$ with a cheapest set of links having both endpoints in $w$.
This is an interval covering problem that can be solved in polynomial time.
Hence, we can go through all of these extensions for each realizing link set $F$ of a DP entry $(Q_j,\pi_j)$ for $X_j$ and keep the best one for each DP entry $(Q_i,\pi_i)$ for $X_i$. In doing so, we use the information provided by $Q_j(w)$ and $\pi_j(\{w\})$ to make sure that the extended link set we consider covers all $k$-cuts fully contained in $V_i^\downarrow=V_j^\downarrow\cup w$.

\subsubsection*{Join bags}

Let $X_i$ be a join bag with children $X_{i_1}$ and $X_{i_2}$.
Then $X_i = X_{i_1} = X_{i_2}$.
For each link set $F_1$ that realizes a DP entry $(Q_1,\pi_1)$ for $X_{i_1}$ and each link set $F_2$ that realizes a DP entry $(Q_2,\pi_2)$ for $X_{i_2}$, we can combine them to obtain a link set $F = F_1\cup F_2\in L_i^{\downarrow}$.
We consider all such combinations and keep the best one for each DP entry $(Q_i,\pi_i)$ for $X_i$.
This way we will find a realizing link set for each DP entry $(Q_i,\pi_i)$ for $X_i$.
Indeed, let $F\subseteq L_i^{\downarrow}$ be a link set of smallest cost that is compatible with the DP entry $(Q_i,\pi_i)$ for $X_i$.
Then, $F$ can be partitioned into two link sets $F_1 \coloneqq F\cap L_{i_1}^{\downarrow}$ and $F_2 \coloneqq F\cap L_{i_2}^\downarrow$, each of which is compatible with a DP entry $(Q_1,\pi_1)$ for $X_{i_1}$ and $(Q_2,\pi_2)$ for $X_{i_2}$, respectively. Note that as $k$-cuts in $G$ are connected (see \cref{prop:minimum_cut_connected}), every $k$-cut in $V_i^\downarrow=V_{i_1}^\downarrow\cup V_{i_2}^\downarrow$ is either a $k$-cut in $V_{i_1}^\downarrow$ or a $k$-cut in $V_{i_2}^\downarrow$ because by the properties of a tree decomposition, there are no edges between $V_{i_1}^\downarrow$ and $V_{i_2}^\downarrow$.
Because we have already computed the DP entries for $X_{i_1}$ and $X_{i_2}$, we have cheapest link sets $\overline{F}_1\subseteq L_{i_1}^{\downarrow}$ and $\overline{F}_2\subseteq L_{i_2}^{\downarrow}$ that are compatible with $(Q_1,\pi_1)$ and $(Q_2,\pi_2)$, respectively.
Hence, $c(\overline{F}_1) \leq c(F_1)$ and $c(\overline{F}_2) \leq c(F_2)$.
Moreover, $\overline{F}\coloneqq \overline{F}_1\cup \overline{F}_2$ is compatible with $(Q_i,\pi_i)$, because the DP entry in $X_i$ to which $\overline{F}$ corresponds is determined fully by the DP entries $(Q_1,\pi_1)$ and $(Q_2,\pi_2)$ of $\overline{F}_1$ and $\overline{F}_2$, respectively.
Hence, $\overline{F}$ is a link set compatible with $(Q_i,\pi_i)$ and has cost at most $c(F)$, and is therefore a cheapest such link set.
In other words, it realizes the DP entry $(Q_i,\pi_i)$ for $X_i$.

\FloatBarrier
\section{Hardness of planar $k$-connectivity augmentation}\label{sec:hardness}

To prove hardness of planar $k$-connectivity augmentation for $k \geq 2$, we reduce from the Linked Planar 3-SAT problem:

\begin{definition}[\cite{pilz2019planar}]
Let $G_{\phi} = (C\cup V, E_{\phi})$ be the incidence graph of a 3-SAT formula $\phi$, where $C$ is the set of clauses and $V$ is the set of variables of $\phi$.
Further, let $H$ be a Hamiltonian cycle of $C\cup V$ that first visits all elements of $C$ and then all elements of $V$.
Suppose that the union of $G_{\phi}$ and $H$ is planar.
The Linked Planar 3-SAT problem asks, given $\phi$ and $H$, whether $\phi$ is satisfiable.
\end{definition}

We use that Linked Planar 3-SAT is \NP-hard even on very structured instances:

\begin{theorem}[Theorem 10 in \cite{pilz2019planar}]\label{thm:linked_planar_sat}
The Linked Planar 3-SAT problem remains \NP-complete even if, for each clause, the edges corresponding to positive occurrences emanate to the interior of the cycle $H$, and the ones to negated occurrences to the exterior.
In addition, each variable occurs in at most three clauses.
\end{theorem}

For an instance of the Linked Planar 3-SAT problem, consisting of $\phi$ and $H$, let $H_{\phi}=(C\cup V, E_{\phi} \cup E(H))$ denote the union of $H$ and $G_{\phi}$.
See the left part of \Cref{fig:overall_hardness_construction} for an illustration of the planar graph $H_{\phi}$.
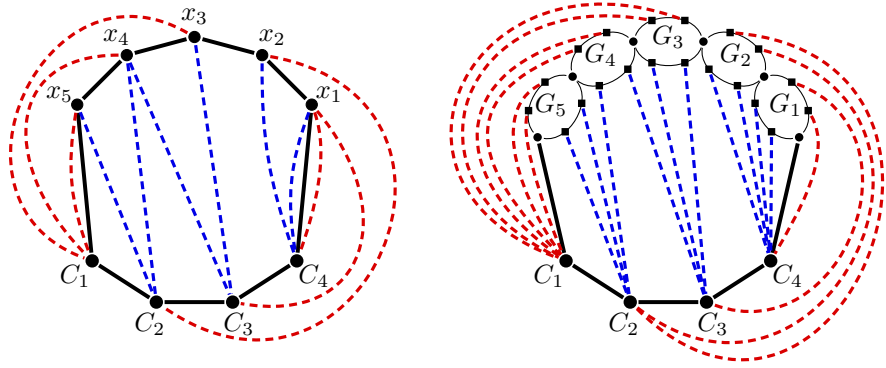
\begin{figure}
\begin{center}
\begin{tikzpicture}[
scale=0.6,
clause/.style={circle, thick,draw=black,fill=black,minimum size=2,inner sep=1.5pt, outer sep=1pt},
outer/.style={thick,draw=black,fill=black,inner sep=1pt, outer sep=1pt},
var/.style={draw=black,fill=black,circle,minimum size=2,inner sep=1.5pt, outer sep=1pt},
vertex/.style={thick,draw=white,fill=black,circle,minimum size=4,inner sep=1pt},
link/.style={very thick, densely dashed},
edge/.style={line width=1.5pt},
cut/.style={line width=2.5pt, opacity=0.5, green!70!black}
]

\useasboundingbox (-4.5,-4.3) rectangle (16,4);

\foreach \i in {1,...,5} {
  \pgfmathsetmacro\r{(\i)*(360/12)}
  \node[var] (q\i) at (\r:3) {};
  \node at (\r:3.45) {$x_{\i}$};
}
\foreach \i in {7,...,10} {
  \pgfmathsetmacro\r{(\i-0.25)*(360/11)}
  \node[clause] (q\i) at (\r:3) {};
   \pgfmathsetmacro\j{int(\i - 6)}
  \node at (\r:3.5) {$C_{\j}$};
}
\draw[edge] (q1) -- (q10);
\draw[edge] (q5) -- (q7);
\foreach \i in {1,...,4}{ 
  \pgfmathsetmacro\j{int(\i+1)}
  \draw[edge] (q\i) -- (q\j);
}
\foreach \i in {7,...,9}{ 
  \pgfmathsetmacro\j{int(\i+1)}
  \draw[edge] (q\i) -- (q\j);
}

\begin{scope}[link]
\begin{scope}[blue!90!black]
\draw (q5) -- (q8);
\draw (q4) -- (q8);
\draw (q4) -- (q9);
\draw (q3) -- (q9);

\draw (q2) to[bend right=10] (q10);
\draw (q1) to[bend right=15] (q10);
\end{scope}
\begin{scope}[red!85!black]
\draw[bend right=15] (q5) to (q7);
\draw[out=180,in=135,looseness=1.5] (q4) to (q7);
\draw[out=150,in=160,looseness=1.70] (q3) to (q7);

\draw[out=-10,in=-35,looseness=2.3] (q2) to (q8);
\draw[out=-70,in=70,looseness=0.8] (q1) to (q10);
\draw[out=-50,in=-10,looseness=1.6] (q1) to (q9);
\end{scope}
\end{scope}

\begin{scope}[shift={(10.5,0)}]
\foreach \i in {1,...,6} {
  \pgfmathsetmacro\r{(\i-0.5)*(360/12)}
  \node[vertex] (q\i) at (\r:3) {};
}
\foreach \i in {1,...,5} {
  \pgfmathsetmacro\r{(\i)*(360/12)}
  \node at (\r:3) {$G_{\i}$};
}
\draw[bend left=70] (q1) to node[pos=0.3, outer] (x1) {}  node[pos=0.7, outer] (y1) {} (q2);
\draw[bend right=70] (q1) to node[pos=0.3, outer] (z1) {}  node[pos=0.7, outer] (w1) {} (q2);
\draw[bend left=70] (q2) to node[pos=0.3, outer] (x2) {}  node[pos=0.7, outer] (y2) {} (q3);
\draw[bend right=70] (q2) to node[pos=0.3, outer] (z2) {}  node[pos=0.7, outer] (w2) {} (q3);
\draw[bend left=70] (q3) to node[pos=0.3, outer] (x3) {}  node[pos=0.7, outer] (y3) {} (q4);
\draw[bend right=70] (q3) to node[pos=0.3, outer] (z3) {}  node[pos=0.7, outer] (w3) {} (q4);
\draw[bend left=70] (q4) to node[pos=0.3, outer] (x4) {}  node[pos=0.7, outer] (y4) {} (q5);
\draw[bend right=70] (q4) to node[pos=0.3, outer] (z4) {}  node[pos=0.7, outer] (w4) {} (q5);
\draw[bend left=70] (q5) to node[pos=0.3, outer] (x5) {}  node[pos=0.7, outer] (y5) {} (q6);
\draw[bend right=70] (q5) to node[pos=0.3, outer] (z5) {}  node[pos=0.7, outer] (w5) {} (q6);

\foreach \i in {7,...,10} {
  \pgfmathsetmacro\r{(\i-0.25)*(360/11)}
  \node[clause] (q\i) at (\r:3) {};
   \pgfmathsetmacro\j{int(\i - 6)}
  \node at (\r:3.5) {$C_{\j}$};
}

\draw[edge] (q1) -- (q10);
\draw[edge] (q6) -- (q7);
\foreach \i in {7,...,9}{ 
  \pgfmathsetmacro\j{int(\i+1)}
  \draw[edge] (q\i) -- (q\j);
}

\begin{scope}[link]
\begin{scope}[blue!90!black]
\draw (x5) -- (q8);
\draw (y5) -- (q8);
\draw (y4) -- (q8);
\draw (x4) -- (q9);
\draw (x3) -- (q9);
\draw (y3) -- (q9);

\draw (x2) to (q10);
\draw (y2) to (q10);
\draw (x1) to (q10);
\draw (y1) to (q10);
\end{scope}
\begin{scope}[red!85!black]
\draw[bend right, out=290, in=-150, looseness=0.9] (z5) to (q7);
\draw[bend right, out=-30, in =190, looseness=0.5] (w5) to (q7);
\draw[out=200,in=135,looseness=1.6] (w4) to (q7);
\draw[out=190,in=145,looseness=1.7] (z4) to (q7);
\draw[out=170,in=155,looseness=1.9] (w3) to (q7);
\draw[out=160,in=160,looseness=2.1] (z3) to (q7);

\draw[out=-5,in=-45,looseness=2.5] (w2) to (q8);
\draw[out=-15,in=-38,looseness=2.1] (z2) to (q8);
\draw[out=-60,in=60,looseness=0.8] (z1) to (q10);
\draw[out=-20,in=-20,looseness=1.5] (w1) to (q9);
\end{scope}
\end{scope}

\end{scope}

\end{tikzpicture}
\end{center}
\caption{\label{fig:overall_hardness_construction}
Illustration of our hardness reduction with the graph $H_{\phi}$ on the left, and the corresponding instance of planar $k$-edge-connectivity augmentation on the right.
}
\end{figure}
We first prove the hardness of $k$-connectivity augmentation for $k=2$.
As we discuss later, this readily extends to every even $k \geq 2$.
We will then show the hardness for $k=3$, and that it extends to every odd $k \geq 3$.
First observe that we may assume without loss of generality that for every variable, both the corresponding negated literal and the positive literal appear at least once in some clause.
(Otherwise, the variable could be eliminated, maintaining all other properties required in Linked Planar 3-SAT.)
Because every variable occurs in at most three clauses, this implies that every literal appears in at most two clauses.

To describe our construction of an instance of planar $k$-edge-connectivity augmentation from a given instance of Linked Planar 3-SAT, we start with a backbone of the construction that is common to both cases $k=2$ and $k=3$.
Given an instance of Linked Planar 3-SAT, and $k\in \{2,3\}$, we construct an instance $(G,L)$ of planar $k$-edge-connectivity augmentation as follows.
Let $x_1,\dots, x_m$ be the variables of our 3-SAT instance, in the order in which they occur along the cycle $H$.
Starting with $G=H$ and $L=E_{\phi}$, we replace the $x_1$-$x_m$ path by gadgets $G_1,\dots G_m$. See the right part of \Cref{fig:overall_hardness_construction}.
Each gadget $G_i$ is given by the graph shown in \Cref{fig:variable_gadget_even}, where the two red links to the outside of the gadget replace the (up to) two links incident to $x_i$ that correspond to negated occurrences, and the two blue links to the outside of the gadget replace the (up to) two links incident to $x_i$ that correspond to positive occurrences.
(If there is just one occurrence, we simply create two copies of the corresponding edge.)
We identify the vertex $t_i$ of gadget $G_i$ with the vertex $s_{i+1}$ from the next gadget $G_{i+1}$, i.e., we have $t_{i}=s_{i+1}$ for $i=1,\dots, m-1$.
Note that the resulting graph $G$ together with the resulting link set $L$ is planar.

\begin{figure}
\begin{center}
\begin{tikzpicture}[
scale=0.7,
outer/.style={thick,draw=black,fill=black,minimum size=6,inner sep=2pt},
inner/.style={thick,draw=black,fill=white,minimum size=6,inner sep=2pt},
vertex/.style={thick,draw=white,fill=black,circle,minimum size=6,inner sep=2pt},
li/.style={very thick, dotted},
lo/.style={very thick, densely dashed},
edge/.style={line width=1.5pt},
cut/.style={line width=2.5pt, opacity=0.5, green!70!black}
]

\begin{scope}[rounded corners, cut]
\draw(1.7,-0.3) rectangle (4.3,2.3);
\draw(7.7,-0.3) rectangle (10.3,2.3);
\draw(4.7,0.3) rectangle (7.3,-2.3);
\draw(10.7,0.3) rectangle (13.3,-2.3);
\end{scope}

\begin{scope}[every node/.style={vertex}]
\node (s) at (0,0) {};
\node (v1) at (3,1) {};
\node (v2) at (9,1) {};
\node (v3) at (6,-1) {};
\node (v4) at (12,-1) {};
\node (t) at (15,0) {};
\end{scope}
\begin{scope}[every node/.style={inner}]
\node (a1) at (2,2) {};
\node (b1) at (4,2) {};
\node (c1) at (2,0) {};
\node (d1) at (4,0) {};

\node (a2) at (8,2) {};
\node (b2) at (10,2) {};
\node (c2) at (8,0) {};
\node (d2) at (10,0) {};

\node (a3) at (5,-2) {};
\node (b3) at (7,-2) {};
\node (c3) at (5,0) {};
\node (d3) at (7,0) {};

\node (a4) at (11,-2) {};
\node (b4) at (13,-2) {};
\node (c4) at (11,0) {};
\node (d4) at (13,0) {};

\node (s2) at (1,0) {};
\node (t2) at (14,0) {};
\end{scope}
\begin{scope}[every node/.style={outer}]
\node (o1) at (3,3) {};
\node (o2) at (9,3) {};
\node (o3) at (6,-3) {};
\node (o4) at (12,-3) {};
\end{scope}

\begin{scope}[edge]
\draw [bend left=45] (s) to (o1);
\draw (o1) -- (o2);
\draw [bend left=28] (o2) to (t);
\draw [bend left=45] (t) to (o4);
\draw (o4) -- (o3);
\draw [bend left=28] (o3) to (s);

\foreach \i in {1,2,3,4}{
  \draw[bend left] (v\i) to (a\i);
  \draw[bend right] (v\i) to (a\i);
  \draw[bend left] (v\i) to (b\i);
  \draw[bend right] (v\i) to (b\i);
  \draw (c\i) to (v\i);
  \draw (v\i) to (d\i);
}

\draw (s) -- (s2) -- (c1);
\draw (d1) -- (c3);
\draw (d3) -- (c2);
\draw (d2) -- (c4);
\draw (d4) -- (t2) -- (t);
\end{scope}

\begin{scope}[red!85!black]
\begin{scope}[li]
\draw (s2) -- (a1);
\draw (b1) -- (a2);
\draw (b2) -- (t2);
\draw (c1) -- (d1);
\draw (c2) -- (d2);
\draw (a3) -- (c3);
\draw (b3) -- (d3);
\draw (a4) -- (c4);
\draw (b4) -- (d4);
\end{scope}
\begin{scope}[lo]
\draw (v3) -- (o3);
\draw (v4) -- (o4);
\draw (o1) -- (3,4.3);
\draw (o2) -- (9,4.3);
\end{scope}
\node at (6,3.8) {$L_{\rm{false}}^{(i)}$};
\end{scope}
\begin{scope}[blue!90!black]
\begin{scope}[li]
\draw (s2) -- (a3);
\draw (b3) -- (a4);
\draw (b4) -- (t2);
\draw (c3) -- (d3);
\draw (c4) -- (d4);
\draw (a1) -- (c1);
\draw (b1) -- (d1);
\draw (a2) -- (c2);
\draw (b2) -- (d2);
\end{scope}
\begin{scope}[lo]
\draw (v1) -- (o1);
\draw (v2) -- (o2);
\draw (o3) -- (6,-4.3);
\draw (o4) -- (12,-4.3);
\end{scope}
\node at (9,-3.8) {$L_{\rm{true}}^{(i)}$};
\end{scope}

\node[left=2pt] at (s) {$t_i$};
\node[right=2pt] at (t) {$s_i$};

\end{tikzpicture}
\end{center}
\caption{\label{fig:variable_gadget_even}
The gadget $G_i$ for the case $k=2$.
The edges belonging to $G$ are shown as black solid lines.
The links are shown by dashed an dotted lines, where the dashed links are those links incident to outer vertices (filled squares) of degree $k=2$ in $G$, and dotted lines correspond to links incident to inner vertices (empty squares) of degree $2$ in $G$. 
We refer to the set of red links as $L^{(i)}_{\rm false}$ and to the set of blue links as $L^{(i)}_{\rm true}$.
We call the cuts indicated in green the \emph{consistency cuts} of the gadget $G_i$.
}
\end{figure}

Consider now the $k=2$ case.
Note that $G$ is a $2$-edge-connected graph.
Observe that the vertices of degree $2$ in $G$ are the clauses $C$ and the vertices shown as squares in the gadgets $G_i$.
We call the vertices of $G_i$ shown by filled squares, the \emph{outer} vertices of $G_i$, and we call the vertices shown as empty squares the \emph{inner} vertices of $G_i$.
Every gadget $G_i$ has $18$ inner vertices and $4$ outer vertices.

\begin{lemma}\label{lem:cardinality_hardness}
The instance $(G,L)$ of planar 2-connectivity augmentation has a solution of cardinality (at most) $13m$ if and only if the SAT-instance $\phi$ is satisfiable.
\end{lemma}
\begin{proof}
First, we show that every feasible solution must have cardinality at least $13m$.
Let $U^i_{\rm inner}$ and $U^i_{\rm outer}$ be the sets of inner and outer vertices of the gadget $G_i$, respectively.
Let $U_{\rm inner} \coloneqq \bigcup_{i\in [m]} U^i_{\rm inner}$ and $U_{\rm outer} \coloneqq \bigcup_{i\in [m]} U^i_{\rm outer}$.
Because every vertex in $U \coloneqq U_{\rm inner} \cup U_{\rm outer}$ has degree $k=2$ in $G$, any feasible solution $F\subseteq L$ must contain a link incident to each vertex from $U$.
Because there is no link connecting any vertex from $U_{\rm outer}$ to any other vertex from $U$,  any feasible solution $F\subseteq L$ satisfies
\[
|F| \ \geq\ \tfrac{1}{2} \cdot \big|U_{\rm inner}\big| + \big|U_{\rm outer}\big| \ =\  \tfrac{1}{2} \cdot 18m + 4m \ =\ 13m.
\]
This argument also shows that any solution of cardinality $13m$ consists of exactly one link incident to each outer vertex, i.e., each vertex from $U_{\rm outer}$, and a perfect matching on the inner vertices, i.e., the vertices from $U_{\rm inner}$.

If the SAT instance $\phi$ is satisfiable, then a solution of size $13m$ exists. 
To see this, include for each $i\in [m]$ the links from $L^{(i)}_{\rm true}$ whenever the variable $x_i$ is set to true and include the links from $L^{(i)}_{\rm false}$ otherwise, i.e., when $x_i$ is set to false.

For the reverse direction, suppose that there exists a solution $F\subseteq L$ of cardinality $13m$.
We need to show that the SAT instance $\phi$ is satisfiable.
Recall that $F$ must contain a perfect matching on the inner vertices and observe that there are no links between inner vertices of different gadgets.
Within each gadget $G_i$, there are only two different perfect matchings on the inner vertices (the red and the blue dotted matching in \Cref{fig:variable_gadget_even}), one of which is contained in $L^{(i)}_{\rm true}$ and the other is contained in $L^{(i)}_{\rm false}$.
If $F$ contains the matching from $L^{(i)}_{\rm true}$, then in order to cover the consistency cuts of the gadget $G_i$ (the green cuts in \Cref{fig:variable_gadget_even}), it must also contain the two blue edges incident two the outer vertices at the top of \Cref{fig:variable_gadget_even}. 
Thus, because every outer vertex must be incident to exactly one link from $F$, the only links connecting a vertex from $G_i$ to a clause are links from $L^{(i)}_{\rm true}$.
Analogously, if $F$ contains the matching from $L^{(i)}_{\rm false}$, then the only links connecting a vertex from $G_i$ to a clause are links from $L^{(i)}_{\rm false}$.
Because every clause vertex has degree $2$ in $G$ and must therefore have at least on incident link from $L$, we obtain a satisfying truth assignment when setting the variable $x_i$ to true if and only if $L$ contains a matching on $U^i_{\rm inner}$ from $L^{(i)}_{\rm true}$.

\end{proof}

\Cref{lem:cardinality_hardness} together with \Cref{thm:linked_planar_sat} implies the hardness of planar $k$-edge-connectivity augmentaion for $k=2$.
For any even integer $k \geq 2$, we can simply consider the instance that contains $\tfrac{k}{2}$ copies of every edge from $G$.

For $k=3$,  we include $2$ copies of each of the edges $\{t_m, C_1\}$, $\{C_1, C_2\}$, \dots, $\{C_{l-1}, C_l\}$, $\{C_l, s_1\}$ (instead of just one copy as before).
Moreover, we replace every clause $C_j$ by the gadget shown in \Cref{fig:clause_gadget}, and we use a different gadget $G_i$ for the variables (see \Cref{fig:variable_gadget_odd}).

\begin{figure}
\begin{center}
\begin{tikzpicture}[
scale=1.3,
vertex/.style={thick,draw=white,fill=black,circle,minimum size=6,inner sep=2pt},
outer/.style={thick,draw=black,fill=black,minimum size=6,inner sep=2pt, outer sep=1pt},
li/.style={very thick, dotted},
lo/.style={very thick, densely dashed},
edge/.style={line width=1.5pt}
]

\begin{scope}[every node/.style = {vertex}]
\node (a) at (-1,0) {};
\node (b) at (1,0) {};
\end{scope}
\node[outer] (c) at (0,0) {};
\node[outer] (v) at (0,1) {};
\node[outer] (w) at (0,-1) {};

\node[below left=1pt and -6pt] at (a) {$a_j$};
\node[below right=1pt and -5pt] at (b) {$b_j$};
\node[below right=0pt and -1pt] at (c) {$c_j$};

\begin{scope}[edge]
\draw (b) -- (v);
\draw (b) -- (c);
\draw (b) -- (w);
\draw[bend left=10]  (a) to (v);
\draw[bend right=10] (a) to (v);
\draw[bend left=10]  (a) to (c);
\draw[bend right=10] (a) to (c);
\draw[bend left=10]  (a) to (w);
\draw[bend right=10] (a) to (w);

\draw (a) -- (-1.5, 0.2);
\draw (a) -- (-1.5, -0.2);

\draw (b) -- (1.5, 0.2);
\draw (b) -- (1.5, -0.2);
\end{scope}

\begin{scope}[lo]
\begin{scope}[blue!90!black]
\draw (v) -- (-0.2,1.7);
\draw (v) -- (0.2,1.7);
\end{scope}
\begin{scope}[red!85!black]
\draw (w) -- (-0.2,-1.7);
\draw (w) -- (0.2,-1.7);
\end{scope}
\end{scope}
\begin{scope}[li]
\begin{scope}[green!70!black]
\draw (v) -- (c) --(w);
\end{scope}
\end{scope}

\end{tikzpicture}
\hspace{20mm}
\begin{tikzpicture}[
scale=0.65,
vertex/.style={thick,draw=white,fill=black,circle,minimum size=3,inner sep=1pt},
outer/.style={thick,draw=black,fill=black,minimum size=3,inner sep=1pt, outer sep=1pt},
li/.style={very thick, dotted},
lo/.style={very thick, densely dashed},
edge/.style={line width=1.5pt}
]

\node[vertex] (s) at (0:4) {};
\node[right] at (s) {$s_1$};
\node[vertex] (t) at (180:4) {};
\node[left] at (t) {$t_m$};
\foreach \i in {3,6,10}{
  \pgfmathsetmacro\ra{180+(\i-2)*(360/22)}
  \pgfmathsetmacro\rc{180+(\i-1)*(360/22)}
  \pgfmathsetmacro\rb{180+(\i)*(360/22)}
  
  \node[outer] (c\i) at (\rc:4) {};
  \node[vertex] (a\i) at (\ra:4) {};
  \node[vertex] (b\i) at (\rb:4) {};
  \node[outer]  (v\i) at (\rc:3) {};
  \node[outer]  (w\i) at (\rc:5) {};
  
\begin{scope}[lo]
\begin{scope}[blue!90!black]
 \pgfmathsetmacro\rvone{\rc-2}
 \pgfmathsetmacro\rvtwo{\rc+2}
\draw (v\i) -- (\rvone:2.5);
\draw (v\i) -- (\rvtwo:2.5);
\end{scope}
\begin{scope}[red!85!black]
 \pgfmathsetmacro\rvone{\rc-2}
 \pgfmathsetmacro\rvtwo{\rc+2}
\draw (w\i) -- (\rvone:5.5);
\draw (w\i) -- (\rvtwo:5.5);
\end{scope}
\end{scope}
\begin{scope}[li]
\begin{scope}[green!70!black]
\draw (v\i) -- (c\i) --(w\i);
\end{scope}
\end{scope}
}

\begin{scope}[edge]

\foreach \i in {3,6,10}{
\draw (b\i) -- (v\i);
\draw (b\i) -- (c\i);
\draw (b\i) -- (w\i);
\draw[bend left=10]  (a\i) to (v\i);
\draw[bend right=10] (a\i) to (v\i);
\draw[bend left=10]  (a\i) to (c\i);
\draw[bend right=10] (a\i) to (c\i);
\draw[bend left=10]  (a\i) to (w\i);
\draw[bend right=10] (a\i) to (w\i);
}

\draw[bend left=10] (t) to (a3);
\draw[bend right=10] (t) to (a3);

\draw[bend left=10] (b3) to (a6);
\draw[bend right=10] (b3) to (a6);

\draw[bend left=10] (b10) to (s);
\draw[bend right=10] (b10) to (s);

  \pgfmathsetmacro\rl{180+(6.5)*(360/22)}
  \pgfmathsetmacro\rr{180+(7.5)*(360/22)}
  
 \draw (b6) -- (\rl:3.9);
 \draw (b6) -- (\rl:4.1);
 
 \draw (a10) -- (\rr:3.9);
 \draw (a10) -- (\rr:4.1);
 
  \pgfmathsetmacro\rp{180+(7)*(360/22)}
  \node at (\rp:4) { $\dots$};

\end{scope}

\end{tikzpicture}
\end{center}
\caption{\label{fig:clause_gadget}
The gadget by which we replace a clause $C_j$ (left) and an illustration of the graph resulting from the path with vertices $t_m, C_1, C_2, \dots, C_l, s_1$ (right) in the hardness reduction for $k=3$.
}
\end{figure}
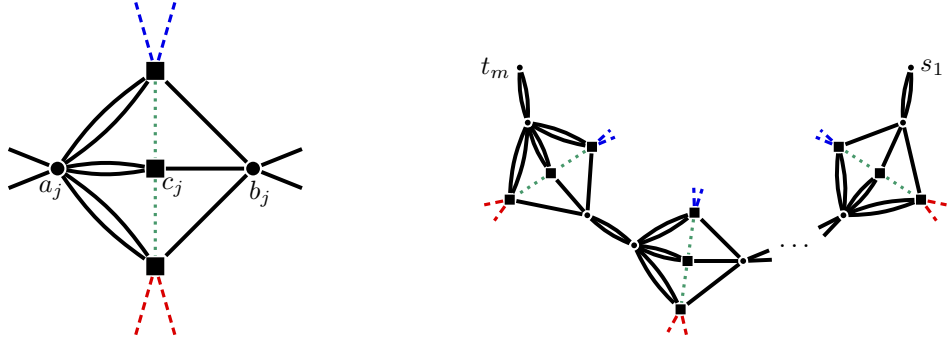

\begin{figure}
\begin{center}
\begin{tikzpicture}[
xscale=1.2, yscale=0.8,
outer/.style={thick,draw=black,fill=black,minimum size=6,inner sep=2pt},
inner/.style={thick,draw=black,fill=white,minimum size=6,inner sep=2pt, outer sep=1pt},
vertex/.style={thick,draw=white,fill=black,circle,minimum size=6,inner sep=2pt},
li/.style={very thick, dotted},
lo/.style={very thick, densely dashed},
edge/.style={line width=1.5pt},
cut/.style={line width=2.5pt, opacity=0.5, green!70!black}
]

\foreach \i in {1,...,22} {
  \pgfmathsetmacro\r{(\i-1)*(360/22)}
  \coordinate (c\i) at (\r:4) {};
}

\begin{scope}[rounded corners, cut]
\foreach \i in {3,9,14,20} {
 \pgfmathsetmacro\a{(\i-1)*(360/22)}
  \pgfmathsetmacro\b{(\i-2.3)*(360/22)}
 \pgfmathsetmacro\c{(\i-1.5)*(360/22)}
 \pgfmathsetmacro\d{(\i -0.5)*(360/22)}
  \pgfmathsetmacro\e{(\i + 0.5)*(360/22)}
  \pgfmathsetmacro\f{(\i + 1.3)*(360/22)}
  \draw (\a:1) -- (\b:4.2) -- (\c:4.3)  -- (\d: 4.3) -- (\e:4.3) -- (\f : 4.2)-- cycle;
  }
\end{scope}

\begin{scope}[every node/.style={vertex}]
\foreach \i in {3,9,14,20} {
  \pgfmathsetmacro\r{(\i-1)*(360/22)}
  \node (q\i) at (\r:1.5) {};
}
\node (v0) at (0,0) {};
\node (v12) at (c12) {};
\node (v3) at (c3) {};
\node (v9) at (c9) {};
\node (v14) at (c14) {};
\node (v20) at (c20) {};
\node (v1) at (c1) {};
\end{scope}
\begin{scope}[every node/.style={inner}]
\node (s2) at (180:3) {};
\node (t2) at (0:3) {};
\foreach \i in {3,9,14,20} {
  \pgfmathsetmacro\r{(\i-1)*(360/22)}
  \node (p\i) at (\r:3) {};
}

\node (v2) at (c2) {};
\node (v4) at (c4) {};
\node (v5) at (c5) {};
\node (v6) at (c6) {};
\node (v7) at (c7) {};
\node (v8) at (c8) {};
\node (v10) at (c10) {};
\node (v11) at (c11) {};
\node (v13) at (c13) {};
\node (v15) at (c15) {};
\node (v16) at (c16) {};
\node (v17) at (c17) {};
\node (v18) at (c18) {};
\node (v19) at (c19) {};
\node (v21) at (c21) {};
\node (v22) at (c22) {};
\end{scope}
\begin{scope}[every node/.style={outer}]
\node (o1) at (4,3.5) {};
\node (o2) at (-4,3.5) {};
\node (o3) at (-4,-3.5) {};
\node (o4) at (4,-3.5) {};
\end{scope}

\begin{scope}[edge]
\draw (v1) -- (v2) -- (v3) -- (v4) -- (v5) -- (v6) -- (v7) -- (v8) -- (v9) -- (v10) -- (v11) -- (v12)
-- (v13) -- (v14) -- (v15) -- (v16) -- (v17) -- (v18) -- (v19) -- (v20) -- (v21) -- (v22)  -- (v1);

\draw (t2) -- (v0) -- (s2);
\draw[bend left] (t2) to (v1);
\draw[bend right] (t2) to (v1);
\draw[bend left] (s2) to (v12);
\draw[bend right] (s2) to (v12);

\foreach \i in {3,9,14,20} {
  \draw (v0) -- (q\i) -- (p\i);
  \draw[bend left] (p\i) to (v\i);
  \draw[bend right] (p\i) to (v\i);
}
\draw (v2) -- (q3);
\draw (v4) -- (q3);
\draw (v5) -- (q3);
\draw (v6) -- (v0);
\draw (v7) -- (v0);
\draw (v8) -- (q9);
\draw (v10) -- (q9);
\draw (v11) -- (q9);
\draw (v13) -- (q14);
\draw (v15) -- (q14);
\draw (v16) -- (q14);
\draw (v17) -- (v0);
\draw (v18) -- (v0);
\draw (v19) -- (q20);
\draw (v21) -- (q20);
\draw (v22) -- (q20);

\draw[bend right=12] (v1) to (o1);
\draw[bend right=20] (v1) to (o1);
\draw[bend right=12] (o2) to (v12);
\draw[bend right=20] (o2) to (v12);
\draw[bend right=12] (v12) to (o3);
\draw[bend right=20] (v12) to (o3);
\draw[bend right=12] (o4) to (v1);
\draw[bend right=20] (o4) to (v1);
\draw[bend right] (o1) to (o2);
\draw[bend right] (o3) to (o4);
\end{scope}

\begin{scope}[red!85!black]
\node at (4.8,4.0) {$L_{\rm{false}}^{(i)}$};
\begin{scope}[li]
\draw (t2) -- (v2);
\draw (p3) -- (v4);
\draw[bend left] (v5) to (v6);
\draw[bend left] (v7) to (v8);
\draw (p9) -- (v10);
\draw (v11) -- (s2);

\draw (v13) -- (p14);
\draw[bend left] (v15) to (v16);
\draw[bend left] (v17) to (v18);
\draw (v19) -- (p20);
\draw[bend left] (v21) to (v22);
\end{scope}
\begin{scope}[lo]
\draw (o3) -- (v14);
\draw (o4) -- (v20);
\draw (o1) -- (4,4.8) {};
\draw (o2) -- (-4,4.8) {};
\end{scope}
\end{scope}

\begin{scope}[blue!90!black]
\node at (4.8,-4.0) {$L_{\rm{true}}^{(i)}$};
\begin{scope}[li]
\draw (v2) -- (p3);
\draw[bend left] (v4) to (v5);
\draw[bend left] (v6) to (v7);
\draw (v8) -- (p9);
\draw[bend left] (v10) to (v11);
\draw (s2) -- (v13);
\draw (p14) -- (v15);
\draw[bend left] (v16) to (v17);
\draw[bend left] (v18) to (v19);
\draw (p20) -- (v21);
\draw (v22) -- (t2);
\end{scope}
\begin{scope}[lo]
\draw (o1) -- (v3);
\draw (o2) -- (v9);

\draw (o3) -- (-4,-4.8) {};
\draw (o4) -- (4,-4.8) {};
\end{scope}
\end{scope}

\node[left=2pt] at (v12) {$t_i$};
\node[right=2pt] at (v1) {$s_i$};

\end{tikzpicture}
\end{center}
\caption{\label{fig:variable_gadget_odd}
The gadget $G_i$ for the case $k=3$.
The edges belonging to $G$ are shown as black solid lines.
The links are shown by dashed and dotted lines, where the dashed links are those links incident to outer vertices (filled squares) of degree $3$ in $G$, and dotted lines correspond to links incident to inner vertices (empty squares) of degree $3$ in $G$. 
We refer to the set of red links as $L^{(i)}_{\rm false}$ and to the set of blue links as $L^{(i)}_{\rm true}$.
We call the cuts indicated in green the \emph{consistency cuts} of the gadget $G_i$.
}
\end{figure}
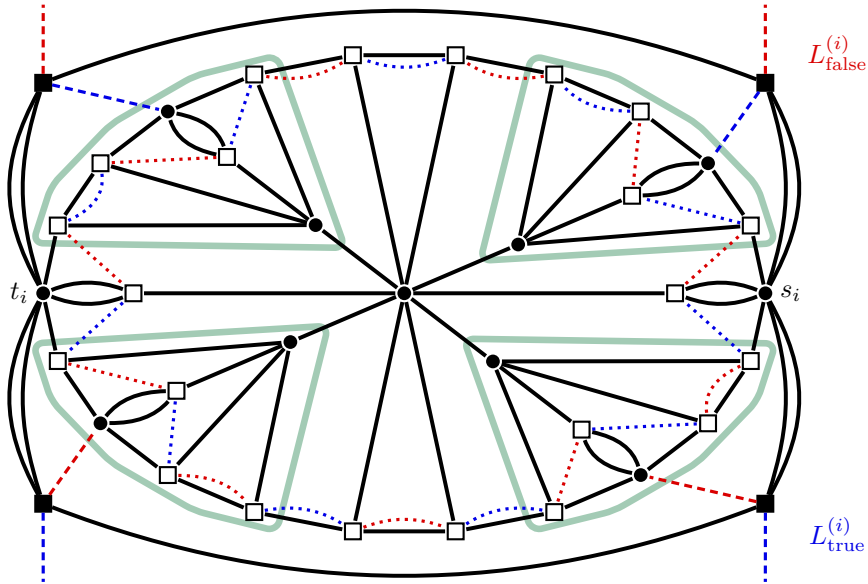

\begin{lemma}\label{lem:cardinality_hardness_k3}
The instance $(G,L)$ of planar 3-connectivity augmentation has a solution of cardinality (at most) $15m+ l$ if and only if the SAT-instance $\phi$ is satisfiable, where $m$ is the number of variables of $\phi$ and $l$ the number of clauses.
\end{lemma}
\begin{proof}
Observe that within each gadget replacing a clause $C_j$, every solution $F\subseteq L$ needs to choose at least one link incident to the vertex $c_j$.
Moreover, if $L$ includes at least one link connecting the clause gadget to a variable gadget $G_i$, then one of the two links incident to $c_j$ is sufficient to cover the 3-cuts within the clause gadget.

Therefore, an analogous argument to the proof of \Cref{lem:cardinality_hardness} implies the claimed statement.

\end{proof}

\Cref{lem:cardinality_hardness_k3} together with \Cref{thm:linked_planar_sat} implies the hardness of planar $k$-edge-connectivity augmentaion for $k=3$.
For any odd integer $k \geq 3$, we add $\tfrac{k-3}{2}$ copies of the edges of some $2$-edge-connected spanning subgraph $\tilde{G}$ of $G$.
We choose the subgraph $\tilde{G}$ such that it contains exactly $2$ edges in every $3$-cut of $G$.
Then every cut in the resulting graph has size at least $3+ \tfrac{k-3}{2} \cdot 2 = k$ and a cut has size exactly $k$ if and only if the corresponding cut in $G$ has size $3$.
Thus, the minimum cuts are preserved, leading to an equivalent instance.
A possible choice of the graph $\tilde{G}$ is shown in \Cref{fig:boost_connectivity}, which shows the edges of $\tilde{G}$ in each variable gadget, and \Cref{fig:boost_connectivity_2}, which shows all other edges of $\tilde{G}$.

\begin{figure}
\begin{center}
\begin{tikzpicture}[
xscale=1.2, yscale=0.8,
outer/.style={thick,draw=black,fill=black,minimum size=6,inner sep=2pt},
inner/.style={thick,draw=black,fill=white,minimum size=6,inner sep=2pt, outer sep=1pt},
vertex/.style={thick,draw=white,fill=black,circle,minimum size=6,inner sep=2pt},
li/.style={very thick, dotted},
lo/.style={very thick, densely dashed},
edge/.style={line width=1.5pt},
cut/.style={line width=2.5pt, opacity=0.5, green!70!black}
]

\foreach \i in {1,...,22} {
  \pgfmathsetmacro\r{(\i-1)*(360/22)}
  \coordinate (c\i) at (\r:4) {};
}

\begin{scope}[rounded corners, cut]
\foreach \i in {3,9,14,20} {
 \pgfmathsetmacro\a{(\i-1)*(360/22)}
  \pgfmathsetmacro\b{(\i-2.3)*(360/22)}
 \pgfmathsetmacro\c{(\i-1.5)*(360/22)}
 \pgfmathsetmacro\d{(\i -0.5)*(360/22)}
  \pgfmathsetmacro\e{(\i + 0.5)*(360/22)}
  \pgfmathsetmacro\f{(\i + 1.3)*(360/22)}
  \draw (\a:1) -- (\b:4.2) -- (\c:4.3)  -- (\d: 4.3) -- (\e:4.3) -- (\f : 4.2)-- cycle;
  }
  
  \foreach \i in {3,9,14,20} {
 \pgfmathsetmacro\a{(\i-0.8)*(360/22)}
  \pgfmathsetmacro\b{(\i-1.2)*(360/22)}
  \draw (\a:2.8) -- (\b:2.8) -- (\b:4.15) -- (\a: 4.15) -- cycle;
  }
\end{scope}

\begin{scope}[every node/.style={vertex}]
\foreach \i in {3,9,14,20} {
  \pgfmathsetmacro\r{(\i-1)*(360/22)}
  \node (q\i) at (\r:1.5) {};
}
\node (v0) at (0,0) {};
\node (v12) at (c12) {};
\node (v3) at (c3) {};
\node (v9) at (c9) {};
\node (v14) at (c14) {};
\node (v20) at (c20) {};
\node (v1) at (c1) {};
\end{scope}
\begin{scope}[every node/.style={inner}]
\node (s2) at (180:3) {};
\node (t2) at (0:3) {};
\foreach \i in {3,9,14,20} {
  \pgfmathsetmacro\r{(\i-1)*(360/22)}
  \node (p\i) at (\r:3) {};
}

\node (v2) at (c2) {};
\node (v4) at (c4) {};
\node (v5) at (c5) {};
\node (v6) at (c6) {};
\node (v7) at (c7) {};
\node (v8) at (c8) {};
\node (v10) at (c10) {};
\node (v11) at (c11) {};
\node (v13) at (c13) {};
\node (v15) at (c15) {};
\node (v16) at (c16) {};
\node (v17) at (c17) {};
\node (v18) at (c18) {};
\node (v19) at (c19) {};
\node (v21) at (c21) {};
\node (v22) at (c22) {};
\end{scope}
\begin{scope}[every node/.style={outer}]
\node (o1) at (4,3.5) {};
\node (o2) at (-4,3.5) {};
\node (o3) at (-4,-3.5) {};
\node (o4) at (4,-3.5) {};
\end{scope}

\begin{scope}[edge]
\draw (v1) -- (v2) -- (v3);
\draw (v4) -- (v5) -- (v6) -- (v7) -- (v8) -- (v9);
\draw (v10) -- (v11) -- (v12) -- (v13) -- (v14);
\draw (v15) -- (v16) -- (v17) -- (v18) -- (v19);
\draw (v20) -- (v21) -- (v22)  -- (v1);

\draw (t2) -- (v0) -- (s2);
\draw (t2) to (v1);
\draw (s2) to (v12);

\foreach \i in {3,9,14,20} {
  \draw (q\i) -- (p\i);
  \draw (p\i) to (v\i);
}

\draw (v4) -- (q3);
\draw (v10) -- (q9);
\draw (v15) -- (q14);
\draw (v19) -- (q20);

\draw[bend right=20] (v1) to (o1);
\draw[bend right=20] (o2) to (v12);
\draw[bend right=20] (v12) to (o3);
\draw[bend right=20] (o4) to (v1);
\draw[bend right] (o1) to (o2);
\draw[bend right] (o3) to (o4);
\end{scope}

\node[left=2pt] at (v12) {$s_i$};
\node[right=2pt] at (v1) {$t_i$};

\end{tikzpicture}
\end{center}
\caption{\label{fig:boost_connectivity}
The set of edges of the gadget $G_i$ of which we take $h=\frac{k-3}{2}$ copies for odd $k\geq 3$.
The green cuts indicate all non-singleton cuts of size $3$ in the gadget $G_i$ we use for $k=3$.
(The vertices of degree $3$ are shown as squares.)
When adding $h$ copies of the shown edges, each of these cuts that had size $3$ in $G$, now has size $3+2h=k$, i.e, all minimum cuts are preserved.
}
\end{figure}
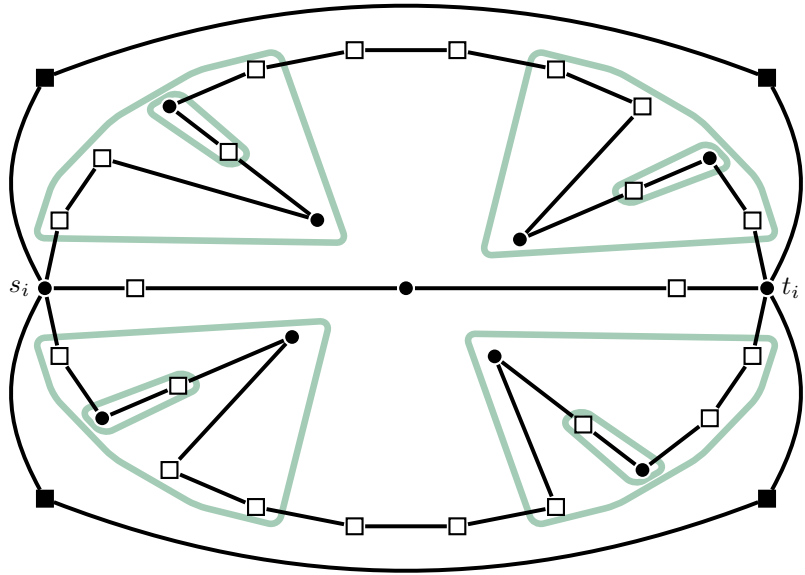

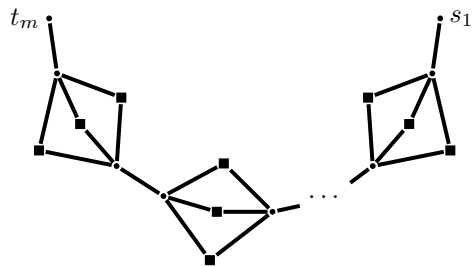
\begin{figure}
\begin{center}
\begin{tikzpicture}[
scale=0.65,
vertex/.style={thick,draw=white,fill=black,circle,minimum size=3,inner sep=1pt},
outer/.style={thick,draw=black,fill=black,minimum size=3,inner sep=1pt, outer sep=1pt},
li/.style={very thick, dotted},
lo/.style={very thick, densely dashed},
edge/.style={line width=1.5pt}
]

\node[vertex] (s) at (0:4) {};
\node[right] at (s) {$s_1$};
\node[vertex] (t) at (180:4) {};
\node[left] at (t) {$t_m$};
\foreach \i in {3,6,10}{
  \pgfmathsetmacro\ra{180+(\i-2)*(360/22)}
  \pgfmathsetmacro\rc{180+(\i-1)*(360/22)}
  \pgfmathsetmacro\rb{180+(\i)*(360/22)}
  
  \node[outer] (c\i) at (\rc:4) {};
  \node[vertex] (a\i) at (\ra:4) {};
  \node[vertex] (b\i) at (\rb:4) {};
  \node[outer]  (v\i) at (\rc:3) {};
  \node[outer]  (w\i) at (\rc:5) {};

}

\begin{scope}[edge]

\foreach \i in {3,6,10}{
\draw (b\i) -- (v\i);
\draw (b\i) -- (c\i);
\draw (b\i) -- (w\i);
\draw (a\i) to (v\i);
\draw (a\i) to (c\i);
\draw (a\i) to (w\i);
}

\draw (t) to (a3);
\draw(b3) to (a6);
\draw (b10) to (s);

  \pgfmathsetmacro\rl{180+(6.5)*(360/22)}
  \pgfmathsetmacro\rr{180+(7.5)*(360/22)}
  
 \draw (b6) -- (\rl:4);
 \draw (a10) -- (\rr:4);
 
  \pgfmathsetmacro\rp{180+(7)*(360/22)}
  \node at (\rp:4) { $\dots$};

\end{scope}

\end{tikzpicture}
\end{center}
\caption{\label{fig:boost_connectivity_2}
The edges of $\tilde{G}$ outside of the variable gadgets.
}
\end{figure}

This completes the proof of the hardness of $k$-edge-connectivity augmentation for each $k\geq 2$.
\FloatBarrier

\paragraph{Data availability statement required per UKRI open access policy:} No data are associated with this article. Data sharing is not applicable to this article.
\FloatBarrier
 \begingroup
 \sloppy
 \emergencystretch=1em
\printbibliography
\endgroup

\FloatBarrier

\newpage
\appendix

\section{Discussion on bidimensionality and ubiquity frameworks}\label{sec:bidimensionality_and_ubiquity}
In this section, we provide a more detailed explanation of why the theory of bidimensionality developed in \cite{demaineBidimensionalityNewConnections2005}, as well as the ubiquity framework introduced in \cite{cohen-addadApproximatingConnectivityDomination2016} and extended in \cite{fominApproximationSchemesWidth2019}, are not applicable in our setting.

\subsection{Bidimensionality}

The bidimensionality framework of \cite{demaineBidimensionalityNewConnections2005} requires problems to have what they call the separation property, which is not fulfilled by $k$-ECSS, $k$-VCSS, or $k$-CAP, even when restricted to planar graphs.
This separation property in particular requires the following:

\begin{quote}
Given any graph $G$, a vertex cut $C$, and an optimal solution $\OPT$ in $G$, then for any union $G'$ of some subset of connected components of $G-C$, we must have $|\OPT\cap G'| = |\OPT(G')| \pm \mathcal{O}(|C|)$.
\end{quote}

To see why the above property does not hold for $k$-ECSS, let $k$ be even and consider a large cycle $G=(V,E)$ with $k$ parallel edges between every pair of vertices that are consecutive on the cycle.
The optimal $k$-ECSS solution takes $k/2$ parallel edges between every pair of consecutive vertices on the cycle, and hence $|\OPT| = |V|\cdot \sfrac{k}{2}$.
Now, consider a vertex cut $C$ consisting of any two non-adjacent vertices.
The graph $G-C$ has two connected components, each being a path with $k$ parallel edges between any consecutive vertices on the path.
Hence, each connected component of $G-C$ only has a single $k$-ECSS solution, which contains all edges of the instance, i.e., $k$ parallel edges between every pair of consecutive vertices on the path.
Hence, if we let $G'$ be both connected components of $G-C$, we need $|\OPT(G')| = (|V|-2) k$ many edges to be $k$-edge-connected in both components.
However, $|\OPT \cap G'| = (|V|-2) \cdot \sfrac{k}{2}$.
Thus the difference between $|\OPT(G')|$ and $|\OPT\cap G'|$ is $(|V|-2)\cdot \sfrac{k}{2}$, which cannot be bounded by $\mathcal{O}(|C|) = \mathcal{O}(1)$.

It is not hard to adapt the above example to $k$-VCSS, obtaining the same issue.
For example, one can start with a cycle and connect each vertex to all vertices of distance at most $k$ on the cycle.
An optimal $k$-VCSS solution consists of including, for each vertex, all edges to the $\sfrac{k}{2}$ closest vertices in each direction on the cycle, which needs $|\OPT| = |V| \cdot \sfrac{k}{2}$ edges altogether.
However, after removing a vertex set $C$ of two segments of $k$ consecutive vertices each, every connected component of $G-C$ needs all remaining edges to be $k$-vertex-connected, which is $k \cdot |V| - \mathcal{O}(k)$ many edges.
Note that for $k=2$, this construction leads to a planar graph, showing that also planarity does not help to fulfill the separation property for $k$-VCSS.
(We recall, that planar $k$-VCSS is only relevant for $k\leq 5$ because a planar graph cannot be $k$-vertex-connected for $k\geq 6$.)

Also $k$-CAP does not fulfill the above-highlighted separation property.
For example, for $1$-CAP, one can get counterexamples by considering instances on a spider with a central vertex $c$ and an even number of paths of length $3$ attached to $c$.
In terms of links, one can think of having an unweighted instance with one link parallel to each edge of the spider, and, additionally, a set of links consisting of a perfect matching on the leaves of the spider.
By setting $C=\{c\}$, one can see that the separation property is violated.
Indeed, in the original instance it suffices to take the perfect matching on the leaves to ensure $2$-edge-connectivity, while in each connected component of $G-C$, one link is necessary to obtain $2$-edge-connectivity.
This leads to twice as many links in total compared to the original solution.

Moreover, note that for $k$-CAP, a subinstance of $G$ obtained by removing a vertex cut $C$ may not be a feasible instance anymore, even if the original instance is feasible.
For $k=1$, this can be addressed as follows. The impact of a link from $G$ to a connected component of $G-C$, which is a subtree, is easily described by a ``shorter'' link within the corresponding connected component in $G-C$, which is also called a \emph{shadow link} (see, e.g., \cite{adjiashviliBeatingApproximationFactor2018}).
However, it is unclear how to do this for $k\geq 2$.
Hence, even defining a clean notion of subinstance in such cases is not straightforward.

\subsection{Ubiquity}
In \cite{cohen-addadApproximatingConnectivityDomination2016}, the authors describe a framework to derive PTASs for special types of minimization problems on vertex- or edge-weighted bounded genus graphs. For the comparison to our work, we will focus on the edge-weighted case. The authors of \cite{cohen-addadApproximatingConnectivityDomination2016} consider minimization problems on edge-weighted graphs $(G,w)$ that satisfy the following properties:
\begin{enumerate}[(a)]
	\item \label{ubiquity:connected_and_treewidth} Every feasible edge set $S$ is connected. There is a constant $t$ such that $G/S$ has treewidth at most $t$ for every feasible $S$. (They call a problem with this property \emph{$t$-ubiquitous}.)
	\item \label{ubiquity:constant_factor} The problem admits a constant factor approximation.
	\item \label{ubiquity:FPT_PTAS} For every fixed $\epsilon > 0$, there is an algorithm, that, in time $\mathrm{poly}(n)\cdot 2^{\mathcal{O}(t)}$, returns a $(1+\epsilon)$-approximate solution to an instance with $n$ vertices and treewidth $t$.\footnote{The authors phrase this property in terms of the branchwidth. However, treewidth and branchwidth are within a constant factor of each other~\cite{ROBERTSON1991153}.}
	\item \label{ubiquity:lift} There exist a constant $\beta$ and a polynomial time algorithm that, given an instance $G$, a subgraph $K$ of $G$ and a feasible solution $S$ to $G/E(K)$, computes a solution of costs at most $w(S)+\beta\cdot w(E(K))$ for $G$. In particular, it is assumed that the class of feasible input instances is closed under edge contractions. 
\end{enumerate}
For planar graphs, the overall approach is to iteratively construct circular separators, that, together, yield a ``cheap'' subgraph $K$ whose contraction results in a graph of treewidth $\mathcal{O}(\log (n))$. For the contracted instance $G/K$, property \ref{ubiquity:FPT_PTAS} is used to obtain a $(1+\epsilon)$-approximate solution. Property \ref{ubiquity:lift} is invoked to complete it to a good solution for the whole instance.

Properties \ref{ubiquity:connected_and_treewidth} and \ref{ubiquity:constant_factor} are clearly satisfied by solutions to the $k$-WECSS and the $k$-WVCSS. Moreover, they can be achieved for $k$-WCAP when interpreting it as a special case of the $(k+1)$-WECSS where all non-link edges have cost $0$. However, issues arise with properties \ref{ubiquity:FPT_PTAS} and \ref{ubiquity:lift}. While polynomial-time algorithms for the $k$-(W)ECSS and the $k$-(W)VCSS in bounded treewidth graphs are known, the dependency of the running time on the treewidth is much worse than $2^{\mathcal{O}(t)}$, see, e.g.,~\cite{chalermsook_et_al:LIPIcs.APPROX-RANDOM.2018.8}. If one, e.g., enumerates the minimum sizes of cuts compatible with the partitions of a single bag, the framework from \cite{cohen-addadApproximatingConnectivityDomination2016} results in an exponential-time algorithm. Another issue arises with property \ref{ubiquity:lift}. As edge contractions may decrease the vertex connectivity, $G/K$ might not even be a feasible instance to the $k$-VCSS. For $k\ge 1$, property \ref{ubiquity:lift} does not hold for the $(k+1)$-WECSS or $k$-CAP. To see this, consider an instance where there is a $k$-edge-connected spanning subgraph of cost $0$, but edges of cost $1$ need to be added to make it $(k+1)$-edge-connected.
Summing up, the techniques used in \cite{cohen-addadApproximatingConnectivityDomination2016} do not seem to be applicable to our setting, or, generally speaking, to graph minimization problems dealing with high connectivity requirements (unless in a multi-subgraph setting). \cite{fominApproximationSchemesWidth2019} extends the results in \cite{cohen-addadApproximatingConnectivityDomination2016} by considering general minor-free graph families and by relaxing condition \ref{ubiquity:connected_and_treewidth};  however, they still rely on properties \ref{ubiquity:FPT_PTAS} and \ref{ubiquity:lift}.

\end{document}